\documentclass[11pt,letterpaper]{article}

\usepackage{latexsym,multirow}
\usepackage{amssymb,amsmath, bm}
\usepackage{graphicx}
\usepackage[color,all,import,arrow]{xy}
\usepackage{enumerate}
\usepackage{mhchem}
\usepackage{booktabs}
\usepackage{natbib}
\usepackage[pdftex, bookmarksopen=true, bookmarksnumbered=true,
pdfstartview=FitH, breaklinks=true, urlbordercolor={0 1 0}, citebordercolor={0 0 1}]{hyperref}
\usepackage{colortbl}
\usepackage{subcaption}

\usepackage{dcolumn}
\newcolumntype{.}{D{.}{.}{-1}}
\newcolumntype{d}[1]{D{.}{.}{#1}}
\usepackage{theorem}
\theoremstyle{plain}
\theoremheaderfont{\scshape}
\newtheorem{assumption}{Assumption}

\newtheorem{corollary}{Corollary}

\newtheorem{proposition}{Proposition}
\newtheorem{theorem}{Theorem}

\newtheorem{lemma}{Lemma}
\newtheorem{estimator}{Estimator}

\newcommand{\ind}{\mbox{$\perp\!\!\!\perp$}}

\DeclareMathOperator{\sgn}{sgn}

\usepackage{rotating}

\usepackage[compact]{titlesec}

\allowdisplaybreaks

\newcommand\spacingset[1]{\renewcommand{\baselinestretch}%
{#1}\small\normalsize}

\newcommand*{\QEDB}{\hfill\ensuremath{\square}}

\newcommand{\bone}{\mathbf{1}}

\newcommand{\E}{\mathbb{E}}

\newcommand{\Pb}{\mathbb{P}}
\newcommand{\Pn}{\mathbb{P}_n}

\begin{document} 

\newcommand{\blind}{0}

\newcommand{\tit}{Policy evaluation and learning  with partially identified utility under truncation by death}


\spacingset{1.25}

\if0\blind
{\title{\bf\tit}
  \author{
   Zi Huang\thanks{School of Mathematics, Sun Yat-sen University, Guangzhou,  Guangdong 510275, China. Email: \href{mailto:huangz5@mail2.sysu.edu.cn}{huangz5@mail2.sysu.edu.cn}}
    \hspace{.5in} Zhichao Jiang\thanks{School of Mathematics, Sun Yat-sen University, Guangzhou,  Guangdong 510275, China. Email: \href{mailto:jiangzhch7@mail.sysu.edu.cn}{jiangzhch7@mail.sysu.edu.cn}}  
  }
\date{\today}
\maketitle
}\fi

\if1\blind
\title{\bf \tit}
\maketitle
\fi

\pdfbookmark[1]{Title Page}{Title Page}

\thispagestyle{empty}
\setcounter{page}{0}
         
\begin{abstract}
Policy evaluation and learning aim to assess and optimize treatment assignment rules based on individual characteristics. A fundamental challenge arises when outcomes are  truncated by death, rendering  conventional policy utilities undefined.
To address this challenge, we study policy evaluation and learning under truncation by death within the principal stratification framework. We propose the survivor average utility and subgroup survival rates for evaluating treatment policies. Under treatment ignorability and monotonicity, the subgroup survival rates are point identified, whereas the survivor average utility is only partially identified through sharp bounds. We then develop semiparametrically efficient estimators for the subgroup survival rates and hybrid estimators for the survivor average utility bounds.  
 Building on this evaluation framework, we formulate a constrained minimax optimization problem for policy learning that minimizes worst-case regret relative to a benchmark policy while requiring the learned policy to achieve a survival rate no lower than that of the benchmark. We show that the optimal policy admits a threshold representation based on two priority scores. We develop an estimation procedure for the optimal policy and establish regret and feasibility guarantees. Simulation studies and an application to the MIMIC-III clinical dataset demonstrate the practical performance of the proposed methods.
 
\noindent {\bf Keywords:} efficient influence function, individualized treatment rules, minimax regret, principal stratification, survivor average causal effect
\end{abstract}


\clearpage
\spacingset{1.5}
\section{Introduction}

Treatment policies assign treatments to individuals  based on  their observed characteristics. Two fundamental tasks in this setting are policy evaluation, which assesses the performance of a given treatment policy, and policy learning, which seeks to construct policies that optimize a target utility. These problems arise in many fields. For example, in clinical practice, investigators are often interested both in evaluating the consequences of candidate treatment strategies and in determining individualized treatment assignments that optimize patient outcomes  \citep[e.g.,][]{derubeis2014personalized,kosorok2015adaptive,kosorok2019precision}. These applications have motivated a rich  literature on policy evaluation and learning across a wide range of settings \citep[e.g.,][]{Qian2011,dudik2011doubly,Zhao2012,Zhang2012,swaminathan2015batch,Kitagawa2018,Athey2021}.

A common complication arises when the primary outcome is well defined only for a subgroup of individuals and treatment assignment affects membership in that subgroup. For example, in medical studies, measures such as disability status, cognitive function, or functional ability are undefined after death \citep[e.g.,][]{kurland2005directly,kurland2009longitudinal,murphy2011treatment}. In evaluations of job training programs, wages are undefined for individuals who fail to obtain employment \citep{zhang2003estimation,zhang2008evaluating,zhang2009likelihood,frumento2012evaluating}. This phenomenon, commonly referred to as truncation by death \citep{frangakis1999addressing}, fundamentally complicates both policy evaluation and policy learning. Because the primary outcome is undefined for some individuals, conventional policy utilities for the entire population cannot be formulated. Consequently, the targets of policy evaluation are no longer identifiable, and policy learning can no longer be formulated using conventional utility functions.

We study policy evaluation and learning in the presence of truncation by death. 
Adopting the principal stratification framework of \citet{fran:rubi:02}, we develop an approach for evaluating and optimizing treatment policies based on two complementary metrics (Section~\ref{sec:notation}). 
 The first is the {\it survivor average utility}, defined as the average potential outcome under a policy among individuals who would survive under both treatment conditions. The second consists of {\it subgroup survival rates}, which characterize survival under the policy across principal strata and provide information complementary to the survivor average utility. Together, these metrics extend causal estimands in the truncation-by-death literature from treatment effect estimation to policy evaluation and learning \citep[e.g.,][]{imai2008sharp,ding2011identifiability,tchetgen2014identification,wang2017identification,yang2018using}.
 
We first study policy evaluation in Section~\ref{sec:policyevaluate}. Under treatment ignorability and monotonicity, the subgroup survival rates are point identified, whereas the survivor average utility is only partially identified. We therefore adopt a partial identification approach and derive sharp bounds on the survivor average utility. These bounds extend partial identification results for the survivor average causal effect in randomized experiments \citep{imai2008sharp,jiang:16} to policy utilities. We then develop semiparametrically efficient estimators for the subgroup survival rates based on their efficient influence functions. For the utility bounds, we propose a hybrid estimation strategy that combines efficient influence function-based adjustment with plug-in estimation to accommodate both the smooth and non-smooth components of the bounds.

We then study policy learning under truncation by death in Section~\ref{sec:policylearningintro}. We formulate the problem as maximizing the survivor average utility subject to a minimum survival rate constraint. Because the survivor average utility is only partially identified, we adopt a minimax approach that seeks policies minimizing the worst-case regret relative to a prespecified benchmark policy while guaranteeing a survival rate at least as large as that of the benchmark. This framework extends existing minimax approaches for partially identified utility functions \citep[e.g.,][]{stoye2012minimax,Kallus2018,Kallus2021,Pu2021,Cui2021_partial,DAdamo2023,ben2024policy,han2024optimal} to settings with survival constraints. Unlike existing methods for policy learning under truncation by death \citep{chu2023multiply,park2025evaluating}, which rely on additional assumptions to point identify the survivor average utility, our framework accommodates partial identification directly while incorporating survival rate constraints. This yields an optimization problem with a different structure from existing formulations, resulting in a new characterization of the optimal policy.

We derive an explicit form of the optimal policy that minimizes the worst-case regret subject to the survival rate constraint. A key feature of our formulation is that, the worst-case regret is piecewise linear, rather than linear, in the policy. 
As a result, the optimal policy admits a threshold representation based on two priority scores, whereas standard constrained policy typically learning relies on a single priority score \citep[e.g.,][]{luedtke2016optimal,qiu2021optimal,sun2021treatment,qiu2022individualized}. The two scores jointly characterize an individual's contributions to the worst-case regret and the survival rate constraint.

Estimating the optimal policy presents additional challenges because a naive plug-in estimator requires stringent nuisance estimation rates to achieve the desired convergence rate. To overcome this difficulty, we augment the plug-in procedure with efficient influence function corrections, thereby weakening the required nuisance estimation rates while preserving the same regret guarantees. We establish regret bounds for the learned policy and show that it asymptotically satisfies the survival rate constraint.

We evaluate the finite-sample performance of the proposed methods through simulation studies in Section~\ref{sec::simulation} and illustrate their application using the MIMIC-III clinical dataset in Section~\ref{sec:application}. Section~\ref{sec:discussion} concludes.

We use the following notation throughout. Let $\|r\|_2 = \left\{ \int r(v)^2 \, d\mathbb{P}(v) \right\}^{1/2}$ denote the $L_2(\mathbb{P})$ norm, where $\mathbb{P}$ is the distribution of the observed data. Write $b_n = O_\mathbb{P}(a_n)$ if $b_n/a_n$ is bounded in probability and $b_n = o_\mathbb{P}(a_n)$ if $b_n/a_n$ converges to zero in probability. Let $\bone(\cdot)$ denote the indicator function and $\Pn$  the empirical mean.

\section{Notation and assumptions}
\label{sec:notation}
For unit $i=1,\dots,n$, let $Z_i\in\left\{0,1\right\}$ denote the binary treatment, $S_i\in\left\{0,1\right\}$ be the survival indicator, $Y_i$ the binary outcome, and $X_i\in\mathcal{X} $ a vector of pre-treatment covariates. We adopt the potential outcomes framework and invoke the stable unit treatment value assumption. Let $S_i(z)$ denote the potential survival status under treatment $z$, and let $Y_i(z)$ denote the potential outcome when it is well defined. The observed value of outcome and survival status satisfy $Y_i=Y_i(Z_i)$ and $S_i=S_i(Z_i)$. We assume $\left\{X_i,Z_i,S_i(0),S_i(1),Y_i(1),Y_i(0)\right\}$ are independent and identically distributed across units. For notational simplicity, we suppress the subscript $i$ whenever no confusion can arise.

We use principal stratification \citep{fran:rubi:02} to classify units according to their joint potential survival status, denoted by $U = (S(1), S(0))$. With a binary survival indicator, $U$ takes values in $\{(1,1), (1,0), (0,1), (0,0)\}$. For notational convenience, we write these strata as $\{11, 10, 01, 00\}$, corresponding to always-survivors, protected units, harmed units, and never-survivors, respectively.

We use $\pi:\mathcal{X}\to[0,1]$ to denote a treatment policy that  maps covariates $X$ to a treatment probability. Deterministic policies are a special case when $\pi(X)\in\{0,1\}$. 

We next introduce quantities of interest for policy evaluation and learning. Because outcomes are well defined under both treatment and control only for always-survivors, we focus on the expected outcome within this stratum,
\begin{eqnarray*}
V(\pi)&=&\mathbb{E}\left\{Y(\pi(X))\mid U=11\right\}.
\end{eqnarray*}
We refer to this quantity as the survivor average utility.

The survivor average utility alone does not provide a complete assessment of a policy under truncation by death. A policy may achieve a high survivor average utility while simultaneously causing deaths in other strata. We therefore also consider how the policy affects survival across principal strata. For always-survivors ($U=11$) and never-survivors ($U=00$), the survival rates are fixed at 1 and 0, respectively, regardless of the policy. Consequently, the policy affects survival only for protected units ($U=10$) and harmed units ($U=01$). We define the corresponding subgroup survival rates as 
\begin{eqnarray*}
B_{10}(\pi)&=&\mathbb{E}\left\{S(\pi(X))\mid U=10\right\},\\
B_{01}(\pi)&=&\mathbb{E}\left\{S(\pi(X))\mid U=01\right\}.
\end{eqnarray*}
Alternatively, one could also consider the overall survival rate, $\mathbb{E}\{S(\pi(X))\}$, which is determined by $B_{10}(\pi)$, $B_{01}(\pi)$, and the proportions of the principal strata.

We now introduce two key assumptions.
\begin{assumption}[Unconfoundedness]\label{assum::treatig} 
    $
    Z \ind \left\{ S(1), S(0), Y(1), Y(0) \right\}\mid X.
    $
\end{assumption}
Assumption~\ref{assum::treatig} rules out latent confounding between treatment assignment and both survival status and the outcome. It holds automatically in randomized experiments. In observational studies, the assumption is plausible only if the covariates $X$ capture all common causes of treatment assignment, survival, and the outcome.

The next assumption states that the treatment has a non-negative effect on survival status for each unit.
\begin{assumption}[Monotonicity]\label{assum::monoto}
 For each unit $i$,   $S_i(1) \geq S_i(0).$
\end{assumption}
Assumption~\ref{assum::monoto} rules out the harmed stratum $U=01$, consisting of units who would survive under control but not under treatment. It is plausible in applications where the treatment is designed to maintain or improve survival. It might be violated if the treatment carries nontrivial risks or induces adverse effects that reduce survival. Under Assumption~\ref{assum::monoto}, the only policy-dependent survival parameter is the protected survival rate $B_{10}(\pi)$. We  simplify it as $B(\pi)$.

Assumptions~\ref{assum::treatig}~and~\ref{assum::monoto} are commonly adopted in the principal stratification  literature
\citep[e.g.,][]{wang2017identification,jiang2022multiply,luo2023causal,chu2023multiply,li2025causal} and form the basis of our policy evaluation and learning approach. When these assumptions are questionable, researchers may conduct sensitivity analyses to assess the robustness of the conclusions.

Importantly, we do not impose principal ignorability or rely on auxiliary variables satisfying additional structural assumptions, both of which are commonly used to achieve point identification in principal stratification problems  \citep[e.g.,][]{ding2011identifiability, wang2017identification, ding:lu:17,jiang2021identification,jiang2022multiply,chu2023multiply,park2025evaluating}.

\section{Policy evaluation}\label{sec:policyevaluate}

\subsection{Nonparametric identification}
We study the identification and estimation of the protected survival rate and the survivor average utility under a given policy $\pi$.

To simplify notation, we define
\begin{eqnarray*}
    e(X)&=&\mathbb{P}(Z=1\mid X),\quad p_z(X)\ =\ \mathbb{P}(S=1\mid Z=z,X),\quad \mu_z(X)\ =\ \mathbb{E}\left(Y\mid S=1,Z=z,X\right)
\end{eqnarray*}
for $z=0,1$, where  $e(X)$ is the propensity score, $p_z(X)$ is the observed conditional survival rate, and $\mu_z(X)$ is the conditional outcome mean among the observed survivors.
We further define
$ p_z=\mathbb E\{p_z(X)\}$,
which corresponds to the marginal survival probability under treatment $z$.

Under Assumptions~\ref{assum::treatig}~and~\ref{assum::monoto}, the proportion of each principal stratum given covariates can be identified by
 \begin{eqnarray*}
    \mathbb{P}(U=11\mid X)&=&p_0(X),\quad \mathbb{P}(U=00\mid X)\ =\ 1-p_1(X),\quad \mathbb{P}(U=10\mid X)\ =\ p_1(X)-p_0(X).
\end{eqnarray*}
This implies that the protected survival rate is point identified,  as shown in part (a) of the next theorem. In contrast, the survivor average utility is not point identified.
Under Assumptions~\ref{assum::treatig} and \ref{assum::monoto}, the observed survivors in the treatment group $(Z = 1, S = 1) $ consist of a mixture of always-survivors ($U=11$) and protected units ($U=10$), and these strata cannot be distinguished using the observed data. As a result, the distribution of potential outcomes among always-survivors cannot be recovered from the observed data alone. Part (b) of the theorem therefore derives the sharp bounds on the survivor average utility.

\begin{theorem}
\label{theo::constraint}
Suppose that Assumptions~\ref{assum::treatig}~and~\ref{assum::monoto} hold.
\begin{enumerate}[(a)]
    \item The protected survival rate $B(\pi) = \mathbb{E}\{S(\pi(X))\mid U=10\}$ is identified by
\begin{eqnarray*}
        B(\pi) &=& \mathbb{E}\left\{\frac{p_1(X) - p_0(X)}{p_1 - p_0} \pi(X)\right\}.
    \end{eqnarray*}
    \item The sharp bounds on the conditional mean $m_{11}(X)=\mathbb{E}\left[Y(1)\mid U=11,X\right]$ are $[L_{11}(X),U_{11}(X)]$, where
     \begin{eqnarray*}
    L_{11}(X)&=&\max\left\{\frac{p_1(X) \mu_1(X) - p_1(X) + p_0(X)}{p_0(X)},0\right\},\\
U_{11}(X)&=&\min \left\{\frac{\mu_1(X) p_1(X)}{p_0(X)},1\right\}.
\end{eqnarray*}
These sharp bounds imply that the survivor average utility $V(\pi)=\mathbb{E}\{Y(\pi(X))\mid U=11\}$ are sharply bounded by $[L(\pi),U(\pi)]$, where
\begin{eqnarray}
L(\pi) &=& \frac{1}{p_0} \mathbb{E}\left\{\pi(X) L_{11}(X)p_0(X) + (1 - \pi(X)) \mu_0(X) p_0(X)\right\}, \label{eqn::utility_lower_a}\\
        U(\pi) &=& \frac{1}{p_0} \mathbb{E}\left\{\pi(X) U_{11}(X)p_0(X) + (1 - \pi(X)) \mu_0(X) p_0(X)\right\}.
        \label{eqn::utility_upper_a}
    \end{eqnarray}
\end{enumerate}
\end{theorem}

Theorem~\ref{theo::constraint}(a) identifies the protected survival rate in terms of $p_z(X)$. One can also derive alternative identification formulas based on the propensity score. Because these derivations and the corresponding efficient estimation closely parallel those in the standard average treatment effect literature, we omit them for brevity. The remainder of this section focuses on the survivor average utility.

Theorem~\ref{theo::constraint}(b) expresses the sharp bounds on the survivor average utility in terms of $p_z(X)$ and $\mu_z(X)$. As with the protected survival rate, the bounds admit multiple equivalent identification formulas, which motivates the search for efficient estimation strategies. However, because the bounds involve non-smooth maximum and minimum operators, efficient estimation cannot be achieved through standard influence function based techniques alone. We address this challenge in the next subsection.

\subsection{Efficient estimation}
\label{sec:effpolicyevl}
We focus on the lower bound $L(\pi)$; results for the upper bound follow similarly and are provided in the Supplementary Material.

We can rewrite the lower bound in~\eqref{eqn::utility_lower_a} as
\begin{eqnarray*}
    L(\pi)&=&\frac{\Pb\left\{c(X)\gamma_1(X)+\gamma_2(X)\right\}}{p_0},
\end{eqnarray*}
where  
\begin{eqnarray*}
    c(X)&=&\bone\left\{p_1(X)\mu_1(X)-p_1(X)+p_0(X)\geq 0\right\},\\\gamma_1(X)&=&\pi(X)\left\{p_1(X)\mu_1(X)-p_1(X)+p_0(X)\right\},\\
    \gamma_2(X)&=&\{1-\pi(X)\}\mu_0(X)p_0(X).
\end{eqnarray*}
For ease of exposition, let $N(\pi)=\Pb\left\{c(X)\gamma_1(X)+\gamma_2(X)\right\}$ denote the numerator of $L(\pi)$, and let  $D=p_0$ denote the denominator. The denominator $D$ can be estimated efficiently using standard doubly robust techniques. The numerator $N(\pi)$, however, involves an indicator function $c(X)$ and is therefore non-smooth.  
As a result, $L(\pi)$  is not pathwise differentiable and therefore does not admit a regular influence function representation \citep{bickel1993efficient}.

To motivate our estimation strategy, we examine the error decomposition of a general plug-in estimator
$\mathbb{P}_n\left\{\hat{c}(X)\hat{\gamma}_1(X)+\hat{\gamma}_2(X)\right\}$ for the numerator $\Pb\left\{c(X)\gamma_1(X)+\gamma_2(X)\right\}$, where  $\hat{c}(X),\hat{\gamma}_z(X)$  denote the estimators of $c(X)$ and $\gamma_z(X)$, respectively.
The estimation error can be decomposed as
\begin{eqnarray*}
   \mathbb{P}_n\left\{\hat{c}(X)\hat{\gamma}_1(X)+\hat{\gamma}_2(X)\right\}-\Pb\left\{c(X)\gamma_1(X)+\gamma_2(X)\right\}&=& T_1+T_2+T_3,
\end{eqnarray*}
where 
\begin{eqnarray*}
 T_1&=&(\mathbb{P}_n-\Pb)\left\{c(X)\gamma_1(X)+\gamma_2(X)\right\},\\
 T_2&=&(\mathbb{P}_n-\Pb)\left\{\hat{c}(X)\hat{\gamma}_1(X)+\hat{\gamma}_2(X)-c(X)\gamma_1(X)-\gamma_2(X)\right\},\\
 T_3&=&\Pb\left[\hat{c}(X)\left\{\hat{\gamma}_1(X)-\gamma_1(X)\right\}\right]+\Pb\left\{\hat{\gamma}_2(X)-\gamma_2(X)\right\}+\Pb\left[\left\{\hat{c}(X)-c(X)\right\}\gamma_1(X)\right].
\end{eqnarray*}
The first term $T_1$ is an empirical average of a fixed function and therefore converges at root-$n$ rate. The second term $T_2$ is an empirical process term  and typically has a smaller order under standard regularity conditions. Thus,  the key challenge lies in controlling $T_3$, which arises from estimation of  the  smooth components $\gamma_1(X)$ and $\gamma_2(X)$, and the non-smooth indicator $c(X)$.

We employ different estimation strategies for the smooth and non-smooth components. For  $\gamma_1(X)$ and $\gamma_2(X)$, we construct estimators based on the efficient influence functions (EIFs) of $\E\{\gamma_1(X)\}$ and $\E\{\gamma_2(X)\}$. For $c(X)$, we use a plug-in approach that estimates the argument of the indicator function. This strategy parallels the recent approaches in estimating non-smooth functionals in other settings \citep[e.g.,][]{ben2024policy,levis2024intervention,levis2025covariate}.

We now  derive the EIFs for $\E\{\gamma_1(X)\}$ and $\E\{\gamma_2(X)\}$. To facilitate the derivations, define
\begin{eqnarray}
    \psi_{p_z\mu_z}&=&p_z(X)\mu_z(X)+\frac{\bone\left(Z=z\right)\left\{YS-p_z(X)\mu_z(X)\right\}}{\mathbb{P}(Z=z\mid X)},\label{eqn:notationpmu}\\
    \psi_{p_z}&=&p_z(X)+\frac{\bone\left\{Z=z\right\}\left\{S-p_z(X)\right\}}{\mathbb{P}(Z=z\mid X)}.\label{eqn:notationp}
\end{eqnarray}
Under Assumption~\ref{assum::treatig}, we have
$\mathbb{E}(\psi_{p_z\mu_z}\mid X) = p_z(X)\mu_z(X)$ and  $\mathbb{E}(\psi_{p_z}\mid X) = p_z(X).$
Moreover, the centered quantities \(\psi_{p_z\mu_z} - \mathbb{E}\{p_z(X)\mu_z(X)\}\) and \(\psi_{p_z} - p_z\) are  the EIFs of \(\mathbb{E}\{p_z(X)\mu_z(X)\}\) and \(p_z\), respectively.
From the chain rule, the EIF of $\mathbb{E}\left\{\gamma_1(X)\right\}$ is $\psi_{\gamma_1}-\mathbb{E}\left\{\gamma_1(X)\right\}$, where
\begin{eqnarray*}
\psi_{\gamma_1}&=&\pi(X)\left\{\psi_{p_1\mu_1}-\psi_{p_1}+\psi_{p_0}\right\}.
\end{eqnarray*}
Similarly, the EIF of $\mathbb{E}\left\{\gamma_2(X)\right\}$ is $\psi_{\gamma_2}-\mathbb{E}\left\{\gamma_2(X)\right\}$, where
\begin{eqnarray*}
\psi_{\gamma_2}&=&\left\{1-\pi(X)\right\}\psi_{p_0\mu_0}.
\end{eqnarray*}

Building on the preceding results, we propose the following ratio estimator for $L(\pi)$.
\begin{estimator}\label{eg:ratio} We propose a three-step estimator for the lower bound $L(\pi)$.
\begin{itemize}
    \item Step 1: Estimate the propensity score $\hat{e}(X)$, the observed survival rate $\hat{p}_z(X)$, and the outcome mean $\hat{\mu}_z(X)$ for $z=0,1$.
    \item Step 2: Construct $\hat{\psi}_{p_z\mu_z}$ and $\hat{\psi}_{p_z}$ for $z=0,1$ by plugging  $\hat{e}(X)$, $\hat{p}_z(X)$, and $\hat{\mu}_z(X)$ into~\eqref{eqn:notationpmu}~and~\eqref{eqn:notationp}, respectively.
    \item Step 3: Compute the estimators $\hat{N}(\pi)$ and $\hat{D}$ for the numerator and denominator using 
    \begin{eqnarray*}
    \hat{N}(\pi)&=&\mathbb{P}_n\left[\pi(X)\bone\left\{\hat{p}_1(X)\hat{\mu}_1(X)-\hat{p}_1(X)+\hat{p}_0(X)\geq 0\right\}\left\{\hat{\psi}_{p_1\mu_1}-\hat{\psi}_{p_1}+\hat{\psi}_{p_0}\right\}\right]\\
    &&+\mathbb{P}_n\left[\left\{1-\pi(X)\right\}\left(\hat{\psi}_{p_0\mu_0}\right)\right],\\
    \hat{D}&=&\mathbb{P}_n\hat{\psi}_{p_0}.
\end{eqnarray*}
The final estimator is given by: $\hat{L}(\pi)=\hat{N}(\pi)/\hat{D}$.
\end{itemize}
\end{estimator}
Although one may, in principle, substitute an EIF-based estimator inside the indicator, doing so may increase bias, as noted in \citet{DAdamo2023}.

We next establish the asymptotic properties of the estimator $\hat{L}({\pi})$. The main challenge arises from estimating the non-smooth indicator $c(X)$ near its boundary, where the argument of the indicator lies close to zero. To control the estimation error, we impose a margin condition that restricts the probability mass in a neighborhood of this threshold.

\begin{assumption}\label{a:margin_posclass1}
  There exist  $\alpha > 0$ and a constant $C$ such that for any $t \geq 0$,
  $$\mathbb{P}\left(\left|p_1(X)\mu_1(X)-p_1(X)+p_0(X)\right| \leq t\right) \leq Ct^\alpha.$$
\end{assumption}
Assumption~\ref{a:margin_posclass1} restricts the concentration of observations near the boundary of the indicator function. It rules out excessive probability mass in regions where the argument of the indicator approaches zero. The parameter $\alpha$ controls the strength of this assumption: larger values of $\alpha $ correspond to lighter concentration near the boundary.
This type of assumption has been proposed in the classification literature \citep[e.g.,][]{tsybakov2004optimal,Audibert2007} and has since been widely used in various semiparametric contexts involving non-smooth functionals \citep[e.g.,][]{luedtke2016optimal,Luedtke2016,kennedy2020sharp,qiu2021optimal,DAdamo2023,levis2025covariate}.

For ease of exposition, we introduce the following notation to characterize the estimation errors of the nuisance functions:
\begin{eqnarray*}
\left\|\hat{\lambda}_1 - \lambda_1\right\|_2^2 &=& \E\left[\left\{\frac{Z}{\hat{e}(X)}-\frac{Z}{e(X)}\right\}^2\right], \quad \left\|\hat{\lambda}_0 - \lambda_0\right\|_2^2 \ = \ \E\left[\left\{\frac{1-Z}{1-\hat{e}(X)}-\frac{1-Z}{1-e(X)}\right\}^2\right],\\
\|\hat{m} - m\|_\infty&=&\sup\left\{\|\hat{p}_1\hat{\mu}_1 - p_1\mu_1\|_\infty, \|\hat{p}_1 - p_1\|_\infty,\|\hat{p}_0 - p_0\|_\infty\right\}.
\end{eqnarray*}
The following theorem provides the asymptotic bias formulas of   $\hat{N}(\pi)$ and $\hat{D}$.

\begin{theorem}\label{thm:convergence_margin}
Suppose that Assumptions~\ref{assum::treatig}~to~\ref{a:margin_posclass1} hold, and the following conditions are satisfied for the nuisance functions: 
\begin{enumerate}[(a)] 
\item The  nuisance estimators are constructed using an independent sample.
\item $\|\hat{e}-e\|_2=o_\mathbb{P}(1)$, $\|\hat{p}_z-p_z\|_2=o_\mathbb{P}(1)$, and $\|\hat{\mu}_z-\mu_z\|_2=o_\mathbb{P}(1)$  for $z=0,1$.
\item $\delta<\left\{e(X),\hat{e}(X)\right\}<1-\delta$ for some $\delta\in(0,1)$ and all $X$.
\end{enumerate}
We have
\begin{eqnarray}
\nonumber \hat{N}(\pi)-N(\pi)
 &=&(\Pn - \Pb)\left\{c(X)\psi_{\gamma_1}+\psi_{\gamma_2}\right\} \\
 \nonumber &&+ O_\Pb\left(\left\|\hat{\lambda}_1-\lambda_1\right\|_2\left(\left\|\hat{p}_1\hat{\mu}_1-p_1\mu_1\right\|_2+\left\|\hat{p}_1-p_1\right\|_2\right)\right)\\
\nonumber &&+O_{\mathbb{P}}\left(\left\|\hat{\lambda}_0-\lambda_0\right\|_2\left(\left\|\hat{p}_0\hat{\mu}_0-p_0\mu_0\right\|_2+\left\|\hat{p}_0-p_0\right\|_2\right)\right)  \\
\label{eqn::biasN}   && +O_\Pb\left(\|\hat{m}-m\|_{\infty}^{1+\alpha} \right) + o_\Pb(n^{-1/2}),\\
\label{eqn::biasD} \hat{D}-D & =& (\Pn - \Pb)\left(\psi_{p_0}\right) + O_\Pb\left(\left\|\hat{\lambda}_0-\lambda_0\right\|_2\left\|\hat{p}_0-p_0\right\|_2\right)+o_\Pb(n^{-1/2}).
\end{eqnarray}
\end{theorem}
Theorem~\ref{thm:convergence_margin} expresses the asymptotic errors of $\hat{N}(\pi)$ and $\hat{D}$
 in terms of the estimation errors of the nuisance functions. Condition (a) allows us to avoid Donsker conditions \citep{chernozhukov2018double,kennedy2024semiparametric}. In practice, when only a single random sample is available, this condition can be implemented through sample splitting and cross-fitting. Condition (b) requires consistency of the nuisance estimators and is standard in the semiparametric literature.
 Condition (c) is analogous to the classic overlap condition \citep{rosenbaum1983central,d2021overlap}, ensuring that both the true and estimated propensity scores remain bounded away from zero and one.

The first term on the right hand side of~\eqref{eqn::biasN} is an empirical average of a fixed function of the true nuisances and therefore converges  at a root-$n$ rate. The second and third terms arise from the EIF-based estimation of $\gamma_1(X)$ and $\gamma_2(X)$. These two terms have the product structure of nuisance estimation errors that also appears in doubly robust estimation of the average treatment effect. The fourth term captures the error induced by plugging the estimated nuisance functions into the indicator $c(X)$. Its order depends on the exponent $\alpha$ in Assumption~\ref{a:margin_posclass1}.
By comparison, the asymptotic error of $\hat{D}$ in~\eqref{eqn::biasD} follows directly from the theory of doubly robust estimation of average treatment effect.

Theorem~\ref{thm:convergence_margin} implies the asymptotic normality of  $\hat{N}(\pi)$ and $\hat{D}$ under the following rate conditions 
on the nuisance estimators.
\begin{assumption}\label{assum:rates}
The estimated nuisance functions satisfy the following conditions:
\begin{enumerate}
\item [(a)] $\|\hat{\lambda}_z-\lambda_z\|_2\|\hat{p}_z-p_z\|_2 = o_\Pb(n^{-1/2})$ for $z=0,1$.
\item [(b)] $\|\hat{\lambda}_z-\lambda_z\|_2\|\hat{p}_z\hat{\mu}_z-p_z\mu_z\|_2= o_\Pb(n^{-1/2})$ for $z=0,1$.
\item [(c)] $\|\hat{m}-m\|_{\infty}^{1+\alpha} = o_\Pb(n^{-1/2})$.
\end{enumerate}
\end{assumption}
Assumption~\ref{assum:rates} requires sufficiently fast convergence of the nuisance estimators. Conditions (a)~and~(b) require that the products of  estimation errors from the propensity score, survival probabilities, and outcome means converge faster than the root-$n$ rate. Such rate requirements are well understood for many commonly used flexible estimation procedures. Condition (c) concerns the non-smooth indicator $c(X)$.  This condition requires  $\hat{m}$ to converge sufficiently quickly to control the bias arising from estimation of the non-smooth indicator. The required rate becomes more stringent as $\alpha$ decreases.

The next theorem establishes the asymptotic normality of  $\hat{L}(\pi)$.
\begin{theorem}\label{thm:convergence_eif}
Suppose that Assumptions~\ref{assum::treatig}~to~\ref{assum:rates} and Conditions (a), (b), and (c) in Theorem~\ref{thm:convergence_margin} hold.
 Then,
 \begin{eqnarray*}
    \sqrt{n}\{\hat{L}(\pi)-L(\pi)\}&\xrightarrow{\textup{d}}&\mathcal{N}\left(0,\mathbb{E}\left\{\frac{c(X)\psi_{\gamma_1}+\psi_{\gamma_2}-L(\pi)\psi_{p_0}}{p_0}\right\}^2 \right).
\end{eqnarray*}
\end{theorem}
We present an analogous result for the asymptotic normality of the upper bound in the Supplementary Material. Together, 
we can construct the following asymptotically valid $100(1-\zeta)\%$ confidence interval for the bounds $[L(\pi),U(\pi)]$ of survivor average utility with 
\begin{eqnarray}\label{eqn:cibounds}
    \left(\hat{L}(\pi)-z_{1-\zeta/2}\sqrt{\hat{\Omega}_L/n},\hat{U}(\pi)+z_{1-\zeta/2}\sqrt{\hat{\Omega}_U/n}\right),
\end{eqnarray}
where $\hat{\Omega}_L$, $\hat{\Omega}_U$ are consistent estimators of the asymptotic variances of $\hat{L}(\pi)$ and $\hat{U}(\pi)$, respectively, and $z_{1-\zeta/2}$ is the $1-\zeta/2$-th quantile of the standard normal distribution. 

\section{Policy learning under a survival rate constraint}\label{sec:policylearningintro}
\subsection{Formulation of optimization problem}
We now turn to policy learning under truncation by death. The notion of an optimal policy differs across principal strata. For always-survivors, the outcome is well defined under both treatment conditions, and a desirable policy maximizes the survivor average utility $V(\pi)$. For protected units, the outcome is well defined only under treatment, so treatment decisions for this group cannot be guided by outcome comparisons. Since these units survive only under treatment, a desirable policy should assign treatment. For never-survivors, the outcome is not well defined under either treatment or control, and treatment has no effect on survival.

These considerations suggest a tradeoff between improving outcomes among always-survivors and preserving survival among protected units. We formalize this tradeoff through a constrained optimization problem. Specifically, we maximize the survivor average utility $V(\pi)$ subject to a minimum protected survival rate $B(\pi)$.

However, Theorem~\ref{theo::constraint} shows that the survivor average utility $V(\pi)$ is not point identified. 
The source of this non-identification is the conditional mean outcome among always-survivors under treatment, $m_{11}(X) = \mathbb{E}\{Y(1) \mid U=11, X\}$, which is sharply bounded by $[L_{11}(X),U_{11}(X)]$ as established in Theorem~\ref{theo::constraint}.
We therefore adopt a minimax regret approach that seeks a policy that minimizes the worst-case regret relative to a benchmark policy $\varpi$. We formulate policy learning  as the following constrained minimax optimization problem:
\begin{eqnarray}
\label{eqn:optimiproblem}
    \pi^*\in \arg\min R_{\sup}(\pi,\varpi) \text{ subject to }B(\pi)\geq B(\varpi),
\end{eqnarray}
where \begin{eqnarray*}
    R_{\sup}(\pi,\varpi)=\max_{m_{11}(X)\in[L_{11}(X),U_{11}(X)]}\{V(\varpi)-V(\pi)\}.
\end{eqnarray*}
The quantity $ R_{\sup}(\pi,\varpi)$ is the worst-case regret of policy $\pi$ relative to  $\varpi$. It equals the largest difference in survivor average utility  between $\varpi$ and $\pi$ over all possible values of $m_{11}(X)$. Because the benchmark policy $\varpi$ always satisfies the constraint, it is feasible by construction. Therefore, the optimal policy $\pi^*$ is guaranteed to perform at least as well as $\varpi$.

The benchmark policy $\varpi$ must be specified in advance. Different choices of $\varpi$ can lead to different optimal policies \citep{Cui2021_partial}. Common choices include the always-treat policy $\varpi(x) = 1$ for all $x$ and the never-treat policy $\varpi(x) = 0$ for all $x$. \citet{ben2024policy} also propose using an oracle policy, which represents the optimal policy if the unidentified components were known.

Theorem~\ref{theo::constraint} implies the following identification formulas for the constraint threshold $B(\varpi)$ and the worst-case regret $ R_{\sup}(\pi,\varpi)$.
\begin{corollary}
\label{cor:formsup}
Under Assumptions~\ref{assum::treatig}~and~\ref{assum::monoto}, we have
\begin{eqnarray}
\label{eqn:orgB} B(\varpi)&=&\mathbb{E}\left\{\frac{p_1(X)-p_0(X)}{p_1-p_0}\varpi(X)\right\},\\
\nonumber R_{\sup}(\pi,\varpi) &=&\frac{1}{p_0}\mathbb{E}\left[\left\{\varpi(X)-\pi(X)\right\}\bone\left\{\pi(X)<\varpi(X)\right\}A_U
        (X)\right]\\
 \label{eqn:orgrsup}        &&+\frac{1}{p_0}\mathbb{E}\left[\left\{\varpi(X)-\pi(X)\right\}\bone\left\{\pi(X)\geq \varpi(X)\right\}A_L
        (X)\right],
\end{eqnarray}
where \begin{eqnarray*}
        A_L(X)&=&\max\left\{p_1(X) \mu_1(X) - p_1(X) + p_0(X),0\right\}-p_0(X)\mu_0(X),\\
        A_U(X)&=&\min \left\{\mu_1(X) p_1(X),p_0(X)\right\}-p_0(X)\mu_0(X).
    \end{eqnarray*}
\end{corollary}

\subsection{Optimal policy based on priority scores}
\label{sec:policyoptimalform}
From Corollary~\ref{cor:formsup}, we can write the optimization problem in~\eqref{eqn:optimiproblem} as
\begin{eqnarray}
\label{eq:pi_opt_orig}
\pi^* &\in& \arg\max \mathbb{E}\left[\left\{\pi(X)-\varpi(X)\right\}\left\{\bone(\pi(X)<\varpi(X))A_U
        (X)+\bone(\pi(X)\geq \varpi(X)) A_L(X)\right\}\right]\\&&\nonumber\text{ subject to } \mathbb{E}\left[\pi(X)\left\{p_1(X)-p_0(X)\right\}\right] \geq \mathbb{E}\left[\varpi(X)\left\{p_1(X)-p_0(X)\right\}\right].
\end{eqnarray}
The objective function in~\eqref{eq:pi_opt_orig} is piecewise linear in $\pi(x)$, with a kink at $\pi(x)=\varpi(x)$, and the constraint is linear in $\pi(x)$.
To understand the structure of the optimization problem, consider an infinitesimal modification of the policy from $\pi(x)$ to $\pi(x)+\Delta \pi(x)$ at a covariate value $x$. 
The resulting marginal gain depends on whether the current policy lies below or above the benchmark policy $\varpi(x)$: it equals $A_U(x)\Delta\pi(x)$ when $\pi(x)<\varpi(x)$ and $A_L(x)\Delta\pi(x)$ when $\pi(x)\ge \varpi(x)$. The same modification also produces a marginal cost in the constraint equal to $\{p_1(x)-p_0(x)\}\Delta \pi(x)$. This motivates characterizing policy modifications through the marginal gain-to-cost ratios, given by
\begin{eqnarray}
\label{eq:ratio}
\rho_L(x) &=& \frac{A_L(x)}{p_1
(x)-p_0(x)},\quad \rho_U(x)\ =\  \frac{A_U(x)}{p_1
(x)-p_0(x)},
\end{eqnarray}
for $x$ with $p_1(x)-p_0(x)>0$.
We refer to them as priority scores, following related work on policy learning with constraints \citep[e.g.,][]{sun2021treatment,levis2024intervention}. Specifically, $\rho_L(x)$ characterizes an upward modification from $\varpi(x)$ toward 1. Such a modification increases the survival cost and improves the objective when $\rho_L(x)>0$. In contrast, $\rho_U(x)$ characterizes a downward modification from $\varpi(x)$ toward 0. Such a modification reduces the survival cost and improves the objective when $\rho_U(x)<0$. In the absence of the constraint, the optimal policy modifies the benchmark policy
only when doing so improves the objective. That is, it modifies $\varpi(x)$ upward to 1
for units with positive $\rho_L(x)$, modifies $\varpi(x)$ downward to $0$ for units with negative $\rho_U(x)$, and leaves it unchanged otherwise.
 
Under the constraint in \eqref{eq:pi_opt_orig}, however, this unconstrained solution may no longer be feasible. If the unconstrained solution already satisfies the constraint, then no further adjustment is needed. Otherwise, the policy must be further adjusted to satisfy the constraint, either by reversing some downward modifications from 0 back toward $\varpi(x)$ or further modifying some previously unchanged units from $\varpi(x)$ toward 1. These additional adjustments generally reduce the objective relative to the unconstrained solution. The policy  therefore proceeds  in decreasing order of priority scores, so that each additional unit of cost incurs the smallest possible loss in the objective. The adjustments stop when the required constraint level is satisfied, yielding a threshold-based rule.

To formalize this stopping rule, we introduce a threshold parameter $\eta\leq 0$ on the priority scores. 
For a given threshold $\eta$, units with $\rho_L(X)>\eta$ are assigned to 1, units with $\rho_L(X)\leq \eta<\rho_U(X)$ are assigned to the benchmark level $\varpi(X)$, and the remaining units are assigned to 0. Relative to the unconstrained solution, decreasing $\eta$ from 0 performs two types of adjustments. First, units with $\eta<\rho_L(X)\leq 0$ are additionally moved from $\varpi(X)$ to 1. Second, units with $\eta<\rho_U(X)<0$ have their downward modifications reversed from 0 back to $\varpi(X)$. 
Whether a threshold $\eta$ is feasible depends on the survival contribution generated by the resulting policy. Define
\begin{eqnarray*}
Q(\eta)
&=& \mathbb{E}\!\left[ \left\{\varpi(X)\,\mathbf{1}\{\rho_L(X)\leq \eta<\rho_U(X)\} + \mathbf{1}\{\rho_L(X)>\eta\}\right\}\{p_1(X)-p_0(X)\} \right],\\
\eta^E &=& \sup\{\eta : Q(\eta)\ge \mathbb{E}[\varpi(X)\{p_1(X)-p_0(X)\}]\}, \quad \eta^S\ =\ \min(0,\eta^E).
\end{eqnarray*}
The quantity $Q(\eta)$ is the left-hand side of the survival constraint evaluated under the threshold rule indexed by $\eta$, with $Q(0)$ corresponding to the unconstrained solution.  
The threshold $\eta^E$ is the largest value of $\eta$ for which the corresponding policy satisfies the survival constraint. 
The unconstrained solution corresponds to the threshold $0$. Therefore, $\eta^S=0$ when the unconstrained solution is feasible and $\eta^S=\eta^E<0$ otherwise. 
The next theorem provides the explicit form of the optimal policy.

\begin{theorem}\label{theo:pi_opt}
Suppose that Assumptions~\ref{assum::treatig} and~\ref{assum::monoto} hold, and $p_1(x)-p_0(x)>0$ for all  $x\in\mathcal{X}$.
Define
\begin{eqnarray*}
  R(\eta)&=&\mathbb{E}\left[ \left\{\varpi(X) \bone(\rho_U(X) = \eta) + (1-\varpi(X))\bone(\rho_L(X) = \eta)\right\} \left\{p_1(X)-p_0(X)\right\}\right].
\end{eqnarray*}
 Then the optimal policy $\pi^*(X)$ is given by
 
\noindent {\bf Case 1:} If $\eta^S = 0$, then
\begin{eqnarray*}
\pi^*(X) &=&  \varpi(X)\bone\{\rho_U(X)>0>\rho_L(X)\}+\bone\{\rho_L(X)>0\}.
\end{eqnarray*}
\noindent {\bf Case 2(a):}  If $\eta^S < 0$ and $R(\eta^S)=0$, then
\begin{eqnarray*}
\pi^*(X) &=&  \varpi(X)\bone\{\rho_U(X)>\eta^S>\rho_L(X)\}+\bone\{\rho_L(X)>\eta^S\}.
\end{eqnarray*}
\noindent {\bf Case 2(b):}  If $\eta^S < 0$ and $R(\eta^S)>0$, then
\begin{eqnarray*}
\pi^*(X) &=& \begin{cases}
\alpha \cdot \varpi(X), & \text{if }\  \rho_U(X) = \eta^S,\\
\varpi(X) + \alpha \cdot \{1-\varpi(X)\}, & \text{if }\  \rho_L(X) = \eta^S, \\
 \varpi(X)\bone\{\rho_U(X)>\eta^S>\rho_L(X)\}+\bone\{\rho_L(X)>\eta^S\}, & \text{otherwise},
\end{cases}
\end{eqnarray*}
where 
\begin{eqnarray*}
    \alpha &=& \frac{\mathbb{E}\left[\varpi(X)\left\{p_1(X)-p_0(X)\right\}\right] - Q(\eta^S)}{R(\eta^S)}.
\end{eqnarray*}
\end{theorem}
The condition on $p_1(x)-p_0(x)$  ensures that the  two priority scores $\rho_U(x)$ and $\rho_L(x)$ are finite.
Under Assumptions~\ref{assum::treatig} and~\ref{assum::monoto},  this condition implies a non-zero proportion of protected units across all covariate values. Similar positivity conditions are common in the literature on constrained optimal treatment regimes \citep[e.g.,][]{qiu2021optimal,sun2021treatment,qiu2022individualized}. 
We emphasize, however,  that the assumption is invoked solely to streamline exposition. The optimal policy can still be derived without it, and the corresponding result is provided in the Supplementary Material.

In Case 1, the unconstrained solution already satisfies the required constraint level and is therefore optimal. We refer to this case as the non-binding case.

In Case 2, the unconstrained solution fails to satisfy the required constraint level and is therefore not optimal. We refer to this case as the binding case. In this case, the policy must adjust assignments among units with priority scores below 0 until the threshold $\eta^S$ is reached. The optimal policy first applies the threshold rule determined by $\eta^S$, assigning treatment whenever $\rho_L(X)>\eta^S$ and reversing downward modifications whenever $\rho_U(X)>\eta^S$. Whether additional modifications are required depends on the behavior of the priority scores at the threshold.
The quantity $R(\eta^S)$ measures the contribution of units whose priority score is exactly equal to $\eta^S$. 
 If $R(\eta^S)=0$, no units lie exactly at the threshold, so modifying all units with priority scores above $\eta^S$ suffices to satisfy the constraint with equality. This corresponds to Case 2(a).
If $R(\eta^S)>0$, then marginal units do exist. In this case, satisfying the constraint may require partial modification of the units at the threshold. The optimal policy therefore randomizes among these marginal units with probability $\alpha$, chosen so that the survival constraint holds with equality. This corresponds to Case 2(b).

The three cases in Theorem~\ref{theo:pi_opt} can be combined into the
following unified expression, which will be useful in subsequent developments.
\begin{eqnarray}
\label{eqn:combinedpi} \pi^*(X) &=&
\begin{cases}
\alpha \cdot \varpi(X), & \text{if } \eta^S<0, R(\eta^S)>0, \ \text{and}\  \rho_U(X) = \eta^S,\\
\varpi(X) + \alpha \cdot \{1-\varpi(X)\}, & \text{if } \eta^S<0, R(\eta^S)>0, \ \text{and}\   \rho_L(X) = \eta^S, \\
 \varpi(X)\bone\{\rho_U(X)>\eta^S>\rho_L(X)\}& \\
 \hspace{0.3cm}+\bone\{\rho_L(X)>\eta^S\}, & \text{otherwise}.
\end{cases}
\end{eqnarray}

In the Supplementary Material, we also derive the oracle policy, which assumes that the partially identified quantities are known, and establish the regret of \(\pi^*\) relative to the oracle policy.
\subsection{Estimation of optimal policy}
\label{sec:identification}
We estimate the optimal policy based on Theorem~\ref{theo:pi_opt}. We consider a plug-in estimation approach, which first estimates the priority scores $\rho_L(X)$ and $\rho_U(X)$, together with the threshold $\eta^S$, and then substitutes them into the analytic form of the optimal policy in \eqref{eqn:combinedpi}. The priority scores are obtained by plugging nuisance estimators into their defining formulas. The estimation of $\eta^S$, however, requires particular care. As Theorem~\ref{theo:pi_opt} shows, even a small estimation error in $\eta^S$ can lead to a discontinuous change in the estimated policy, particularly when the estimated and true thresholds lie on opposite sides of zero.

The threshold $\eta^S$ is determined by the constraint level 
$\mathbb{E}[\varpi(X)\{p_1(X)-p_0(X)\}]$ and the threshold functions $Q(\eta)$ and $R(\eta)$. A naive approach estimates $\eta^S$ by first constructing estimators of these components, and then solving the empirical version of the defining formula for $\eta^S$. However, such a plug-in estimator generally cannot achieve a root-$n$ convergence rate, as this would require the nuisance estimators themselves to converge at the root-$n$ rate, which is often unrealistic in practice.
 To address these issues, we propose two improvements for estimating $\eta^S$.
First, we construct an EIF-based estimator of 
the constraint level to  improve estimation accuracy.
Second, similar to Section~\ref{sec:effpolicyevl}, we adopt a hybrid estimation strategy for $Q(\eta)$ and $R(\eta)$, since they involve non-smooth indicator functions. This leads to the following estimators.
\begin{estimator}
\label{est:eif}
We propose a three-step plug-in estimator for the optimal policy $\pi^*(X)$.
\begin{itemize}
\item Step 1: Estimate the priority scores.
\begin{enumerate}[(a)]
\item Estimate the propensity score $\hat{e}(X)$, the observed survival probability $\hat{p}_z(X)$, and the outcome mean $\hat{\mu}_z(X)$ for $z=0,1$.
\item Compute the priority scores $\hat{\rho}_{L}(X)$ and $\hat{\rho}_{U}(X)$ by plugging $\hat{e}(X)$, $\hat{p}_z(X)$, and $\hat{\mu}_z(X)$ into~\eqref{eq:ratio}.
\end{enumerate}
\item Step 2: Estimate the marginal threshold.
\begin{enumerate}[(a)]
\item Estimate the constraint level $\mathbb{E}\left[\varpi(X)\left\{p_1(X)-p_0(X)\right\}\right]$ using the EIF-based estimator
\begin{eqnarray*}
\hat{C}_{\textnormal{eif}}&=&\frac{1}{n}\sum_{i=1}^n \varpi(X_i)\left\{\hat{\psi}_{p_1}(X_i)-\hat{\psi}_{p_0}(X_i)\right\}.
\end{eqnarray*}
    \item Construct the estimators of $Q(\eta)$ and $R(\eta)$. For any candidate threshold $\eta$, define the plug-in estimators for $Q(\eta)$ and $R(\eta)$ as 
    \begin{eqnarray*}
\hat{Q}(\eta) =\Pn  \hat{Q}(\cdot;\eta) , \quad \hat{R}(\eta)=\Pn \hat{R}(\cdot;\eta),
\end{eqnarray*}
where the pointwise estimating functions are given by
    \begin{eqnarray*}
    \hat{Q}(x;\eta) &=& \left[\varpi(x)\bone\{\hat{\rho}_{U}(x) > \eta\}+\{1-\varpi(x)\}\bone\{\hat{\rho}_{L}(x) > \eta\}  \right]\left\{\hat{\psi}_{p_1}(x)-\hat{\psi}_{p_0}(x)\right\},\\
    \hat{R}(x;\eta) &=&  \left[\varpi(x)\bone\{\hat{\rho}_{U}(x) = \eta\}+\{1-\varpi(x)\}\bone\{\hat{\rho}_{L}(x) = \eta\}  \right]\left\{\hat{\psi}_{p_1}(x)-\hat{\psi}_{p_0}(x)\right\}.
    \end{eqnarray*}
 Then, the  threshold estimators are given by $\hat{\eta}^E= \sup \{\eta : \hat{Q}(\eta) \geq \hat{C}_{\textup{eif}}\}$ and $\hat{\eta}^S = \min\{\hat{\eta}^E, 0\}$.
\end{enumerate}
\item Step 3: Compute the estimated optimal policy $\hat{\pi}^*(X) $ by plugging $\hat{\rho}_{L}(X)$, $\hat{\rho}_{U}(X)$, and $\hat{\eta}^S$ into~\eqref{eqn:combinedpi}.
\end{itemize}
\end{estimator}
Step 1(a) estimates the nuisance models using an independent sample, consistent with the evaluation procedure described in Section~\ref{sec:effpolicyevl}. Steps 2(a) and 2(b) implement the EIF correction and hybrid estimation strategy discussed above.

\subsection{Statistical guarantees for the estimated optimal policy}\label{sec:optimalbound}

To establish theoretical guarantees, we first introduce the following regularity condition.

\begin{assumption}\label{assum:c311}
\begin{enumerate}[(a)]
    \item The function $ Q(\eta) $  is differentiable with a nonzero derivative in a neighborhood of $\eta ^E$. 
   \item If $\eta^E<0$, then $\hat{\eta}^E$ solves $\hat{Q}(\hat{\eta}^E)\geq \hat{C}_{\textup{eif}}$ up to $o_{\mathbb{P}}(n^{-1/2})$ error, i.e., $\hat{Q}(\hat{\eta}^E)-\hat{C}_{\textup{eif}}=o_{\mathbb{P}}(n^{-1/2})$.
\end{enumerate}

\end{assumption}

Assumption \ref{assum:c311}(a) requires \(Q(\eta)\) to vary smoothly around the true threshold. This condition holds when the priority scores have positive density at the threshold and the relevant conditional expectations are continuous at the threshold. It is a standard regularity condition for the asymptotic distribution  in $Z$-estimation problems \citep{van2000asymptotic}.

Assumption~\ref{assum:c311}(b) requires the estimated threshold to solve the empirical constraint equation up to an error that is asymptotically negligible relative to the root-$n$ scale.  It places restrictions on both the numerical accuracy of the threshold search and the local behavior of the empirical threshold function near the cutoff. A similar condition has also been used in the analysis of threshold-based policies in a different context \citep{levis2024intervention}.

Next, we introduce a margin condition that controls the behavior of the true and estimated priority scores near their boundaries $\eta^S$ and $\hat{\eta}^S$, respectively.
\begin{assumption}\label{assum:margincondition1}
\begin{enumerate}[(a)]
   \item There exist $\beta > 0$ and $C$ such that for any $\varepsilon \geq 0$, 
   \begin{eqnarray*}
    \mathbb{P}\left(|\rho_{L}(X)-\eta^S|\leq \varepsilon\right)\leq C \varepsilon^\beta,\quad \mathbb{P}\left(|\rho_{U}(X)-\eta^S|\leq \varepsilon\right)\leq C \varepsilon^\beta.
    \end{eqnarray*}
    \item The boundary events for the estimated quantities satisfy
     \begin{eqnarray*}
 \mathbb{P}\left\{\hat{\rho}_L(X)=\hat{\eta}^S\right\}=O_{\mathbb{P}}(n^{-1/2}),  \quad \mathbb{P}\left\{\hat{\rho}_U(X)=\hat{\eta}^S\right\}=O_\mathbb{P}(n^{-1/2}).
\end{eqnarray*}
\end{enumerate}
\end{assumption}
Assumption~\ref{assum:margincondition1}(a) is a margin condition that limits
the probability mass of the priority scores $\rho_L(X)$ and $\rho_U(X)$
in a shrinking neighborhood of $\eta^S$.
This condition serves  two purposes. First, it controls the estimation error induced by plug-in estimation of the priority scores and the threshold. Second, it rules out the boundary events $\rho_L(X)=\eta^S$ and $\rho_U(X)=\eta^S$. Specifically, taking $\varepsilon=0$ in Assumption~\ref{assum:margincondition1}(a) implies $\Pb(\rho_{U}(X)=\eta^S)=\Pb(\rho_{L}(X)=\eta^S)=0$.
Consequently, the optimal policy $\pi^*(X)$ can be represented almost surely as
\begin{eqnarray*}
\varpi(X)\bone\{\rho_U(X)>\eta^S\}
+
(1-\varpi(X))\bone\{\rho_L(X)>\eta^S\}.
\end{eqnarray*}

Assumption~\ref{assum:margincondition1}(b) requires the boundary events
$\hat{\rho}_L(X)=\hat{\eta}^S$ and $\hat{\rho}_U(X)=\hat{\eta}^S$
to occur with probability converging to zero at root-$n$ rate.
Similar technical conditions have been adopted in other constrained policy learning settings \citep[e.g.,][]{qiu2021optimal,qiu2022individualized}.
As a consequence, with probability tending to one, the estimated
optimal policy reduces to
\begin{eqnarray*}
\varpi(X)\mathbf{1}\{\hat{\rho}_U(X)>\hat{\eta}^S\}
+
(1-\varpi(X))\mathbf{1}\{\hat{\rho}_L(X)>\hat{\eta}^S\},
\end{eqnarray*}
thereby avoiding explicit treatment of the boundary cases.

Given the assumptions above, we first establish the consistency of the estimated threshold $\hat{\eta}^S$.
\begin{proposition}\label{lem:consistent}
    Suppose that Assumptions~\ref{assum::treatig},~\ref{assum::monoto},~\ref{assum:rates}(a), ~\ref{assum:c311},~\ref{assum:margincondition1}(a), and Conditions (a), (b), and (c) in Theorem~\ref{thm:convergence_margin} hold, and $p_1(x)-p_0(x)>0$ for all  $x\in\mathcal{X}$. Then, $ \hat{\eta}^S\xrightarrow[]{p}\eta^S$.
\end{proposition}
Proposition~\ref{lem:consistent} implies that, when $\eta^S$ is negative, the estimated threshold $\hat{\eta}^S$ is also negative with probability tending to one.
This property is crucial for establishing the asymptotic behavior of the estimated optimal policy.

We next establish a guarantee regarding the feasibility of the estimated policy.
\begin{theorem}\label{thm:constraintviolation}
Suppose that Assumptions~\ref{assum::treatig},~\ref{assum::monoto},~\ref{assum:rates}(a),~\ref{assum:c311},~\ref{assum:margincondition1}, and Conditions (a), (b), and (c) in Theorem~\ref{thm:convergence_margin} hold, and $p_1(x)-p_0(x)>0$ for all  $x\in\mathcal{X}$. Then,
    \begin{eqnarray*}
        B(\hat{\pi}^*)-B(\varpi)&=&\begin{cases}
            \frac{Q(0)-\mathbb{E}\left[\varpi(X)\{p_1(X)-p_0(X)\}\right]}{p_1-p_0}+o_{\mathbb{P}}(1), & \eta^S=0,\\
            o_{\mathbb{P}}(1), & \eta^S<0.
        \end{cases}
    \end{eqnarray*}
\end{theorem}
Theorem~\ref{thm:constraintviolation} characterizes the difference between the constraint levels of the estimated optimal policy and that of benchmark policy. We refer to this quantity as the constraint slackness. Positive values indicate that the estimated policy satisfies the constraint with slack, whereas negative values correspond to constraint violation.
In the non-binding case, the optimal policy coincides with the unconstrained solution. By definition, $Q(0)$ exceeds the constraint level $\mathbb{E}[\varpi(X)\{p_1(X)-p_0(X)\}]$. Consequently, Theorem~\ref{thm:constraintviolation} shows that the constraint slackness converges to a strictly positive constant. The constraint slackness decreases as the constraint level increases. In the binding case, the unconstrained solution is infeasible. Theorem~\ref{thm:constraintviolation} states that the constraint slackness is $o_{\mathbb{P}}(1)$, implying that $\hat{\pi}^*$ is asymptotically feasible.

 Finally, we derive a rate bound on the excess worst-case regret  of the estimated policy $\hat{\pi}^*$ relative to the true optimal policy $\pi^*$.
\begin{theorem}\label{thm:regretbound}
Suppose that Assumptions~\ref{assum::treatig},~\ref{assum::monoto},~\ref{assum:rates}(a),~\ref{assum:c311},~\ref{assum:margincondition1}, and Conditions (a), (b), and (c) in Theorem~\ref{thm:convergence_margin} hold, and $p_1(x)-p_0(x)>0$ for all  $x\in\mathcal{X}$. Then, 
\begin{eqnarray*}
    R_{\sup}(\hat{\pi}^*,\varpi)-R_{\sup}(\pi^*,\varpi)&=&O_{\mathbb{P}}(n^{-1/2}+R_{1}+R_{2}),
\end{eqnarray*}
where \begin{eqnarray*}
    R_{1}&=&\frac{1}{p_0}(\|\hat{\rho}_{L}-\rho_{L}\|_{\infty}+|\hat{\eta}^S-\eta^S | )^{1+\beta},\\
    R_{2}&=&\frac{1}{p_0}(\|\hat{\rho}_{U}-\rho_{U}\|_{\infty}+|\hat{\eta}^S-\eta^S|)^{1+\beta}.
\end{eqnarray*}
\end{theorem}
The terms $R_{1}$ and $R_{2}$  arise from estimator error in the priority scores $\rho_L(X)$ and $\rho_U(X)$, as well as in the threshold $\eta^S$. The remaining contribution is of order $n^{-1/2}$, which results from the discrepancy between the simplified policy and the original optimal policy $\hat{\pi}^*$.

\section{Simulation}
\label{sec::simulation}
We conduct simulations to evaluate the finite-sample performance of the proposed methods for policy evaluation and learning.

We generate covariates $X\in\mathbb{R}^3$  independently as $X_1\sim U(-1,1)$, $X_2\sim \textup{Bernoulli}(0.5)$ and $X_3\sim N(0,1)$.
To satisfy Assumption~\ref{assum::monoto}, we assign each unit to one of three principal strata: protected units $(U=10)$, always-survivors $(U=11)$, or never-survivors $(U=00)$ according to a ordered logistic model with probabilities
\begin{eqnarray*}
\Pb(U=10\mid X) &=&\frac{1}{1+e^{-(2+X_1-X_2+0.5X_3)}} -\frac{1}{1+e^{-(X_1-X_2+0.5X_3)}},\\
\Pb(U=11\mid X) &=& \frac{1}{1+e^{-(X_1-X_2+0.5X_3)}},\\
\Pb(U=00\mid X) &=& 1-\Pb(U=10\mid X)-\Pb(U=11\mid X).
\end{eqnarray*}
Potential survival statuses are  determined by stratum membership.
For units in the always-survivor stratum, we generate both $Y(1)$ and $Y(0)$ according to $Y(z) \sim\mathrm{Bernoulli}\{\mu_z(X)\}$ where
\begin{eqnarray*}
\mu_z(X) &=&\frac{1}{1+e^{-\{(2z+1)X_1+0.2X_2-0.5X_3\}}}.
\end{eqnarray*}
For protected units, we generate only $Y(1)\sim\mathrm{Bernoulli}\{\mu_1(X)\}$. For never-survivors, neither potential outcome is defined.

We generate treatment assignments according to $Z\mid X\sim \textup{Bernoulli}\{e(X)\}$,
where 
\begin{eqnarray*}
e(X)&=&\frac{1}{1+e^{-(1.5X_1-0.5X_2+0.5X_3)}}.
\end{eqnarray*}
The observed survival indicator is  $S = ZS(1)+(1-Z)S(0)$. The observed outcome is  $Y=Y(Z)$ when $S=1$, and is undefined otherwise.

We consider three fixed policies, $\pi(x)= c$ with $c\in\{0.3,0.5,0.7\}$, which assign treatment with constant probability $c$. For each policy, we consider sample size $n\in\{500,1000, 5000\}$. We estimate all nuisance functions using both logistic regression and generalized additive models (GAMs) with 5-fold cross-fitting.

Figure~\ref{fig:bound-estimation-error} presents the mean absolute errors (MAEs) and root-mean-square errors (RMSEs) of the estimated lower and upper bounds under different sample sizes based on $500$ replications. For both bounds, the MAEs and RMSEs decrease steadily as the sample size increases across all three policies. Logistic regression and GAM yield very similar performance, although logistic regression tends to achieve slightly smaller errors in the smaller sample settings. 
\begin{figure}[htbp]
    \centering

    \begin{subfigure}{0.8\textwidth}
        \centering
        \includegraphics[width=\textwidth]{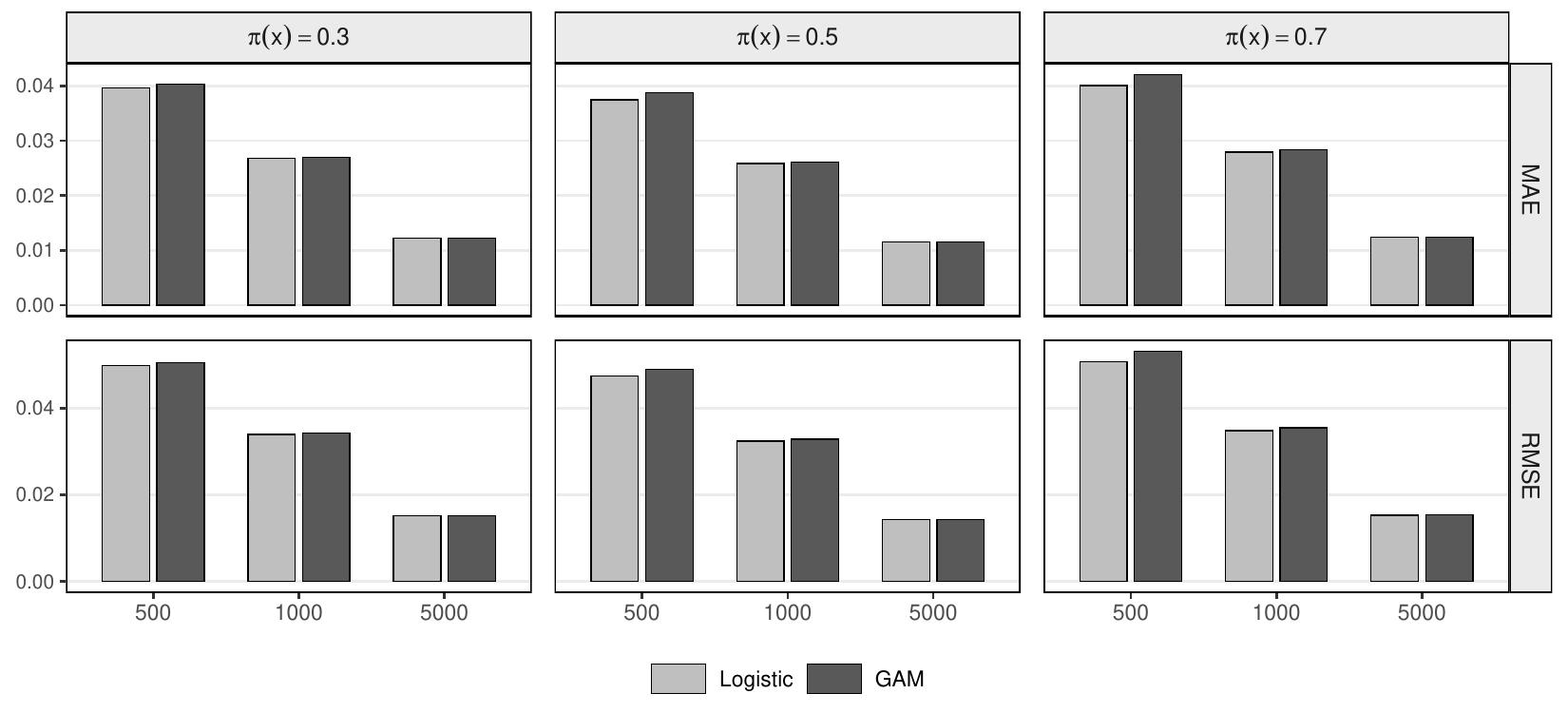}
        \caption{Lower bound.}
        \label{fig:error-lower-bound}
    \end{subfigure}

    \vspace{1em}

    \begin{subfigure}{0.8\textwidth}
        \centering
        \includegraphics[width=\textwidth]{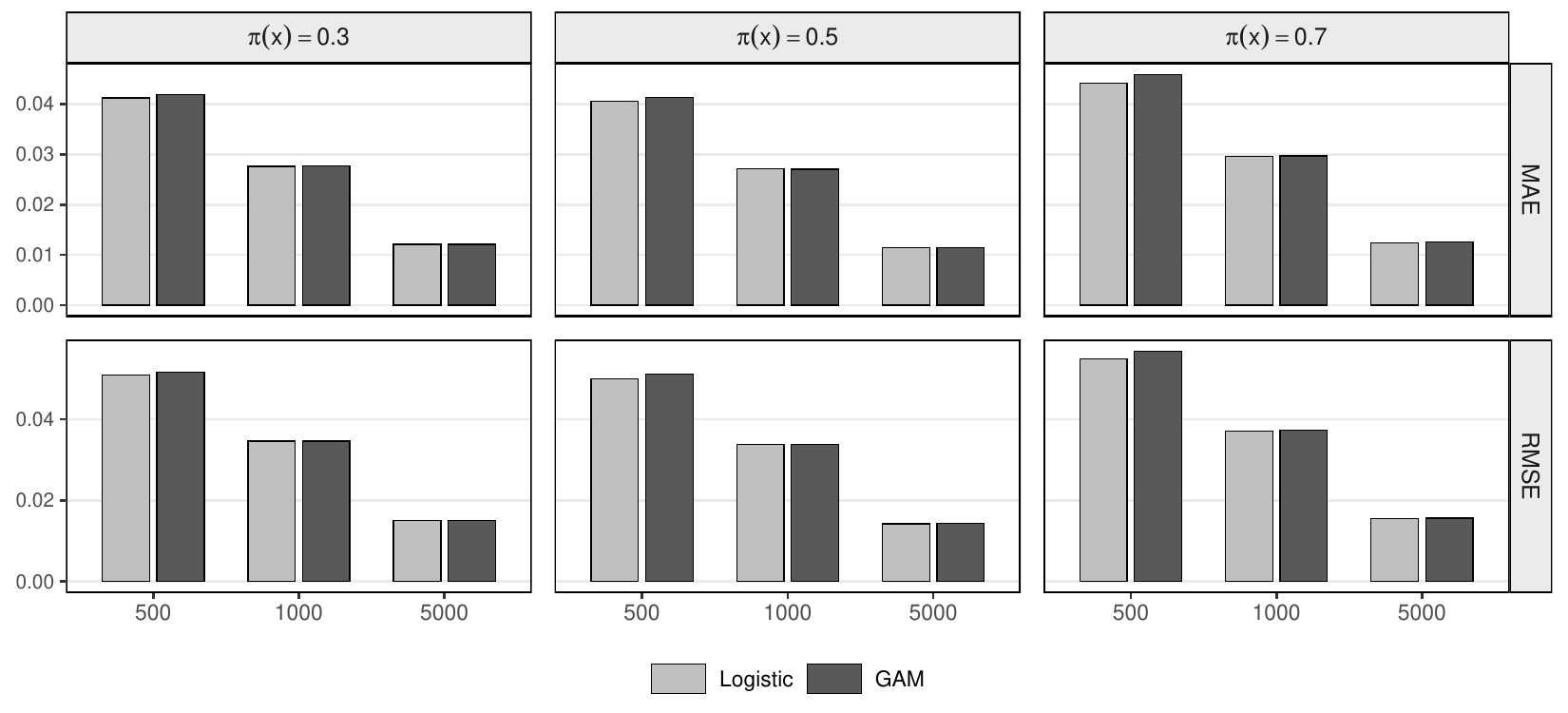}
        \caption{Upper bound.}
        \label{fig:error-upper-bound}
    \end{subfigure}

    \caption{
        MAEs and RMSEs of the estimated lower and upper bounds under different sample sizes and fixed policies. In each panel, light gray bars correspond to logistic regression nuisance estimation and dark gray bars correspond to GAM nuisance estimation.
    }
    \label{fig:bound-estimation-error}
\end{figure}

We also examine coverage rates of the confidence intervals for the bounds on the survivor average utility, constructed according to~\eqref{eqn:cibounds}. Table~\ref{tab:cp_ci} shows the result. Across all settings, the empirical coverage rates remain close to the nominal 95\% level.
\begin{table}[htbp]
\centering
\caption{Empirical coverage rates of the confidence interval for the bounds on the survivor average utility $[L(\pi),U(\pi)]$ under different nuisance estimation methods, sample sizes, and fixed policies.}
\label{tab:cp_ci}
\setlength{\tabcolsep}{6pt}
\begin{tabular}{llccc}
\hline
Nuisance estimator & Sample size 
& \(\pi(x)=0.3\) & \(\pi(x)=0.5\) & \(\pi(x)=0.7\) \\
\hline
\multirow{3}{*}{Logistic}
& \(n=500\)  & 0.958 & 0.958 & 0.952 \\
& \(n=1000\) & 0.952 & 0.954 & 0.958 \\
& \(n=5000\) & 0.954 & 0.960 & 0.966 \\
\hline
\multirow{3}{*}{GAM}
& \(n=500\)  & 0.954 & 0.964 & 0.954 \\
& \(n=1000\) & 0.956 & 0.964 & 0.960 \\
& \(n=5000\) & 0.956 & 0.964 & 0.962 \\
\hline
\end{tabular}
\end{table}

We next evaluate the finite-sample performance of the proposed estimator for the
optimal policy. We consider benchmark policies $\varpi(x)=c$ that assign treatment with constant probability $c$. We approximate true threshold $\eta^S$ by solving the constraint equation using a Monte Carlo sample of size 10,000. For $c \leq 0.58$, $\eta^S=0$, indicating that the unconstrained optimal policy already satisfies the constraint. As $c$ increases beyond 0.58, $\eta^S$ becomes negative, indicating that the constraint becomes binding.

Similar to the policy evaluation setting, we estimate the nuisance functions using both logistic regression and generalized additive models with 5-fold cross-fitting. For each estimated policy, we evaluate three performance metrics: (i) the constraint slackness, (ii) the misclassification rate relative to the true optimal policy, and (iii) the excess worst-case regret. To approximate these metrics, we generate an independent validation dataset of size $10,000$ and apply both the estimated and true optimal policies to this dataset.

We first fix the treatment assignment probability of benchmark policy at 
$c = 0.3$, $0.5$, and $0.7$, and vary the sample size over $
n \in \{500,1000,5000\}$. Figure~\ref{fig:sample_var} presents the metrics of the estimated policy. For the non-binding cases ($c=0.3,0.5$), the slackness metric converges to a positive value. For the binding case ($c=0.7$), the slackness approaches zero as the sample size increases, consistent with the asymptotic feasibility result in Theorem~\ref{thm:constraintviolation}. The misclassification rate and excess worst-case regret both decrease steadily with sample size, reaching relatively small values by $n=5000$. Logistic regression generally achieves lower misclassification rates and excess regret than GAM, particularly in the smaller sample settings.

We then fix the sample size at $n=1000$ and vary the treatment assignment probability $c$ of benchmark policies from 0 to 1. Figure~\ref{fig:varpi_vary} reports the results. The slackness metric is positive for smaller values of $c$, indicating that the estimated policy is the unconstrained solution, which already satisfies the constraint. As $c$ increases, the slackness metric decreases toward zero, reflecting that the constraint becomes increasingly binding. Once $c$ exceeds the boundary between the non-binding and binding regimes, the estimated policy remains close to the feasibility boundary, with only a small amount of slack remaining.  

The misclassification rate remains relatively stable for most values of $c$. At $c=0$ and $c=1$, however, it drops substantially because the optimal policy depends on only one priority score in these two extreme cases, eliminating one source of estimation error. The excess worst-case regret shows a different pattern, decreasing gradually as $c$ increases. At $c=1$, the estimated regret becomes slightly negative. This occurs because the estimated policy exhibits a small constraint violation, as reflected by the slightly negative slackness. 
By relaxing the constraint marginally, the estimated policy attains utility levels that are unavailable to the strictly feasible optimal policy. Across all values of $c$, logistic regression consistently outperforms GAM in terms of both misclassification rate and excess worst-case regret.


\begin{figure}[htbp]
     \centering
     \includegraphics[width=0.8\linewidth]{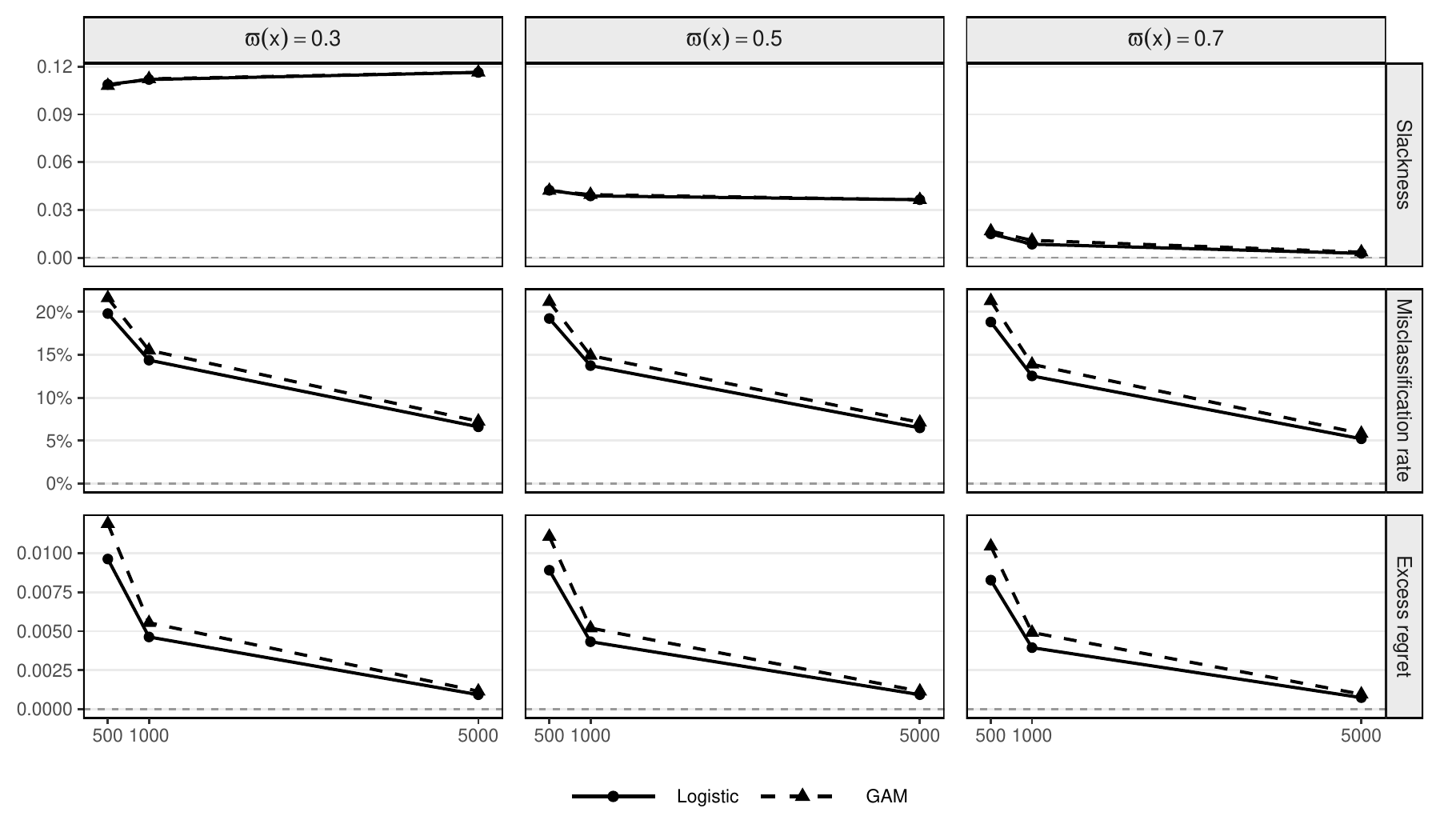}
     \caption{Performance metrics of the estimated optimal policy under different sample sizes and benchmark policies. Solid lines correspond to logistic-regression nuisance estimation and dashed lines correspond to GAM nuisance estimation. The true worst-case regret $R_{\sup}(\pi^*, \varpi)$ equals $0.047$, $0.036$, and $0.025$ under the three benchmark policies, respectively.}
     \label{fig:sample_var}
 \end{figure}
 
\begin{figure}[htbp]
    \centering
    \includegraphics[width=0.7\linewidth]{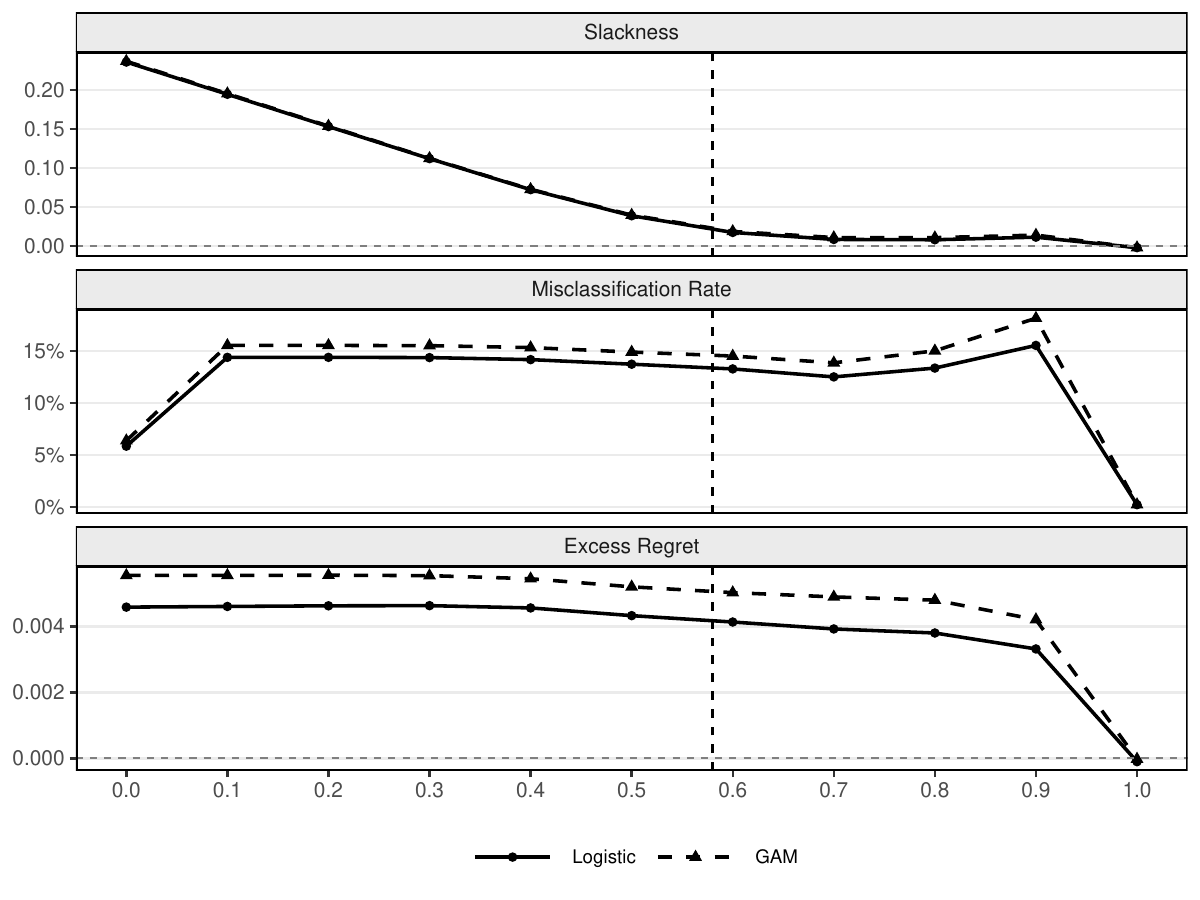}
    \caption{Performance metrics of the estimated optimal policy under benchmark policies $\varpi(x)= c$, with $c$ varying from 0 to 1 and the sample size fixed at $n=1000$. The dashed line marks the  approximate boundary between the non-binding and binding regimes at $c=0.58$.}
    \label{fig:varpi_vary}
\end{figure}

\section{Application}
\label{sec:application}

We apply the proposed method to the Medical Information Mart for Intensive Care III (MIMIC-III) v1.4 database, which contains de-identified health records from more than 40,000 patients admitted to critical care units at Beth Israel Deaconess Medical Center between 2001 and 2012. This database includes comprehensive patient information, including vital signs, laboratory results, medications, and survival outcomes. \citet{johnson2016mimic} provides a detailed description of the database. 

We study treatment policies for mechanical ventilation among severely ill patients with sepsis. The primary outcome is the Sequential Organ Failure Assessment (SOFA) score at 48 hours after the estimated onset of sepsis, which quantifies the severity of organ dysfunction. Lower SOFA scores indicate better clinical outcomes. Truncation by death arises because this outcome is only well defined for patients who survive to 48 hours and have a recorded  SOFA measurement at that time.

We focus on the sepsis cohort previously established by \citet{komorowski2018artificial} and \citet{killian2020empirical}. We restrict the study sample to those with baseline SOFA scores greater than 10 to focus on severely ill patients, where the baseline period is defined as the first 24 hours after the estimated onset of sepsis. The final sample consists of 1429 subjects. Among them, 861 received mechanical ventilation ($Z = 1$), while the remaining did not ($Z = 0$). We define the binary outcome $Y = 1$ if the 48‑hour SOFA score is unchanged or improved relative to baseline, and $Y = 0$ otherwise. The survival status $S$ indicates whether a patient survived the first 48 hours after ICU admission: $S = 1$ if the patient was alive with a recorded 48‑hour SOFA score, and $S = 0$ otherwise. We include the following baseline covariates: age (years); the baseline SOFA score; the Glasgow Coma Scale, where lower scores indicate more severe impairment of consciousness; creatinine (mg/dL), a marker of kidney function; and the \ce{PaO2}/\ce{FiO2} ratio, a measure of hypoxemia severity for which lower values indicate worse gas exchange.

We  estimate the nuisance functions using GAMs. We consider benchmark policies with constant treatment probability $\varpi(x)=c$ and the observed treatment policy  $\varpi(x)=e(x)$. To evaluate the estimated optimal policy, we examine both its treatment assignment distribution and its estimated worst-case regret. The latter is computed by plugging the estimated nuisance functions into \eqref{eqn:orgrsup} in Corollary~\ref{cor:formsup}. Table~\ref{tab:policy_performance} summarizes the results.

\begin{table}[htbp]
\centering
\caption{Estimated optimal policy under different benchmark policies $\varpi(x)$. The estimated optimal policy columns show counts of patients assigned $0$, $\varpi$, and $1$, respectively,     with boundary cases grouped into the corresponding adjacent categories. The $\hat{R}_{\sup}(\hat{\pi}^*,\varpi)$ column shows the estimated worst-case regret.}
\label{tab:policy_performance}
\begin{tabular}{ccccc|ccccc}
\toprule
$\varpi(x)$ & \multicolumn{3}{c}{Estimated optimal policy} & $\hat{R}_{\sup}(\hat{\pi}^*,\varpi)$
& $\varpi(x)$ & \multicolumn{3}{c}{Estimated optimal policy} & $\hat{R}_{\sup}(\hat{\pi}^*,\varpi)$ \\
\cline{2-4} \cline{7-9}
& 0 & unchanged & 1 & & & 0 & unchanged & 1 & \\
\midrule
0.0 & 1377 &  & 52 & -0.027
& 0.6 & 37 & 1334 & 58 & -0.026 \\

0.1 & 37 & 1340 & 52 & -0.027
& 0.7 & 37 & 1307 & 85 & -0.026 \\

0.2 & 37 & 1340 & 52 & -0.027
& 0.8 & 37 & 1282 & 110 & -0.025 \\

0.3 & 37 & 1340 & 52 & -0.027
& 0.9 & 37 & 1200 & 192 & -0.025 \\

0.4 & 37 & 1340 & 52 & -0.026
& 1.0 & 28 &  & 1401 & -0.024 \\

0.5 & 37 & 1338 & 54 & -0.026
& $\hat{e}(x)$  & 37 & 1335 & 57 & -0.025 \\
\bottomrule
\end{tabular}
\end{table}

 With $\varpi(x)=c$, the estimated policy coincides with the unconstrained solution for $c \leq 0.4$. In this case, the estimated policy modifies the benchmark policy for only a small subset of patients. This pattern reflects the conservative nature of the minimax-regret criterion.

As $c$ increases beyond $0.4$, the constraint becomes binding. The policy must increase the protected survival rate, even at some loss in utility. Correspondingly, the number of patients  assigned to treatment increases steadily, whereas the number assigned to control remains stable. This reflects that the additional survival is achieved primarily by expanding deterministic treatment to more patients rather than reducing the number assigned deterministic control. With $\varpi(x)=\hat{e}(x)$, the constraint is also binding. 

Across all considered benchmarks, the estimated worst-case regret $\hat{R}_{\sup}(\hat{\pi}^*, \varpi)$ remains negative, indicating that the estimated policy consistently improves upon the benchmark policy. The magnitude of the estimated regret is relatively stable across benchmark policies. When $\varpi(x)=1$, the estimated regret is slightly negative. This occurs because the estimated policy exhibits a small violation of the survival rate constraint, thereby producing gains that are unavailable to the feasible optimal policy.

\section{Discussion}
\label{sec:discussion}

We discuss several possible extensions of our framework. 

First, although we have focused on binary outcomes, many decision problems involve continuous outcomes. Extending our framework to this setting is therefore of substantial practical interest. The identification and estimation of utility bounds in the continuous outcome setting would likely connect to Lee bounds and their generalizations, which are commonly used to bound causal effects in the presence of outcome truncation and selection bias \citep{lee2009training,semenova2025generalized}. 

Second, although our framework avoids principal ignorability and other auxiliary assumptions that are often difficult to justify in practice, it still relies on monotonicity. This assumption may be violated in some applications.  Future work could develop sensitivity analysis to assess the robustness of policy evaluation and learning results to departures from monotonicity.

Third, our framework can be extended to accommodate more general utility functions. For example, 
treatment costs could be incorporated  into the utility definition, allowing policies to balance outcome improvement against treatment burden. More broadly, utilities may depend on the joint potential outcomes rather than a single observed outcome, thereby capturing preferences over different treatment-response profiles \citep{ben2024policy}.
\newpage
\pdfbookmark[1]{References}{References}
\spacingset{1.45}
\bibliographystyle{Chicago}
\bibliography{policytruncation-ref}

\newpage

\appendix

\setcounter{equation}{0}
\setcounter{figure}{0}
\setcounter{theorem}{0}
\setcounter{lemma}{0}
\setcounter{section}{0}
\setcounter{assumption}{0}
\setcounter{estimator}{0}
\renewcommand {\theequation} {S\arabic{equation}}
\renewcommand {\thefigure} {S\arabic{figure}}
\renewcommand {\thetheorem} {S\arabic{theorem}}
\renewcommand {\thelemma} {S\arabic{lemma}}
\renewcommand {\thesection} {S\arabic{section}}
\renewcommand {\theassumption} {S\arabic{assumption}}
\renewcommand {\theestimator} {S\arabic{estimator}}

\begin{center}
  \LARGE {\bf Supplementary Material}
\end{center}
\par Section~\ref{sec:addres} provides additional results. Section~\ref{sec:estimationupper} presents the efficient estimation of upper bound of survivor average utility. Section~\ref{app:ext_optimal} derives the form of optimal policy without the positivity condition. Section~\ref{app::regret-oracle} establishes the regret of estimated optimal policy to the oracle policy with known $m_{11}(X)$. 

Section~\ref{app::proof} contains the proofs of all results.
\section{Additional results}\label{sec:addres}
\subsection{Efficient estimation of the upper bound of survivor average utility}\label{sec:estimationupper}
We can rewrite the upper bound in~\eqref{eqn::utility_upper_a} as
\begin{eqnarray*}
    U(\pi)&=&\frac{\Pb\left\{c_U(X)\gamma_{U1}(X)+\gamma_{U2}(X)\right\}}{p_0},
\end{eqnarray*}
where  
\begin{eqnarray*}
    c_U(X)&=&\bone\left\{p_1(X)\mu_1(X)-p_0(X)\geq 0\right\},\\\gamma_{U1}(X)&=&\pi(X)\left\{p_0(X)-p_1(X)\mu_1(X)\right\},\\
    \gamma_{U2}(X)&=&\pi(X)p_1(X)\mu_1(X)+\{1-\pi(X)\}\mu_0(X)p_0(X).
\end{eqnarray*}
For ease of exposition, denote $N_U(\pi)=\Pb\left\{c_U(X)\gamma_{U1}(X)+\gamma_{U2}(X)\right\}$ and $D=p_0$.

Similar to the efficient estimation of lower bound $L(\pi)$, we employ different estimation strategies for the smooth and non-smooth components. For $\gamma_{U1}(X)$ and $\gamma_{U2}(X)$, we construct estimators based on the efficient influence functions of $\E\{\gamma_{U1}(X)\}$ and $\E\{\gamma_{U2}(X)\}$. For $c_U(X)$, we use a plug-in approach that estimates the function inside the indicator function.

From the chain rule, the EIF of $\mathbb{E}\left\{\gamma_{U1}(X)\right\}$ is $\psi_{\gamma_{U1}}-\mathbb{E}\left\{\gamma_{U1}(X)\right\}$, where
\begin{eqnarray*}
\psi_{\gamma_{U1}}&=&\pi(X)\left\{\psi_{p_0}-\psi_{p_1\mu_1}\right\}.
\end{eqnarray*}
Similarly, the EIF of $\mathbb{E}\left\{\gamma_{U2}(X)\right\}$ is $\psi_{\gamma_{U2}}-\mathbb{E}\left\{\gamma_{U2}(X)\right\}$, where
\begin{eqnarray*}
\psi_{\gamma_{U2}}&=&\pi(X)\psi_{p_1\mu_1}+\left\{1-\pi(X)\right\}\psi_{p_0\mu_0}.
\end{eqnarray*}

Building on these results, we propose the following ratio estimator for $U(\pi)$.
\begin{estimator}\label{eg:ratio_upper} We propose a three-step estimator for the upper bound $U(\pi)$.
\begin{itemize}
    \item Step 1: Estimate the propensity score $\hat{e}(X)$, the observed survival rate $\hat{p}_z(X)$, and the outcome mean $\hat{\mu}_z(X)$ for $z=0,1$.
    \item Step 2: Obtain the estimators $\hat{\psi}_{p_z\mu_z}$ and $\hat{\psi}_{p_z}$ for $z=0,1$ by plugging  $\hat{e}(X)$, $\hat{p}_z(X)$, and $\hat{\mu}_z(X)$ into~\eqref{eqn:notationpmu}~and~\eqref{eqn:notationp}, respectively.
    \item Step 3: Compute the estimators $\hat{N}_U(\pi)$ and $\hat{D}$ for the numerator and denominator using 
    \begin{eqnarray*}
    \hat{N}_U(\pi)&=&\mathbb{P}_n\left[\pi(X)\bone\left\{\hat{p}_1(X)\hat{\mu}_1(X)-\hat{p}_0(X)\geq 0\right\}\left\{\hat{\psi}_{p_0}-\hat{\psi}_{p_1\mu_1}\right\}\right]\\
    &&+\mathbb{P}_n\left[\pi(X)\left\{\hat{\psi}_{p_1}\hat{\psi}_{\mu_1}\right\}+(1-\pi(X))\left\{\hat{\psi}_{p_0\mu_0}\right\}\right],\\
    \hat{D}&=&\mathbb{P}_n\hat{\psi}_{p_0}.
\end{eqnarray*}
The final estimator is given by: $\hat{U}(\pi)=\hat{N}_U(\pi)/\hat{D}$.
\end{itemize}
\end{estimator}

We next establish the asymptotic properties of the estimator $\hat{U}(\pi)$. To control the estimation error of indicator function, we impose an assumption parallel to Assumption~\ref{a:margin_posclass1} that restricts the probability mass in a neighborhood of the threshold.

\begin{assumption}\label{a:margin_posclass1_upper}
  There exists an $\alpha > 0$ and a constant $C$ such that for any $t \geq 0$,
  $$\mathbb{P}\left(\left|p_1(X)\mu_1(X)-p_0(X)\right| \leq t\right) \leq Ct^\alpha.$$
\end{assumption}

Denote
\begin{eqnarray*}
\|\hat{m}_U - m_U\|_\infty&=&\sup\left\{\|\hat{p}_1\hat{\mu}_1 - p_1\mu_1\|_\infty, \|\hat{p}_0 - p_0\|_\infty\right\}.
\end{eqnarray*}
The following theorem provides the asymptotic bias formulas of $\hat{N}_U(\pi)$ and $\hat{D}$.

\begin{theorem}\label{thm:convergence_margin_upper}
Suppose that Assumptions~\ref{assum::treatig}~to~\ref{a:margin_posclass1},~\ref{a:margin_posclass1_upper} and Conditions (a), (b), and (c) in Theorem~\ref{thm:convergence_margin} hold.
We have
\begin{eqnarray}
\nonumber \hat{N}_U(\pi)-N_U(\pi)
 &=&(\Pn - \Pb)\left\{c_U(X)\psi_{\gamma_{U1}}+\psi_{\gamma_{U2}}\right\} \\
 \nonumber &&+ O_\Pb\left(\left\|\hat{\lambda}_1-\lambda_1\right\|_2\left\|\hat{p}_1\hat{\mu}_1-p_1\mu_1\right\|_2)\right)\\
\nonumber &&+O_{\mathbb{P}}\left(\left\|\hat{\lambda}_0-\lambda_0\right\|_2\left\{\left\|\hat{p}_0\hat{\mu}_0-p_0\mu_0\right\|_2+\left\|\hat{p}_0-p_0\right\|_2\right\}\right)  \\
\label{eqn::biasN_upper}   && +O_\Pb\left(\|\hat{m}_U-m_U\|_{\infty}^{1+\alpha} \right) + o_\Pb(n^{-1/2}),\\
\label{eqn::biasD_upper} \hat{D}-D & =& (\Pn - \Pb)\left(\psi_{p_0}\right) + O_\Pb\left(\left\|\hat{\lambda}_0-\lambda_0\right\|_2\left\|\hat{p}_0-p_0\right\|_2\right)+o_\Pb(n^{-1/2}).
\end{eqnarray}
\end{theorem}

The next theorem shows the asymptotic normality of  $\hat{U}(\pi)$.
\begin{theorem}\label{thm:convergence_eif_upper}
Suppose that Assumptions~\ref{assum::treatig}~to~\ref{a:margin_posclass1},~\ref{assum:rates}(a)(b),~\ref{a:margin_posclass1_upper} and Conditions (a), (b), and (c) in Theorem~\ref{thm:convergence_margin} hold, and the estimated nuisance functions satisfy $\|\hat{m}_U-m_U\|_{\infty}^{1+\alpha} = o_\Pb(n^{-1/2})$. Then,
 \begin{eqnarray*}
    \sqrt{n}\{\hat{U}(\pi)-U(\pi)\}&\xrightarrow{\textup{d}}&\mathcal{N}\left(0,\mathbb{E}\left\{\frac{c_U(X)\psi_{\gamma_{U1}}+\psi_{\gamma_{U2}}-U(\pi)\psi_{p_0}}{p_0}\right\}^2 \right).
\end{eqnarray*}
\end{theorem}

\subsection{Optimal policy without the positivity condition} 
\label{app:ext_optimal}
Without requiring $p_1(x)-p_0(x)>0$ for all $x$, we may have $p_1(x)-p_0(x)=0$ for some $x$. Therefore, units can be divided into two regions according to the value of $p_1(x)-p_0(x)$. Define
\begin{eqnarray*}
\mathcal X_+ &=& \{x\in\mathcal X: p_1(x)-p_0(x)>0\},\\
\mathcal X_0 &=& \{x\in\mathcal X: p_1(x)-p_0(x)=0\}.
\end{eqnarray*}

For units with $X\in\mathcal X_0$, they do not contribute to the survival cost in the constraint. Therefore, whether these units should be modified, and in which direction, depends only on whether doing so increases the objective relative to the benchmark policy.

For units with $X\in\mathcal X_+$, the same priority scores based rule as in the main text applies. Since the priority scores $\rho_U$ and $\rho_L$ are only defined for units in $\mathcal X_+$, we define
\begin{eqnarray*}
Q_+(\eta)
&=&
\mathbb{E}\left[
\left\{
\varpi(X)\mathbf{1}\{\rho_L(X)\leq \eta<\rho_U(X)\}
+
\mathbf{1}\{\rho_L(X)>\eta\}
\right\}
\{p_1(X)-p_0(X)\}
\mathbf{1}\left(X\in \mathcal X_+\right)
\right],\\
\eta_+^E
&=&
\sup\left\{
\eta:
Q_+(\eta)
\geq
\mathbb{E}\left[\varpi(X)\{p_1(X)-p_0(X)\}\mathbf{1}\left(X\in \mathcal{X}_+\right)\right]
\right\},\quad\eta_+^S
\ =\ 
\min(0,\eta_+^E).
\end{eqnarray*}
The next theorem provides the explicit form of the optimal policy without the positivity condition.
\begin{theorem}
    \label{theo:pi_opt_ge0}
    Suppose that Assumptions~\ref{assum::treatig} and~\ref{assum::monoto} hold, define
\begin{eqnarray*}
  R_+(\eta)&=&\mathbb{E}\left[ \left\{\varpi(X)\bone(\rho_U(X) = \eta) + (1-\varpi(X))\bone(\rho_L(X) = \eta)\right\} \left\{p_1(X)-p_0(X)\right\}\mathbf{1}\left\{X\in \mathcal X_+\right\}\right].
\end{eqnarray*}
    \begin{itemize}
        \item On $\mathcal X_+$,
define $\pi_+^*(X)$ as

\noindent {\bf Case 1:} If $\eta_+^S = 0$, then
\begin{eqnarray*}
\pi_+^*(X) &=&  \varpi(X)\bone\{\rho_U(X)>0>\rho_L(X)\}+\bone\{\rho_L(X)>0\}.
\end{eqnarray*}
\noindent {\bf Case 2(a):}  If $\eta_+^S < 0$ and $R_+(\eta_+^S)=0$, then
\begin{eqnarray*}
\pi_+^*(X) &=&  \varpi(X)\bone\{\rho_U(X)>\eta^S_+>\rho_L(X)\}+\bone\{\rho_L(X)>\eta_+^S\}.
\end{eqnarray*}
\noindent {\bf Case 2(b):}  If $\eta_+^S < 0$ and $R_+(\eta_+^S)>0$, then
\begin{eqnarray*}
\pi_+^*(X) &=& \begin{cases}
\alpha_+ \cdot \varpi(X), & \text{if }\  \rho_U(X) = \eta_+^S,\\
\varpi(X) + \alpha_+ \cdot \{1-\varpi(X)\}, & \text{if }\  \rho_L(X) = \eta_+^S, \\
 \varpi(X)\bone\{\rho_U(X)>\eta_+^S>\rho_L(X)\}+\bone\{\rho_L(X)>\eta_+^S\}, & \text{otherwise},
\end{cases}
\end{eqnarray*}
where 
\begin{eqnarray*}
    \alpha_+ &=& \frac{\mathbb{E}\left[\varpi(X)\left\{p_1(X)-p_0(X)\right\}\right] - Q_+(\eta_+^S)}{R_+(\eta_+^S)}.
\end{eqnarray*}
        \item On $\mathcal X_0$, define $\pi_0^*(X)$ as 
        \begin{eqnarray}
            \label{eq:pi0_case_form}
            \pi_0^*(x)
            &=&
            \begin{cases}
            1, & \text{if } A_L(x)>0,\\[0.2cm]
            \varpi(x), & \text{if } A_U(x)\ge 0 \ge A_L(x),\\[0.2cm]
            0, & \text{if } A_U(x)<0.
            \end{cases}
            \end{eqnarray}
    \end{itemize}
    
    Then the optimal policy for~\eqref{eq:pi_opt_orig} is given by
    \begin{eqnarray}
    \label{eq:pi_star_ge0}
    \pi^*(X)
    &=&
    \pi^*_+(X)\bone\left(X\in\mathcal X_+\right)
    +
    \pi_0^*(X)\bone\left(X\in\mathcal X_0\right).
    \end{eqnarray}
    \end{theorem}
To estimate the optimal policy, we need to estimate the regions $\mathcal X_+$ and $\mathcal X_0$, as well as the two policy components $\pi_+^*$ and $\pi_0^*$. The regions $\mathcal X_+$ and $\mathcal X_0$, together with $\pi_0^*$, can be estimated by a naive plug-in approach. The estimation of $\pi_+^*$ is similar to Estimator~\ref{est:eif}.

In finite samples, however, the estimated survival contrast $\hat p_1(x)-\hat p_0(x)$ will generally not be exactly zero. Thus, for implementation, one may introduce a small tolerance parameter $\tau_n>0$ and classify units according to
\begin{eqnarray*}
\widehat{\mathcal X}_0(\tau_n)
&=&
\{x:\hat p_1(x)-\hat p_0(x)\leq \tau_n\},\\
\widehat{\mathcal X}_+(\tau_n)
&=&
\{x:\hat p_1(x)-\hat p_0(x)> \tau_n\}.
\end{eqnarray*}
This tolerance-based classification serves as a truncation of the estimated nuisance function. A formal analysis of the choice of $\tau_n$ and its implications for constraint violation and excess worst-case regret are left for future work.
\subsection{Regret relative to the oracle policy}
\label{app::regret-oracle}
The oracle policy, denoted by $\pi_o^*$, is the optimal policy when the non-identified component $m_{11}(X)$ is known. Under this setting, the survivor average utility $V(\pi)$ is point-identified, and policy learning under truncation by death can be formulated as\begin{eqnarray}
\label{eqn:optimiproblem_knownm}
    \pi_o^*\in \arg\max V(\pi) \text{ subject to }B(\pi)\geq B(\varpi),
\end{eqnarray}
without adopting a minimax regret approach.

We can further write the optimization problem in \eqref{eqn:optimiproblem_knownm} as 
\begin{eqnarray*}
    \pi_o^*\in \arg \max \mathbb{E}\left\{\pi(X)A(X)\right\} \text{ subject to } \mathbb{E}\left[\pi(X)\{p_1(X)-p_0(X)\}\right]\geq \mathbb{E}\left[\varpi(X)\{p_1(X)-p_0(X)\}\right],
\end{eqnarray*}
where \begin{eqnarray*}
    A(X)&=&p_0(X)\{m_{11}(X)-\mu_0(X)\}.
\end{eqnarray*}
The objective function and the constraint are both linear in $\pi(x)$. Therefore, we can define the following priority score
\begin{eqnarray*}
    \rho(X)\ =\  \frac{A(X)}{p_1(X)-p_0(X)}
\end{eqnarray*}
for $x$ with $p_1(x)-p_0(x)>0$,
which characterizes both an upward modification from $\varpi(x)$ toward 1 and a downward modification from $\varpi(x)$ toward 0. When $\rho(x)>0$, an upward modification increases the survival cost and improves the objective, whereas a downward modification decreases both the survival cost and the objective. In the absence of the constraint, the optimal policy would modify the benchmark policy in the direction that increases the objective.

Under the constraint, however, this unconstrained solution may fail to satisfy the constraint. In that case, the policy must be further adjusted to satisfy the constraint by replacing some downward modifications with upward modifications toward 1.  These additional adjustments generally reduce the objective relative to the unconstrained solution. The policy therefore proceeds in decreasing order of priority score, stopping when the required constraint level is satisfied. This yields a threshold-based rule, similar to Theorem~\ref{theo:pi_opt}, but simpler because only a single priority score is involved.  Similarly, we define
\begin{eqnarray*}
Q_o(\eta)
&=& \mathbb{E}\!\left[ \mathbf{1}\{\rho(X)>\eta\}\{p_1(X)-p_0(X)\} \right],\\
\eta_o^E &=& \sup\{\eta : Q_o(\eta)\ge \mathbb{E}[\varpi(X)\{p_1(X)-p_0(X)\}]\}, \quad \eta_o^S\ =\ \min(0,\eta_o^E).
\end{eqnarray*}
The next theorem provides the explicit form of the oracle policy.
\begin{theorem}\label{thm:pi_opt_knownm}
Suppose that Assumptions~\ref{assum::treatig} and~\ref{assum::monoto} hold, and $p_1(x)-p_0(x)>0$ for all  $x\in\mathcal{X}$.
Define
\begin{eqnarray*}
  R_o(\eta)&=&\mathbb{E}\left[ \bone(\rho(X) = \eta) \left\{p_1(X)-p_0(X)\right\}\right].
\end{eqnarray*}
 Then the oracle policy $\pi_o^*(X)$ is given by
 
\noindent {\bf Case 1:} If $\eta_o^S = 0$, then
\begin{eqnarray*}
\pi_o^*(X) &=&  \bone\{\rho(X)>0\}.
\end{eqnarray*}
\noindent {\bf Case 2(a):}  If $\eta_o^S < 0$ and $R_o(\eta_o^S)=0$, then
\begin{eqnarray*}
\pi_o^*(X) &=&  \bone\{\rho(X)>\eta_o^S\}.
\end{eqnarray*}
\noindent {\bf Case 2(b):}  If $\eta_o^S < 0$ and $R_o(\eta_o^S)>0$, then
\begin{eqnarray*}
\pi_o^*(X) &=& \begin{cases}
\alpha_o, & \text{if }\  \rho(X) = \eta_o^S, \\
\bone\{\rho(X)>\eta_o^S\}, & \text{otherwise},
\end{cases}
\end{eqnarray*}
where 
\begin{eqnarray*}
    \alpha_o &=& \frac{\mathbb{E}\left[\varpi(X)\left\{p_1(X)-p_0(X)\right\}\right] - Q_o(\eta^S)}{R_o(\eta^S)}.
\end{eqnarray*}
\end{theorem}

Next, we introduce a margin condition that controls the behavior of the true priority score $\rho(x)$ near its boundaries $\eta_o^S$.
\begin{assumption}\label{assum:margincondition2}
  There exists an $b > 0$ and a constant $C$ such that for any $\varepsilon \geq 0$, 
   \begin{eqnarray*}
    \mathbb{P}\left(|\rho(X)-\eta_o^S|\leq \varepsilon\right)\leq C \varepsilon^b.
    \end{eqnarray*}
\end{assumption}
The following theorem provides an upper bound of utility difference between the oracle policy $\pi_o^*$ relative to the optimal policy $\pi^*$ based on a minimax-regret approach.
\begin{theorem}\label{thm:dif_to_oracle}
Suppose that Assumption~\ref{assum::treatig},~\ref{assum::monoto},~\ref{assum:margincondition1}, and~\ref{assum:margincondition2} hold, and $p_1(x)-p_0(x)>0$ for all  $x\in\mathcal{X}$. Then 
    \begin{eqnarray*}
        R(\pi^*,\pi_o^*)&\leq& \frac{1}{p_0}C\left(\left\|\frac{p_0(X)\{m_{11}(X)-U_{11}(X)\}}{p_1(X)-p_0(X)}\right\|_{\infty}+|\eta_o^S-\eta^S|\right)^{1+b}\\
        &&+\frac{1}{p_0}C\left(\left\|\frac{p_0(X)\{m_{11}(X)-L_{11}(X)\}}{p_1(X)-p_0(X)}\right\|_{\infty}+|\eta_o^S-\eta^S|\right)^{1+b}\\
        &&-\frac{\eta_o^S}{p_0}\bone\{\sgn(\eta^S)\ne\sgn(\eta_o^S)\}\{Q(0)-\mathbb{E}\left[\varpi(X)\{p_1(X)-p_0(X)\}\right]\}.
    \end{eqnarray*}
\end{theorem}
The first two terms arise from the discrepancy between the non-identified quantity $m_{11}(X)$ and its identifiable bounds $L_{11}(X)$ and $U_{11}(X)$. This discrepancy is also reflected in the difference between the stopping thresholds $\eta_o^S$ and $\eta^S$. The third term captures the difference of the binding status, which is zero when the two policies are both binding or both nonbinding. These terms become small when the interval $[L_{11}(X),U_{11}(X)]$ is narrow.

\section{Proofs}
\label{app::proof}

\subsection{Proof of Theorem~\ref{theo::constraint}}
Under Assumptions~\ref{assum::treatig} and \ref{assum::monoto}, the probabilities of principal strata conditional on the covariates $X$ are:
\begin{eqnarray*}
    \Pb(U=11 \mid X) &=& \Pb(S(0)=1 \mid Z=0, X)\ =\ \Pb(S=1 \mid Z=0, X) = p_0(X), \\
    \Pb(U=00 \mid X) &=& \Pb(S(1)=0 \mid Z=1, X)\ =\ \Pb(S=0 \mid Z=1, X) = 1 - p_1(X), \\
    \Pb(U=10 \mid X) &=& 1 - \Pb(U=11 \mid X) - \Pb(U=00 \mid X)\ =\ p_1(X) - p_0(X).
\end{eqnarray*}
By the law of total probability, the marginal probabilities are:
\begin{eqnarray*}
    \Pb(U=11) &=& p_0, \quad \Pb(U=00)\ =\ 1 - p_1, \quad \Pb(U=10)\ =\ p_1 - p_0.
\end{eqnarray*}
Thus, we have
\begin{eqnarray}
    B(\pi) &=& \mathbb{E}\left\{\pi(X) \mid U=10\right\} \ =\ \mathbb{E}\left\{ \frac{\Pb(U=10 \mid X)} {\Pb(U=10)}\pi(X)\right\}\ =\ \mathbb{E}\left\{\frac{p_1(X) - p_0(X)}{p_1 - p_0} \pi(X)\right\}. \label{eqn::b1pi_proof}
\end{eqnarray}
For the survivor average utility, we have 
\begin{eqnarray}
    V(\pi)&=&\nonumber\mathbb{E}\{\pi(X) Y(1) + (1 - \pi(X)) Y(0) \mid U=11\} \\
    &=&\nonumber \mathbb{E}[\{1 - \pi(X)\} \mathbb{E}\{Y(0) \mid U=11, X\} \mid U=11]\\
    && + \mathbb{E}\{\pi(X) \mathbb{E}\{Y(1) \mid U=11, X\} \mid U=11\}. \label{eqn::utility_decomp}
\end{eqnarray}
From Assumptions~\ref{assum::treatig} and \ref{assum::monoto}, we have 
\begin{eqnarray*}
    \mathbb{E}\{Y(0) \mid U=11, X\} &=& \mathbb{E}\{Y(0) \mid S(0)=1, S(1)=1, X\}\\
    &=& \mathbb{E}\{Y(0) \mid S(0)=1, X\} \\
    &=& \mathbb{E}\{Y(0) \mid S=1, Z=0, X\} \\
    &=& \mu_0(X).
\end{eqnarray*}
Thus, for the first term on the right side of \eqref{eqn::utility_decomp}, we obtain
\begin{eqnarray}
    \nonumber && \mathbb{E}[\{1 - \pi(X)\} \mathbb{E}\{Y(0) \mid U=11, X\} \mid U=11]\\
    \nonumber &=& \mathbb{E}[\{1 - \pi(X)\} \mu_0(X) \mid U=11]\\ 
    \nonumber &=& \mathbb{E}\left[\{1 - \pi(X)\} \mu_0(X) \frac{\Pb(U=11\mid X)}{\Pb(U=11)}\right]\\
    &=& \frac{1}{p_0} \mathbb{E}[\{1 - \pi(X)\} \mu_0(X) p_0(X)]. \label{eqn::second erm}
\end{eqnarray}
For the second term on the right side of \eqref{eqn::utility_decomp}, consider
\begin{eqnarray*}
    \mu_1(X) &=& \mathbb{E}(Y \mid Z=1, S=1, X)\\ 
    &=& \mathbb{E}\{Y(1) \mid S(1)=1, X\} \\
    &=& \frac{\Pb(U=11\mid X)}{\Pb\{S(1)=1\mid X\}} \mathbb{E}\{Y(1) \mid U=11, X\} + \frac{\Pb(U=10\mid X)}{\Pb\{S(1)=1\mid X\}}\mathbb{E}\{Y(1) \mid U=10, X\} \\
    &=& \frac{p_0(X)}{p_1(X)} \mathbb{E}\{Y(1) \mid U=11, X\} + \frac{p_1(X) - p_0(X)}{p_1(X)} \mathbb{E}\{Y(1) \mid U=10, X\}. \label{eqn::mu1_decomp}
\end{eqnarray*}
Because $Y$ is binary, we can obtain the following upper and lower bounds on \(\mathbb{E}\{Y(1) \mid U=11, X\}\) by setting $\mathbb{E}\{Y(1) \mid U=10, X\}$ to $0$ and $1$:
\begin{eqnarray}
    U_{11}(X) &=& \min \left\{\frac{\mu_1(X) p_1(X)}{p_0(X)},1\right\}, \quad L_{11}(X)\ =\ \max\left\{\frac{p_1(X) \mu_1(X) - p_1(X) + p_0(X)}{p_0(X)},0\right\}. \label{eqn::bounds_y1}
\end{eqnarray}
Plugging~\eqref{eqn::second erm}~and~\eqref{eqn::bounds_y1} into~\eqref{eqn::utility_decomp} yields the bounds.
\QEDB
\subsection{Proof of Theorem~\ref{thm:convergence_margin}}
We first introduce a lemma to streamline the proof.
\begin{lemma}\label{lem:bias_decomp}
The EIF-based estimator satisfies
\begin{eqnarray}
\mathbb{P}\left\{ \hat{\psi}_{p_z} - p_z(X)\right\}  &=&  \mathbb{P}\left[\left\{\hat{p}_z(X)- p_z(X)\right\}\left\{\frac{\bone(Z=z)}{\mathbb{P}(Z=z\mid X)} - \frac{\bone(Z=z)}{\widehat{\mathbb{P}}(Z=z\mid X)}\right\}\right] \label{eq:bias_pz}, \\
\nonumber\mathbb{P}\left\{\hat{\psi}_{p_z\mu_z} - p_z(X)\mu_z(X)\right\}
 &=&  \mathbb{P}\left[\left\{\hat{p}_z(X)\hat{\mu}_z(X) - p_z(X)\mu_z(X)\right\}\right.\\
 &&\quad\cdot\left.\left\{\frac{\bone(Z=z)}{\mathbb{P}(Z=z\mid X)} - \frac{\bone(Z=z)}{\widehat{\mathbb{P}}(Z=z\mid X)}\right\}\right], \label{eq:bias_pzmuz}
\end{eqnarray}
where $\widehat{\mathbb{P}}(Z=z\mid X)=\hat{e}(X)$ for $z=1$ and $1-\hat{e}(X)$ for $z=0$.
\end{lemma}
\noindent {\it Proof of Lemma~\ref{lem:bias_decomp}. } 

For \eqref{eq:bias_pz}, we have 
\begin{eqnarray}\label{eqn:p1decomp}
 \nonumber\hat{\psi}_{p_z} - p_z(X) 
& = & \left\{ \hat{p}_z(X) -  p_z(X)\right\}\left(1 - \frac{\bone(Z=z)}{\mathbb{P}(Z=z\mid X)}\right) + \frac{\bone(Z=z)}{\widehat{\mathbb{P}}(Z=z\mid X)}\left\{ S - p_z(X)\right\} \\&&+\left\{\hat{p}_z(X) - p_z(X)\right\}\left\{\frac{\bone(Z=z)}{\mathbb{P}(Z=z\mid X)} - \frac{\bone(Z=z)}{\widehat{\mathbb{P}}(Z=z\mid X)}\right\}.
\end{eqnarray}
Then,~\eqref{eq:bias_pz} follows from taking expectation, because the first two terms on the right side of \eqref{eqn:p1decomp} vanish by tower property. 

For \eqref{eq:bias_pzmuz}, we have
\begin{eqnarray}
\label{eqn:p1mu1decomp}
&&\nonumber\hat{\psi}_{p_z\mu_z} - p_z(X)\mu_z(X)\\ &=&  \nonumber\left\{\hat{p}_z(X)\hat{\mu}_z(X) - p_z(X)\mu_z(X)\right\}\left(1 - \frac{\bone(Z=z)}{\mathbb{P}(Z=z\mid X)}\right) \\
&&\nonumber + \frac{\bone(Z=z)}{\widehat{\mathbb{P}}(Z=z\mid X)}\left\{YS - p_z(X)\mu_z(X)\right\} \\
&& +\left\{\hat{p}_z(X)\hat{\mu}_z(X) - p_z(X)\mu_z(X)\right\}\left(\frac{\bone(Z=z)}{\mathbb{P}(Z=z\mid X)} - \frac{\bone(Z=z)}{\widehat{\mathbb{P}}(Z=z\mid X)}\right).
\end{eqnarray}
Then,~\eqref{eq:bias_pzmuz} follows from taking expectation, because the first two terms on the right side of \eqref{eqn:p1mu1decomp} vanish by tower property.  
\QEDB

\subsubsection*{Asymptotic bias of the estimated numerator $\hat{N}(\pi)$}
We can write
\begin{eqnarray}
\hat{N}(\pi)-N(\pi)&=&
\nonumber\mathbb{P}_n\left\{\hat{c}(X)\hat{\psi}_{\gamma_1}+\hat{\psi}_{\gamma_2}\right\} - \mathbb{P}\left\{c(X)\psi_{\gamma_1}+\psi_{\gamma_2}\right\} \\
&=& R_1 +R_2 + (\Pn - \Pb)\{c(X)\psi_{\gamma_1}+\psi_{\gamma_2}\}, \label{eqn:Ndecomposen}
\end{eqnarray}
where
\begin{eqnarray*}
    R_1 &=&(\mathbb{P}_n - \mathbb{P})\{ \hat{c}(X)\hat{\psi}_{\gamma_1}+\hat{\psi}_{\gamma_2}-c(X)\psi_{\gamma_1}-\psi_{\gamma_2}\},\\
    R_2 &=& \mathbb{P}\{ \hat{c}(X)\hat{\psi}_{\gamma_1}+\hat{\psi}_{\gamma_2}-c(X)\psi_{\gamma_1}-\psi_{\gamma_2}\}.
\end{eqnarray*}
From Conditions (a), (b), and (c) together with Lemma 2 of \citet{kennedy2020sharp}, we have $R_1=o_{\mathbb{P}}(n^{-1/2})$. 

We then analyze the term $R_2$. We can decompose $R_2$ as $R_2=R_{21}+R_{22}+R_{23}$, where
\begin{eqnarray*}
    R_{21}&=&\Pb\left\{\hat{c}(X)\left(\hat{\psi}_{\gamma_1}-\psi_{\gamma_1}\right) \right\}, \quad R_{22}\ = \ \Pb\left(\hat{\psi}_{\gamma_2}-\psi_{\gamma_2}\right), \quad  R_{23}\ = \ \Pb\left[\psi_{\gamma_1}\left\{\hat{c}(X)-c(X)\right\}\right].
\end{eqnarray*}
For $R_{21}$, from Lemma~\ref{lem:bias_decomp}, we have 
\begin{eqnarray*}
\nonumber R_{21}&=& \Pb\left[\pi(X)\hat{c}(X)\left\{(\hat{\psi}_{p_1\mu_1}-\hat{\psi}_{p_1}+\hat{\psi}_{p_0})-(\psi_{p_1\mu_1}-\psi_{p_1}+\psi_{p_0})\right\} \right]\\
\nonumber &=& \Pb\left[\pi(X)\hat{c}(X)\left\{(\hat{\psi}_{p_1\mu_1}-\hat{\psi}_{p_1}+\hat{\psi}_{p_0})-(p_1(X)\mu_1(X)-p_1(X)+p_0(X))\right\} \right]\\
\nonumber &=&\mathbb{P}\left[\pi(X)\hat{c}(X)\left\{\hat{p}_1(X)\hat{\mu}_1(X)-\hat{p}_1(X)-p_1(X)\mu_1(X)+p_1(X)\right\}\left\{\frac{Z}{e(X)}-\frac{Z}{\hat{e}(X)}\right\}\right]\\&&+\mathbb{P}\left[\pi(X)\hat{c}(X)\left\{\hat{p}_0(X)-p_0(X)\right\}\left\{\frac{1-Z}{1-e(X)}-\frac{1-Z}{1-\hat{e}(X)}\right\}\right],
\end{eqnarray*} 
where the second equality follows from $\mathbb{E}(\psi_{p_z\mu_z}\mid X) = p_z(X)\mu_z(X)$ and  $\mathbb{E}(\psi_{p_z}\mid X) = p_z(X)$ for $z=0,1$, and the third equality follows from~\eqref{eq:bias_pz}~to~\eqref{eq:bias_pzmuz}.
By the Cauchy-Schwartz inequality, we have 
\begin{eqnarray}
\label{eqn:Ntermr2a} R_{21}&=&O_{\mathbb{P}}\left(\left\|\hat{\lambda}_1-\lambda_1\right\|_2\left\{\left\|\hat{p}_1\hat{\mu}_1-p_1\mu_1\right\|_2+\left\|\hat{p}_1-p_1\right\|_2\right\}\right).
\end{eqnarray}

Similarly, for $R_{22}$, we have
\begin{eqnarray*}
    R_{22}=\mathbb{P}\left[(1-\pi(X))\left\{\hat{p}_0(X)\hat{\mu}_0(X)-p_0(X)\mu_0(X)\right\}\left\{\frac{1-Z}{1-e(X)}-\frac{1-Z}{1-\hat{e}(X)}\right\}\right].
\end{eqnarray*}
By the Cauchy-Schwartz inequality, we have 
\begin{eqnarray}
\label{eqn:Ntermr2b} R_{22}&=&O_{\mathbb{P}}\left(\left\|\hat{\lambda}_0-\lambda_0\right\|_2\left\{\left\|\hat{p}_0\hat{\mu}_0-p_0\mu_0\right\|_2+\left\|\hat{p}_0-p_0\right\|_2\right\}\right).
\end{eqnarray}

For $R_{23}$, we have 
\begin{eqnarray*}
|R_{23}|&=&|\Pb\left[\left\{\hat{c}(X)-c(X)\right\}\left\{p_1(X)\mu_1(X)-p_1(X)+p_0(X)\right\}\right]|\\
&\leq&\Pb\left(|\hat{c}(X)-c(X) ||p_1(X)\mu_1(X)-p_1(X)+p_0(X)|\right)\\
& = &\E\left[\bone\left\{\hat{c}(X) \neq c(X)\right\}\left|p_1(X)\mu_1(X)-p_1(X)+p_0(X)\right|\right].
\end{eqnarray*}
For ease of exposition, denote 
\begin{eqnarray*}
    g(X)&=& p_1(X)\mu_1(X)-p_1(X)+p_0(X), \quad \hat{g}(X)\ = \ \hat{p}_1(X)\hat{\mu}_1(X)-\hat{p}_1(X)+\hat{p}_0(X).
\end{eqnarray*}
Then, we can write $\hat{c}(X)=\bone(\hat{g}(X)\geq 0)$ and $c(X)=\bone(g(X)\geq 0) $. If $\hat{c}(X)\ne c(X)$, then $\hat{g}(X)$ and $g(X)$ have opposite signs, which implies that $|\hat{g}(X)-g(X)|\geq |g(X)|$.
Therefore, we have 
\begin{eqnarray}
  \nonumber  |R_{23}|&\leq & \E \left\{ \bone(|\hat{g}(X)-g(X)|\geq |g(X)|)\cdot |g(X)| \right\}\\
  \nonumber  &\leq &\E \left\{ \bone(|\hat{g}(X)-g(X)|\geq |g(X)|)\cdot |\hat{g}(X)-g(X)| \right\}\\
 \nonumber   &\leq& ||\hat{g}(X)-g(X)||_\infty \Pb\{|\hat{g}(X)-g(X)|\geq |g(X)|\}\\
  \nonumber  &\leq&  C ||\hat{g}(X)-g(X)||_\infty^{1+\alpha}\\
  \nonumber  &=&C(\|\hat{p}_1\hat{\mu}_1 - p_1\mu_1\|_\infty + \|\hat{p}_1 - p_1\|_\infty+\|\hat{p}_0 - p_0\|_\infty)^{1+\alpha}\\
\label{eqn:Ntermr2c}    &\leq&C\|\hat{m}-m\|_{\infty}^{1+\alpha},
\end{eqnarray}
where the third inequality follows from Assumption~\ref{a:margin_posclass1}.
Combining~\eqref{eqn:Ntermr2a}~to~\eqref{eqn:Ntermr2c} yields 
\begin{eqnarray}
 \nonumber   R_2&=&O_\Pb\left(\left\|\hat{\lambda}_1-\lambda_1\right\|_2\left\{\left\|\hat{p}_1\hat{\mu}_1-p_1\mu_1\right\|_2+\left\|\hat{p}_1-p_1\right\|_2\right\}\right)\\
\nonumber &&+O_{\mathbb{P}}\left(\left\|\hat{\lambda_0}-\lambda_0\right\|_2\left\{\left\|\hat{p}_0\hat{\mu}_0-p_0\mu_0\right\|_2+\left\|\hat{p}_0-p_0\right\|_2\right\}\right)  \\
\label{eqn::biasN-R2}   && +O_\Pb\left(\|\hat{m}-m\|_{\infty}^{1+\alpha} \right).
\end{eqnarray}
Plugging~\eqref{eqn::biasN-R2}~and $R_1=o_\Pb(n^{-1/2})$ into~\eqref{eqn:Ndecomposen} yields the asymptotic bias of $\hat{N}(\pi)$.
\subsubsection*{Asymptotic bias of the estimated denominator $\hat{D}$}
We have 
\begin{eqnarray}
\nonumber\hat{D}-D & = &(\mathbb{P}_n - \mathbb{P})\{ \hat{\psi}_{p_0}-\psi_{p_0}\} + \mathbb{P}( \hat{\psi}_{p_0}-\psi_{p_0})  + (\mathbb{P}_n - \mathbb{P})\psi_{p_0}. \label{eqn:Ddecomposed}
\end{eqnarray}
The first term on the right hand side of~\eqref{eqn:Ddecomposed} satisfies $(\mathbb{P}_n - \mathbb{P})\{ \hat{\psi}_{p_0}-\psi_{p_0}\}=o_{\mathbb{P}}(n^{-1/2})$ by Conditions (a), (b), and (c) together with Lemma 2 of \citet{kennedy2020sharp}. For the second term, we have 
\begin{eqnarray*}
   \mathbb{P}( \hat{\psi}_{p_0}-\psi_{p_0})&=&  \mathbb{P}\left[\left\{\hat{p}_0(X)-p_0(X)\right\}\left\{\frac{1-Z}{1-\hat{e}(X)}-\frac{1-Z}{1-e(X)}\right\}\right]\ = \  O_\Pb\left(\left\|\hat{\lambda}_0-\lambda_0\right\|_2\left\|\hat{p}_0-p_0\right\|_2\right),
\end{eqnarray*}
where the last equality follows from  the Cauchy-Schwartz inequality.
\QEDB

\subsection{Proof of Theorem~\ref{thm:convergence_eif} }

From Theorem~\ref{thm:convergence_margin} and Assumption~\ref{assum:rates}, we have
\begin{eqnarray}\label{eqn:ndconvergence}
\hat{N}(\pi)-N(\pi)&=&(\mathbb{P}_n-\mathbb{P})\left\{c(X)\psi_{\gamma_1}+\psi_{\gamma_2}\right\}+o_\Pb(n^{-1/2})\nonumber\\
\hat{D}-D&=&(\mathbb{P}_n-\mathbb{P})\left(\psi_{p_0}\right)+o_\Pb(n^{-1/2}).
\end{eqnarray}
By Taylor expansion, we have
\begin{eqnarray}
\hat{L}(\pi) - L(\pi) 
& = & \frac{\hat{N}(\pi)}{\hat{D}} - \frac{N(\pi)}{D} \nonumber \\
& = & \frac{1}{D} \left\{\hat{N}(\pi) - N(\pi)\right\} - \frac{N(\pi)}{D^2} \left\{\hat{D} - D\right\} + o_\Pb(n^{-1/2}). \label{eqn:taylor_expansion}
\end{eqnarray}

Substituting~\eqref{eqn:ndconvergence} into \eqref{eqn:taylor_expansion}, we obtain
\begin{eqnarray*}
\hat{L}(\pi) - L(\pi) 
& = & (\mathbb{P}_n - \mathbb{P}) \left\{ \frac{c(X)\psi_{\gamma_1} + \psi_{\gamma_2}}{D} - \frac{N(\pi)}{D^2} \psi_{p_0} \right\} + o_\Pb(n^{-1/2}) \nonumber \\
& = & (\mathbb{P}_n - \mathbb{P}) \left\{ \frac{c(X)\psi_{\gamma_1} + \psi_{\gamma_2} - L(\pi)\psi_{p_0}}{p_0} \right\} + o_\Pb(n^{-1/2}).
\end{eqnarray*}
The asymptotic normality then follows from the central limit theorem.
\QEDB

\subsection{Proof of Theorem~\ref{thm:convergence_margin_upper}}

We only need to prove the the asymptotic bias formula of the estimated denominator $\hat{N}_U(\pi)$.

We can write
\begin{eqnarray}
\hat{N}_U(\pi)-N_U(\pi)&=&
\nonumber\mathbb{P}_n\left\{\hat{c}_U(X)\hat{\psi}_{\gamma_{U1}}+\hat{\psi}_{\gamma_{U2}}\right\} - \mathbb{P}\left\{c_U(X)\psi_{\gamma_{U1}}+\psi_{\gamma_{U2}}\right\}\\
& =& R_1 + R_2 + (\mathbb{P}_n - \mathbb{P})\left\{c_U(X)\psi_{\gamma_{U1}}+\psi_{\gamma_{U2}}\right\}.\label{eqn:decomposenu}
\end{eqnarray}
where \begin{eqnarray*}
    R_1&=&(\mathbb{P}_n - \mathbb{P})\{ \hat{c}_U(X)\hat{\psi}_{\gamma_{U1}}+\hat{\psi}_{\gamma_{U2}}-c_U(X)\psi_{\gamma_{U1}}-\psi_{\gamma_{U2}}\}\\
    R_2&=&\mathbb{P}\{ \hat{c}_U(X)\hat{\psi}_{\gamma_{U1}}+\hat{\psi}_{\gamma_{U2}}-c_U(X)\psi_{\gamma_{U1}}-\psi_{\gamma_{U2}}\}
\end{eqnarray*}
From Conditions (a), (b), and (c) in Theorem~\ref{thm:convergence_margin} together with Lemma~2 of \citet{kennedy2020sharp}, we have $R_1=o_{\mathbb{P}}(n^{-1/2})$. 

We then analyze the term $R_2$. We can decompose $R_2$ as $R_2=R_{21}+R_{22}+R_{23}$, where 
  \begin{eqnarray*}
    R_{21} & =&\Pb\left[\hat{c}_U(X)\left\{\hat{\psi}_{\gamma_{U1}}-\psi_{\gamma_{U1}}\right\} \right],\quad R_{22}\ =\ \Pb\left[\hat{\psi}_{\gamma_{U2}}-\psi_{\gamma_{U2}}\right],\quad R_{23}\ =\ \Pb\left[\psi_{\gamma_{U1}}\left\{\hat{c}_U(X)-c_U(X)\right\}\right].
  \end{eqnarray*}
  For $R_{21}$, from Lemma~\ref{lem:bias_decomp}, we have
\begin{eqnarray*}
\nonumber R_{21}&=& \Pb\left[\pi(X)\hat{c}_U(X)\left\{(\hat{\psi}_{p_0}-\hat{\psi}_{p_1\mu_1})-(\psi_{p_0}-\psi_{p_1\mu_1})\right\} \right]\\
\nonumber &=& \Pb\left[\pi(X)\hat{c}_U(X)\left\{(\hat{\psi}_{p_0}-\hat{\psi}_{p_1\mu_1})-(p_0(X)-p_1(X)\mu_1(X))\right\} \right]\\
\nonumber &=&-\mathbb{P}\left[\pi(X)\hat{c}(X)\left\{\hat{p}_1(X)\hat{\mu}_1(X)-p_1(X)\mu_1(X)\right\}\left\{\frac{Z}{e(X)}-\frac{Z}{\hat{e}(X)}\right\}\right]\\&&+\mathbb{P}\left[\pi(X)\hat{c}(X)\left\{\hat{p}_0(X)-p_0(X)\right\}\left\{\frac{1-Z}{1-\hat{e}(X)}-\frac{1-Z}{1-e(X)}\right\}\right],
\end{eqnarray*} 
where the second equality follows from $\mathbb{E}(\psi_{p_z\mu_z}\mid X) = p_z(X)\mu_z(X)$ and  $\mathbb{E}(\psi_{p_z}\mid X) = p_z(X)$ for $z=0,1$, and the third equality follows from~\eqref{eq:bias_pz}~and~\eqref{eq:bias_pzmuz}.
By the Cauchy-Schwartz inequality, we have 
\begin{eqnarray}
\label{eqn:Nutermr2a} R_{21}&=&O_{\mathbb{P}}\left(\left\|\hat{\lambda}_1-\lambda_1\right\|_2\left\|\hat{p}_1\hat{\mu}_1-p_1\mu_1\right\|_2\right)+O_{\mathbb{P}}\left(\left\|\hat{\lambda}_0-\lambda_0\right\|_2\left\|\hat{p}_0-p_0\right\|_2\right).
\end{eqnarray}
Similarly, for $R_{22}$, we have
\begin{eqnarray*}
    \nonumber R_{22}&=&\mathbb{P}\left[(1-\pi(X))\left\{\hat{p}_0(X)\hat{\mu}_0(X)-p_0(X)\mu_0(X)\right\}\left\{\frac{1-Z}{1-\hat{e}(X)}-\frac{1-Z}{1-e(X)}\right\}\right]\\&&+\mathbb{P}\left[\pi(X)\left\{\hat{p}_1(X)\hat{\mu}_1(X)-p_1(X)\mu_1(X)\right\}\left\{\frac{Z}{\hat{e}(X)}-\frac{Z}{e(X)}\right\}\right]\label{eqn:termr2b}.
\end{eqnarray*}
By the Cauchy-Schwartz inequality, we have 
\begin{eqnarray}
\label{eqn:Nutermr2b} R_{22}&=&O_{\mathbb{P}}\left(\left\|\hat{\lambda}_1-\lambda_1\right\|_2\left\|\hat{p}_1\hat{\mu}_1-p_1\mu_1\right\|_2\right)+O_{\mathbb{P}}\left(\left\|\hat{\lambda}_0-\lambda_0\right\|_2\left\|\hat{p}_0\hat{\mu}_0-p_0\mu_0\right\|_2\right).
\end{eqnarray}

For $R_{23}$, we have 
\begin{eqnarray*}
|R_{23}|&=&|\Pb\left(\left[\hat{c}_U(X)-c_U(X) \right]\left\{p_1(X)\mu_1(X)-p_0(X)\right\}\right)|\\&\leq&\Pb\left(|\hat{c}_U(X)-c_U(X) ||p_1(X)\mu_1(X)-p_0(X)|\right)\\& = &\E\left[\bone\left\{\hat{c}_U(X) \neq c_U(X)\right\}\left|p_1(X)\mu_1(X)-p_0(X)\right|\right].
\end{eqnarray*}
For ease of exposition, denote 
\begin{eqnarray*}
    g(X)&=& p_1(X)\mu_1(X)-p_0(X), \quad \hat{g}(X)\ = \ \hat{p}_1(X)\hat{\mu}_1(X)-\hat{p}_0(X).
\end{eqnarray*}
Then, we can write $\hat{c}_U(X)=\bone(\hat{g}(X)\geq 0)$ and $c_U(X)=\bone(g(X)\geq 0) $. If $\hat{c}_U(X)\ne c_U(X)$, then $\hat{g}(X)$ and $g(X)$ have opposite signs, which implies that $|\hat{g}(X)-g(X)|\geq |g(X)|$.
Therefore, we have 
\begin{eqnarray}
  \nonumber  |R_{23}|&\leq & \E \left\{ \bone(|\hat{g}(X)-g(X)|\geq |g(X)|)\cdot |g(X)| \right\}\\
  \nonumber  &\leq &\E \left\{ \bone(|\hat{g}(X)-g(X)|\geq |g(X)|)\cdot |\hat{g}(X)-g(X)| \right\}\\
 \nonumber   &\leq& ||\hat{g}(X)-g(X)||_\infty \Pb\{|\hat{g}(X)-g(X)|\geq |g(X)|\}\\
  \nonumber  &\leq&  C ||\hat{g}(X)-g(X)||_\infty^{1+\alpha}\\
  \nonumber  &=&C(\|\hat{p}_1\hat{\mu}_1 - p_1\mu_1\|_\infty +\|\hat{p}_0 - p_0\|_\infty)^{1+\alpha}\\
\label{eqn:Nutermr2c}    &\leq&C\|\hat{m}_U-m_U\|_{\infty}^{1+\alpha},
\end{eqnarray}
where the third inequality follows from Assumption~\ref{a:margin_posclass1_upper}. Combining~\eqref{eqn:Nutermr2a}~to~\eqref{eqn:Nutermr2c} yields 
\begin{eqnarray}
 \nonumber   R_2&=&O_\Pb\left(\left\|\hat{\lambda}_1-\lambda_1\right\|_2\left\|\hat{p}_1\hat{\mu}_1-p_1\mu_1\right\|_2\right)\\
\nonumber &&+O_{\mathbb{P}}\left(\left\|\hat{\lambda_0}-\lambda_0\right\|_2\left\{\left\|\hat{p}_0\hat{\mu}_0-p_0\mu_0\right\|_2+\left\|\hat{p}_0-p_0\right\|_2\right\}\right)  \\
\label{eqn::biasNu-R2}   && +O_\Pb\left(\|\hat{m}_U-m_U\|_{\infty}^{1+\alpha} \right).
\end{eqnarray}
Plugging~\eqref{eqn::biasNu-R2}~and $R_1=o_\Pb(n^{-1/2})$ into~\eqref{eqn:decomposenu} yields the asymptotic bias of $\hat{N}_U(\pi)$.
\QEDB

\subsection{Proof of Theorem~\ref{thm:convergence_eif_upper} }

From Theorem~\ref{thm:convergence_margin_upper},~Assumption~\ref{assum:rates}(a) and (b) together with condition of Theorem~\ref{thm:convergence_eif_upper}, we have\begin{eqnarray}\label{eqn:ndconvergence_upper}
\hat{N}_U(\pi)-N_U(\pi)&=&(\mathbb{P}_n-\mathbb{P})\left\{c_U(X)\psi_{\gamma_{U1}}+\psi_{\gamma_{U2}}\right\}+o_\Pb(n^{-1/2})\nonumber\\\hat{D}-D&=&(\mathbb{P}_n-\mathbb{P})\left(\psi_{p_0}\right)+o_\Pb(n^{-1/2}).
\end{eqnarray}
By Taylor expansion, we have
\begin{eqnarray}
\hat{U}(\pi) - U(\pi) 
& = & \frac{\hat{N}_U(\pi)}{\hat{D}} - \frac{N_U(\pi)}{D} \nonumber \\
& = & \frac{1}{D} \left\{\hat{N}_U(\pi) - N_U(\pi)\right\} - \frac{N_U(\pi)}{D^2} \left\{\hat{D} - D\right\} + o_\Pb(n^{-1/2}). \label{eqn:taylor_expansion_upper}
\end{eqnarray}

Substituting the results from \eqref{eqn:ndconvergence_upper} into \eqref{eqn:taylor_expansion_upper}, we obtain
\begin{eqnarray*}
\hat{U}(\pi) - U(\pi) 
& = & (\mathbb{P}_n - \mathbb{P}) \left\{ \frac{c_U(X)\psi_{\gamma_{U1}} + \psi_{\gamma_{U2}}}{D} - \frac{N_U(\pi)}{D^2} \psi_{p_0} \right\} + o_\Pb(n^{-1/2}) \nonumber \\
& = & (\mathbb{P}_n - \mathbb{P}) \left\{ \frac{c_U(X)\psi_{\gamma_{U1}} + \psi_{\gamma_{U2}} - U(\pi)\psi_{p_0}}{p_0} \right\} + o_\Pb(n^{-1/2}) .\label{eqn:l_convergence}
\end{eqnarray*}
 The asymptotic normality then follows from the central limit theorem.
\QEDB

\subsection{Proof of Corollary~\ref{cor:formsup}}
The formula for $B(\varpi)$ follows directly from Theorem~\ref{theo::constraint}. We now derive the formula for $ R_{\sup}(\pi,\varpi)$.

For any policy $\pi$, we have
\begin{eqnarray*}
    \nonumber V(\pi)&=&\mathbb{E}\left[Y(1)\pi(X)+Y(0)\{1-\pi(X)\} \mid U=11\right]\\
\nonumber&=&\mathbb{E}\left[\mathbb{E}\left\{Y(1)\pi(X)+Y(0)\{1-\pi(X)\} \mid U=11,X\right\}\mid U=11\right]\\
&=&\nonumber\mathbb{E}\left[\pi(X)m_{11}(X)+\{1-\pi(X)\}\mu_0(X)\mid U=11\right]\\
\nonumber&=&\mathbb{E}\left[\frac{p_0(X)}{p_0}\pi(X)m_{11}(X)+\frac{p_0(X)}{p_0}\{1-\pi(X)\}\mu_0(X)\right]\\
    &=&\mathbb{E}\left[\frac{p_0(X)}{p_0}\pi(X)\left\{m_{11}(X)-\mu_0(X)\right\}\right]+\mathbb{E}\left\{\frac{p_0(X)}{p_0}\mu_0(X)\right\}.
    \end{eqnarray*}
Therefore, the regret of $\pi$ relative to  $\varpi$ is
\begin{eqnarray*}
    R(\pi,\varpi)&=&\frac{1}{p_0}\mathbb{E}\left[\left\{\varpi(X)-\pi(X)\right\}p_0(X)\left\{m_{11}(X)-\mu_0(X)\right\}\right]
\end{eqnarray*}
To derive the worst-case regret, note that for each covariate value $x$, 
\begin{itemize}
\item if $\pi(x)<\varpi(x)$, then $\left\{\varpi(x)-\pi(x)\right\}p_0(x)\left\{m_{11}(x)-\mu_0(x)\right\}$  is maximized when $m_{11}(x)$ takes its upper bound $U_{11}(x)$;
\item if  $\pi(x)\geq \varpi(x)$, then $\left\{\varpi(x)-\pi(x)\right\}p_0(x)\left\{m_{11}(x)-\mu_0(x)\right\}$  is maximized when $m_{11}(x)$ takes its lower bound $L_{11}(x)$.
\end{itemize}
As a result, we have 
\begin{eqnarray*}
        R_{\sup}(\pi,\varpi)
    &=& \frac{1}{p_0} \mathbb{E} \left[ \left\{\varpi(X)-\pi(X)\right\} \left\{A_U(X)\bone(\pi(X)<\varpi(X))+A_L(X)\bone(\pi(X)\geq \varpi(X))\right\} \right],
    \end{eqnarray*}
where \begin{eqnarray*}
A_L(X) &= p_0(X)\{L_{11}(X) - \mu_0(X)\},\\
A_U(X) &= p_0(X)\{U_{11}(X) - \mu_0(X)\}.
\end{eqnarray*}
\QEDB

\subsection{Proof of Theorem~\ref{theo:pi_opt_ge0}}

We introduce some lemmas to simplify the proof.
\begin{lemma}
\label{lem:pointwise_x+}
For any fix $x\in\mathcal{X}_+$, define 
\begin{eqnarray*}
g(\pi) &=& \{\pi - \varpi(x)\} \left\{ A_U(x) \bone(\pi < \varpi(x)) + A_L(x) \bone(\pi \ge \varpi(x)) \right\} - \eta_+^S \pi\{p_1(x)-p_0(x)\},
\quad \pi\in[0,1].
\end{eqnarray*}
Then one maximizer of $g(\pi)$ over $\pi\in[0,1]$ is given by
\begin{eqnarray*}
\pi^*_+(x)
=
\begin{cases}
1, & \rho_L(x)>\eta_+^S,\\
\varpi(x), & \rho_U(x)> \eta_+^S > \rho_L(x),\\
0, & \rho_U(x)<\eta_+^S,\\
\alpha_+\varpi(x), & \rho_U(x)=\eta_+^S,\\
\varpi(x)+\alpha_+\left\{1-\varpi(x)\right\}, & \rho_L(x)=\eta_+^S.\\
\end{cases}
\end{eqnarray*}
\end{lemma}

\noindent {\it Proof of Lemma~\ref{lem:pointwise_x+}. } We derive the maximizer for a fixed $x\in\mathcal X_+$ and thus omit $x$ in the argument of functions.

Because $g(\pi)$ is a piecewise linear function in $\pi$, its (one-sided) derivatives on the two linear segments are
\begin{eqnarray*}
g'(\pi) &=&
\begin{cases}
(\rho_U - \eta_+^S)(p_1-p_0) & \pi \in [0, \varpi), \\
(\rho_L - \eta_+^S)(p_1-p_0) & \pi \in (\varpi, 1].
\end{cases}
\end{eqnarray*}
At the non-differentiable point  $\varpi$,  the left and right derivatives are
\begin{eqnarray*}
g'_{-}(\varpi)&=& (\rho_U - \eta_+^S)(p_1-p_0),\\
g'_{+}(\varpi)&=& (\rho_L - \eta_+^S)(p_1-p_0).
\end{eqnarray*} 
We then analyze $g(\pi)$ case by case. 
\begin{itemize}
\item $\rho_L > \eta_+^S$. We have both one-sided derivatives to be positive on $[0,1]$. Therefore, $g(\pi)$ achieves its maximum at $\pi=1$.
\item $\rho_U > \eta_+^S > \rho_L$. We have $g'(\pi) >0$ on $[0, \varpi)$, $g'(\pi) <0$ on $(\varpi, 1]$, $g'_{-}(\varpi)>0$, and $g'_{+}(\varpi)<0$. Thus,  $g(\pi)$ achieves its maximum at $\pi=\varpi$.
\item $\rho_U < \eta_+^S$. We have both one-sided derivatives to be negative on $[0,1]$. Therefore, $g(\pi)$ achieves its maximum at $\pi=0$.
\item $ \rho_U = \eta_+^S$. We have $g'(\pi) =0$ on $[0, \varpi)$, $g'(\pi) <0$ on $(\varpi, 1]$, $g'_{-}(\varpi)=0$, and $g'_{+}(\varpi)<0$. 
Therefore, $g(\pi)$ is constant on $[0, \varpi]$ and decreasing on $[\varpi, 1]$.
As a result, $g(\pi)$ achieves its maximum at any value in $[0, \varpi]$.
\item  $\rho_L = \eta_+^S$. We have $g'(\pi) >0$ on $[0, \varpi)$, $g'(\pi) =0$ on $(\varpi, 1]$, $g'_{-}(\varpi)>0$, and $g'_{+}(\varpi)=0$. 
Therefore, $g(\pi)$ is increasing on $[0, \varpi]$ and constant on  $[\varpi, 1]$.
As a result, $g(\pi)$ achieves its maximum at any value in $[\varpi, 1]$.
\end{itemize}
From the definition of $\alpha_+$, we have $\alpha_+\in [0,1]$, therefore,  $ \pi_+^*$ maximizes $g(\pi)$ in all cases.
 \QEDB

\begin{lemma}
\label{lem:pointwise_x0}
For any fix $x\in\mathcal X_0$, define
\begin{eqnarray*}
h(\pi)
=
\{\pi-\varpi(x)\}
\left\{
A_U(x)\bone(\pi<\varpi(x))
+
A_L(x)\bone(\pi\ge\varpi(x))
\right\},
\quad \pi\in[0,1].
\end{eqnarray*}
Then one maximizer of $h(\pi)$ over $\pi\in[0,1]$ is given by
\begin{eqnarray*}
\pi^*_0(x)
=
\begin{cases}
1, & A_L(x)>0,\\
\varpi(x), & A_U(x)\ge 0 \ge A_L(x),\\
0, & A_U(x)<0.
\end{cases}
\end{eqnarray*}
\end{lemma}

\noindent {\it Proof of Lemma~\ref{lem:pointwise_x0}. } We derive the maximizer for a fixed $x\in\mathcal X_0$ and thus omit $x$ in the argument of functions.

Because $h(\pi)$ is a piecewise linear function in $\pi$, its (one-sided) derivatives on the two linear segments are
\begin{eqnarray*}
h'(\pi)
&=&
\begin{cases}
A_U, & \pi\in[0,\varpi),\\
A_L, & \pi\in(\varpi,1].
\end{cases}
\end{eqnarray*}
At the non-differentiable point  $\varpi$,  the left and right derivatives are
\begin{eqnarray*}
h'_{-}(\varpi)&=&A_U,\qquad h'_{+}(\varpi)=A_L.
\end{eqnarray*}
We then analyze $h(\pi)$ case by case.

\begin{itemize}
\item $A_L>0$. Then $h(\pi)$ is increasing on $[0,\varpi]$
and increasing on $[\varpi,1]$, so it achieves its maximum at $\pi=1$.

\item $A_U<0$. Then $h(\pi)$ is decreasing on $[0,\varpi]$
and decreasing on $[\varpi,1]$, so it achieves its maximum at $\pi=0$.

\item $A_U\ge 0 \ge A_L$. Then $h(\pi)$ is increasing on $[0,\varpi]$ and decreasing on
$[\varpi,1]$, so it achieves its maximum at $\pi=\varpi$.
\end{itemize}
In the boundary cases, the maximizer may not be unique. For example, if $A_L=0$, then any
$\pi\in[\varpi,1]$ is optimal; if $A_U=0$ then any $\pi\in[0,\varpi]$ is optimal. Nevertheless, in all cases, $\pi_0^*$ remains a valid maximizer.
\QEDB 

We then prove Theorem~\ref{theo:pi_opt_ge0}. We first show that $\pi^*(X)$ satisfies the constraint, i.e, \begin{eqnarray*}
    \E\left[\pi^*(X)\{p_1(X)-p_0(X)\}\right]\geq\E\left[\varpi(X)\{p_1(X)-p_0(X)\}\right]. 
\end{eqnarray*}

\noindent
\textbf{Case 1: $\eta_+^S = 0$.} By definition, we have $\eta_+^E \geq 0$, which implies \begin{eqnarray*}
    Q_+(0) \geq  \mathbb{E}\left[\varpi(X)\left\{p_1(X)-p_0(X)\right\}\bone(X\in\mathcal{X}_+)\right].
\end{eqnarray*}
Therefore, we have
\begin{eqnarray}
\nonumber&&\mathbb{E}\left[\pi_+^*(X)\left\{p_1(X)-p_0(X)\right\}\bone\left(X\in\mathcal{X}_+\right)\right] \\
\nonumber&=& \mathbb{E}\left[\bone(\rho_L(X) > 0)\left\{p_1(X)-p_0(X)\right\}\bone\left(X\in\mathcal{X}_+\right)\right] \\
\nonumber&&+ \mathbb{E}\left[\varpi(X)\left\{p_1(X)-p_0(X)\right\} \bone(\rho_U(X) > 0\ge \rho_L(X))\bone(X\in\mathcal{X}_+)\right] \\
\nonumber &=& Q_+(0)\\
&\geq& \mathbb{E}\left[\varpi(X)\left\{p_1(X)-p_0(X)\right\}\bone(X\in\mathcal{X}_+)\right].\label{eqn:proof-constraint_excess}
\end{eqnarray}

\noindent
\textbf{Case 2: $\eta_+^S < 0$.}  We have \begin{eqnarray*}
    \pi_+^*(X)\bone\left(X\in\mathcal{X}_+\right)&=&\bone(\rho_L(X) > \eta_+^S)\bone\left(X\in\mathcal{X}_+\right)+\varpi(X)\bone(\rho_U(X) > \eta_+^S > \rho_L(X))\bone\left(X\in\mathcal{X}_+\right)\\
    &&+\alpha_+ \varpi(X) \bone(\rho_U(X) = \eta_+^S)\bone\left(X\in\mathcal{X}_+\right)\\
    &&+\{\varpi(X) + \alpha_+ - \alpha_+\varpi(X)\}\bone(\rho_L(X) = \eta_+^S)\bone\left(X\in\mathcal{X}_+\right)\\
    &=&\left\{\varpi(X) \bone(\rho_U(X) > \eta_+^S\geq \rho_L(X)) + \bone(\rho_L(X) > \eta_+^S)\right\}\bone\left(X\in\mathcal{X}_+\right)\\
    &&+\alpha_+\left\{\varpi(X)\bone(\rho_U(X) = \eta_+^S) + (1 - \varpi(X)) \bone(\rho_L(X) = \eta_+^S)\right\}\bone\left(X\in\mathcal{X}_+\right).
\end{eqnarray*}
Thus, we have
\begin{eqnarray}
\nonumber &&\mathbb{E}\left[\pi_+^*(X)\left\{p_1(X)-p_0(X)\right\}\bone\left(X\in\mathcal{X}_+\right)\right] \\
\nonumber &=& \mathbb{E}\left[ \left\{\varpi(X) \bone(\rho_U(X) > \eta_+^S\geq \rho_L(X)) + \bone(\rho_L(X) > \eta_+^S)\right\}\left\{p_1(X)-p_0(X)\right\}\bone\left(X\in\mathcal{X}_+\right) \right] \\
\nonumber && + \alpha_+ \mathbb{E}\left[ \left\{\varpi(X)\bone(\rho_U(X) = \eta_+^S) + (1 - \varpi(X)) \bone(\rho_L(X) = \eta_+^S)\right\}\left\{p_1(X)-p_0(X)\right\}\bone\left(X\in\mathcal{X}_+\right) \right] \\
\nonumber &=&Q_+(\eta_+^S)+\alpha_+ R_+(\eta_+^S)\\
\label{eqn:proof-constraint}&=& \mathbb{E}\left[\varpi(X)\left\{p_1(X)-p_0(X)\right\}\bone\left(X\in\mathcal{X}_+\right)\right],
\end{eqnarray}
where the last equality follows from the definition of $\alpha_+$.
    
    By definition of $\mathcal X_0$, we have $p_1(x)-p_0(x)=0$ for any $x\in\mathcal{X}_0$.
    Hence
    \begin{eqnarray*}
    \mathbb{E}\!\left[\pi^*(X)\left\{p_1(X)-p_0(X)\right\}\right]
    &=&
    \mathbb{E}\!\left[\pi^*_+(X)\left\{p_1(X)-p_0(X)\right\}\bone\{X\in\mathcal X_+\}\right],
    \end{eqnarray*}
    and similarly,
    \begin{eqnarray*}
    \mathbb{E}\!\left[\varpi(X)\left\{p_1(X)-p_0(X)\right\}\right]
    &=&
    \mathbb{E}\!\left[\varpi(X)\left\{p_1(X)-p_0(X)\right\}\bone\{X\in\mathcal X_+\}\right].
    \end{eqnarray*}
Combining with \eqref{eqn:proof-constraint_excess} and \eqref{eqn:proof-constraint} yields \begin{itemize}
        \item if $\eta_+^S=0$,
    \begin{eqnarray}\label{eqn:proof_ext_x+1_new}
    \mathbb{E}\!\left[\pi^*(X)\left\{p_1(X)-p_0(X)\right\}\right]
    \ \ge\
    \mathbb{E}\!\left[\varpi(X)\left\{p_1(X)-p_0(X)\right\}\right];
    \end{eqnarray}
    \item if $\eta_+^S<0$, 
    \begin{eqnarray}\label{eqn:proof_ext_x+2_new}
    \mathbb{E}\!\left[\pi^*(X)\left\{p_1(X)-p_0(X)\right\}\right]
    \ =\
    \mathbb{E}\!\left[\varpi(X)\left\{p_1(X)-p_0(X)\right\}\right].
    \end{eqnarray}
    \end{itemize}
    Thus $\pi^*(X)$ satisfies the constraint.
    
    \vspace{0.15cm}
    Next, we prove that $R_{\sup}(\pi^*,\varpi)\le R_{\sup}(\pi',\varpi)$ for any policy
$\pi'(X)$ satisfying the constraint. 

\noindent For any feasible policy $\pi'$, by the definition of $\mathcal{X}_+$ and $\mathcal{X}_0$, it can further be expressed by \begin{eqnarray*}
    \pi'(X)
    &=&
    \pi'_+(X)\bone\{X\in\mathcal X_+\}
    +
    \pi'_0(X)\bone\{X\in\mathcal X_0\},
    \end{eqnarray*}
where 
\begin{eqnarray*}
    \mathbb{E}\left[\pi'_+(X)\left\{p_1(X)-p_0(X)\right\}\bone\{X\in\mathcal X_+\}\right]&\geq&\mathbb{E}\left[\varpi(X)\left\{p_1(X)-p_0(X)\right\}\right].
\end{eqnarray*}
Thus, we have
\begin{eqnarray}
   \nonumber   R_{\sup}(\pi',\varpi)-R_{\sup}(\pi^*,\varpi)  =T_1+T_2,\end{eqnarray}
where \begin{eqnarray}\nonumber T_1&=&\mathbb{E}\left[\left\{\pi_0^*(X)-\varpi(X) \right\} \left\{ A_U(X) \bone(\pi_0^*(X) < \varpi(X)) + A_L(X) \bone(\pi_0^*(X) \ge \varpi(X)) \right\}\bone\{X\in\mathcal X_0\}\right]\\
&&\nonumber -\mathbb{E}\left[\left\{\pi_0'(X)-\varpi(X) \right\} \left\{ A_U(X) \bone(\pi_0'(X) < \varpi(X)) + A_L(X) \bone(\pi_0'(X) \ge \varpi(X)) \right\}\bone\{X\in\mathcal X_0\}\right],\\
T_2&=&\nonumber\mathbb{E}\left[\left\{ \pi_+^*(X)-\varpi(X)\right\} \left\{ A_U(X) \bone(\pi_+^*(X) < \varpi(X)) + A_L(X) \bone(\pi_+^*(X) \ge \varpi(X)) \right\}\bone\{X\in\mathcal X_+\}\right]\\
&&\nonumber -\mathbb{E}\left[\left\{\pi_+'(X)-\varpi(X)  \right\} \left\{ A_U(X) \bone(\pi_+'(X) < \varpi(X)) + A_L(X) \bone(\pi_+'(X) \ge \varpi(X)) \right\}\bone\{X\in\mathcal X_+\}\right].
\end{eqnarray}

First, consider term $T_1$. From Lemma~\ref{lem:pointwise_x0}, for all $x\in\mathcal X_0$,
$\pi_0^*(x)$ maximizes $h(\pi)$ over $\pi\in[0,1]$.
Therefore, we have $
h(\pi_0^*(x)) \ge h(\pi_0'(x))$ for all $x\in\mathcal X_0$, which implies $\mathbb{E}\left\{h(\pi_0^*(X))\bone(X\in\mathcal X_0)\right\}\geq \mathbb{E}\left\{h(\pi_0'(X))\bone(X\in\mathcal X_0)\right\}$. As a result, $T_1\ge 0$.

Then, consider term $T_2$. From Lemma~\ref{lem:pointwise_x+}, for all $x\in\mathcal X_+$,
$\pi_+^*(x)$ maximizes $g(\pi)$ over $\pi\in[0,1]$. Therefore, we have $
g(\pi_+^*(x)) \ge g(\pi_+'(x))$ for all $x\in\mathcal X_+$, which implies $\mathbb{E}\left\{g(\pi_+^*(X))\bone(X\in\mathcal X_+)\right\}\geq \mathbb{E}\left\{g(\pi_+'(X))\bone(X\in\mathcal X_+)\right\}$.
As a result,
\begin{eqnarray}
\label{eq:compare_ge0_new}
\nonumber T_2&\geq& \eta_+^S\mathbb{E}\left[\left\{\pi_+^*(X)-\varpi(X)\right\}\left\{p_1(X)-p_0(X)\right\}\bone(X\in\mathcal{X}_1)\right]\\
\nonumber &&-\eta_+^S\mathbb{E}\left[\left\{\pi_+'(X)-\varpi(X)\right\}\left\{p_1(X)-p_0(X)\right\}\bone(X\in\mathcal{X}_1)\right]\\
\nonumber&=&\
\eta_+^S
\mathbb{E}\!\left[\{\pi_+^*(X)-\pi'_+(X)\}\left\{p_1(X)-p_0(X)\right\}\bone(X\in\mathcal{X}_1)\right]\\
\nonumber&=&\
\eta_+^S\mathbb{E}\!\left[\{\pi_+^*(X)-\pi'_+(X)\}\left\{p_1(X)-p_0(X)\right\}\bone(X\in\mathcal{X}_1)\right]\\
\nonumber&&+\eta_+^S
\mathbb{E}\!\left[\{\pi_0^*(X)-\pi'_0(X)\}\left\{p_1(X)-p_0(X)\right\}\bone(X\in\mathcal{X}_0)\right]\\
\nonumber&=&\ \eta_+^S
\mathbb{E}\!\left[\{\pi^*(X)-\pi'(X)\}\left\{p_1(X)-p_0(X)\right\}\right]\\
&\geq&\ \eta_+^S
\mathbb{E}\!\left[\{\pi^*(X)-\varpi(X)\}\left\{p_1(X)-p_0(X)\right\}\right].
\end{eqnarray}
where the second equality follows from the definition of $\mathcal{X}_0$, the third equality from the definition of $\pi^*(x)$ and $\pi'(x)$, and the last inequality from $\eta^S_+\leq0$ and $B(\pi')\geq B(\varpi)$. If $\eta_+^S=0$, then the right hand side of~\eqref{eq:compare_ge0_new} equals zero; if $\eta_+^S<0$, the the right hand side of~\eqref{eq:compare_ge0_new} equals zero by~\eqref{eqn:proof_ext_x+2_new}. In both cases, we have $T_2\geq 0$. Combining this with the result $T_1\geq 0$ yields the result. \QEDB 

\subsection{Proof of Theorem~\ref{theo:pi_opt}}

The result follows immediately from Theorem~\ref{theo:pi_opt_ge0}.\QEDB

\subsection{Proof of Proposition~\ref{lem:consistent}}
\subsubsection{Lemmas}
We introduce some lemmas to simplify the proof.
\begin{lemma}\label{lem:indicator_diff}
For any $a,b \in \mathbb{R}$,
\begin{eqnarray*}
    |\mathbf{1}(a > 0) - \mathbf{1}(b > 0)| \leq \mathbf{1}(\max\{|a|,|b|\} \leq |a - b|).
\end{eqnarray*}
\end{lemma}
{\it Proof of Lemma~\ref{lem:indicator_diff}. }
If $|\mathbf{1}(a > 0) - \mathbf{1}(b > 0)| = 0$, the inequality holds trivially. Otherwise, if $a > 0$ and $b \leq 0$, then $|a| = a \leq a + |b| = a-b = |b - a|$. If $a \leq 0$ and $b > 0$, then $|a| = -a \leq b - a = |b - a|$. 
Symmetrically, $|b| \leq |a - b|$ in these two cases. Therefore, $\max\{|a|,|b|\} \leq |a- b|$ when $|\mathbf{1}(a > 0) - \mathbf{1}(b > 0)| = 1$.\QEDB

\begin{lemma}\label{lem:indicator_convergence}
Let  $T(\eta)= \mathbb{E}\left[\bone\left\{\xi_1(V)>\eta\right\}\Gamma_1(V)+\bone\left\{\xi_2(V)>\eta\right\}\Gamma_2(V)\right]$, where $0\leq \Gamma_1(V)\leq M$ and $0\leq \Gamma_2(V)\leq M$ almost surely for some constant $M < \infty$.
If $\|\hat{\xi}_1 - \xi_1\|_1 = o_{\mathbb{P}}(1)$, $\|\hat{\xi}_2 - \xi_2\|_1 = o_{\mathbb{P}}(1)$, and $T(\eta)$ is continuous at $\eta_0$, then
\begin{eqnarray}\label{eqn:fixthreshold}
    \nonumber&&\mathbb{E}\left[\left\{\bone(\hat{\xi}_1(V) > \eta_0) - \bone(\xi_1(V) > \eta_0)\right\}\Gamma_1(V)+\left\{\bone(\hat{\xi}_2(V) > \eta_0) - \bone(\xi_2(V) > \eta_0)\right\}\Gamma_2(V)\right]\\&=&o_{\mathbb{P}}(1).
\end{eqnarray} Moreover, for any sequence $\hat{\eta}$ satisfying $\hat{\eta}\overset{p}{\to}\eta_0$,
\begin{eqnarray}\label{eqn:randomthreshold}
\nonumber &&\mathbb{E}\left[\left\{\bone(\hat{\xi}_1(V)>\hat{\eta})-\bone(\xi_1(V)>\hat{\eta})\right\}\Gamma_1(V)+\left\{\bone(\hat{\xi}_2(V)>\hat{\eta})-\bone(\xi_2(V)>\hat{\eta})\right\}\Gamma_2(V)\right]\\&=&o_{\mathbb P}(1).
\end{eqnarray}
\end{lemma}

\noindent {\it Proof of Lemma~\ref{lem:indicator_convergence}.}
We only prove \eqref{eqn:randomthreshold}; \eqref{eqn:fixthreshold} follows directly from \eqref{eqn:randomthreshold} by taking $\hat\eta\equiv \eta_0$. For ease of exposition, let 
\begin{eqnarray*}
R_n&=&\left|\mathbb{E}\left[\left\{\bone(\hat{\xi}_1(V)>\hat{\eta})-\bone(\xi_1(V)>\hat{\eta})\right\}\Gamma_1(V)+
\left\{\bone(\hat{\xi}_2(V)>\hat{\eta})-\bone(\xi_2(V)>\hat{\eta})\right\}\Gamma_2(V)
\right]\right|.
\end{eqnarray*}
We have
\begin{eqnarray*}
R_n
&\leq&\left|\mathbb{E}\left[\left\{\bone(\hat{\xi}_1(V)>\hat\eta)-\bone(\xi_1(V)>\hat\eta)\right\}\Gamma_1(V)\right]\right|+\left|\mathbb{E}\left[\left\{\bone(\hat{\xi}_2(V)>\hat\eta)-\bone(\xi_2(V)>\hat\eta)\right\}\Gamma_2(V)\right]\right|\\
&\leq&\mathbb{E}\left[\left|\bone(\hat{\xi}_1(V)>\hat\eta)-\bone(\xi_1(V)>\hat\eta)\right|\Gamma_1(V)\right]+\mathbb{E}\left[\left|\bone(\hat{\xi}_2(V)>\hat\eta)-\bone(\xi_2(V)>\hat\eta)\right|\Gamma_2(V)
\right]\\
&\leq&\mathbb{E}\left[\bone\left\{|\xi_1(V)-\hat\eta|\leq |\hat{\xi}_1(V)-\xi_1(V)|\right\}\Gamma_1(V)\right]+
\mathbb{E}\left[\bone\left\{|\xi_2(V)-\hat\eta|\leq |\hat{\xi}_2(V)-\xi_2(V)|\right\}\Gamma_2(V)\right],
\end{eqnarray*}
where the last inequality follows from Lemma~\ref{lem:indicator_diff}. Thus, for any $t>0$,
\begin{eqnarray*}
R_n
&\leq&
\mathbb{E}\left\{
\bone\left(|\hat{\xi}_1(V)-\xi_1(V)|>t\right)\Gamma_1(V)
\right\}
+
\mathbb{E}\left\{
\bone\left(|\hat{\xi}_2(V)-\xi_2(V)|>t\right)\Gamma_2(V)
\right\}\\
&&+
\mathbb{E}\left\{
\bone\left(|\xi_1(V)-\hat\eta|\leq t\right)\Gamma_1(V)
+
\bone\left(|\xi_2(V)-\hat\eta|\leq t\right)\Gamma_2(V)
\right\}\\
&\leq&
\frac{M}{t}\|\hat{\xi}_1-\xi_1\|_1
+
\frac{M}{t}\|\hat{\xi}_2-\xi_2\|_1\\
&&+
\mathbb{E}\left\{
\bone\left(|\xi_1(V)-\hat\eta|\leq t\right)\Gamma_1(V)
+
\bone\left(|\xi_2(V)-\hat\eta|\leq t\right)\Gamma_2(V)
\right\},
\end{eqnarray*}
where the last inequality follows from the boundedness of $\Gamma_1(V)$ and $\Gamma_2(V)$ together with Markov's inequality.

Given $s>0$, define $A_n(s)=\left\{|\hat\eta-\eta_0|\leq s\right\}$. On the event $A_n(s)$, we have $
\left\{|\xi_j(V)-\hat\eta|\leq t\right\}\subseteq
\left\{|\xi_j(V)-\eta_0|\leq t+s\right\}$ for $j=1,2$.
Therefore, on $A_n(s)$,
\begin{eqnarray*}
R_n
&\leq&\frac{M}{t}\|\hat{\xi}_1-\xi_1\|_1+
\frac{M}{t}\|\hat{\xi}_2-\xi_2\|_1
\\&&+
\mathbb{E}\left\{
\bone\left(|\xi_1(V)-\eta_0|\leq t+s\right)\Gamma_1(V)
+
\bone\left(|\xi_2(V)-\eta_0|\leq t+s\right)\Gamma_2(V)
\right\}\\
&\leq&
\frac{M}{t}\|\hat{\xi}_1-\xi_1\|_1
+
\frac{M}{t}\|\hat{\xi}_2-\xi_2\|_1
+ T\{\eta_0-(t+s)\}-T\{\eta_0+(t+s)\}.
\end{eqnarray*}
 
For any fix $\epsilon>0$. By the continuity of $h$ at $\eta_0$, there exist $t_0>0$ and $s_0>0$ such that
$0\leq T\{\eta_0-(t_0+s_0)\}-T\{\eta_0+(t_0+s_0)\}<\epsilon/2$. Then
\begin{eqnarray*}
\mathbb{P}(R_n>\epsilon)
&\leq&
\mathbb{P}(R_n>\epsilon,\ A_n(s_0))
+
\mathbb{P}\{A_n(s_0)^c\}\\
&\leq&
\mathbb{P}\left(
\frac{M}{t_0}\|\hat{\xi}_1-\xi_1\|_1
+
\frac{M}{t_0}\|\hat{\xi}_2-\xi_2\|_1
>
\epsilon/2
\right)
+
\mathbb{P}\left(|\hat\eta-\eta_0|>s_0\right)\\
&\to& 0,
\end{eqnarray*}
where the last line follows from the consistency of $\hat\xi_1$, $\hat\xi_2$ and $\hat{\eta}$. Thus, we obtain $R_n=o_{\mathbb{P}}(1)$.
\QEDB

\begin{lemma}
\label{lem:effkappa}
Suppose that Conditions (a), (b), and (c) in Theorem~\ref{thm:convergence_margin} and Assumption~\ref{assum:rates}(a) hold,
    \begin{eqnarray*}
        \hat{C}_{\textup{eif}}-\mathbb{E}\left[\varpi(X)\left\{p_1(X)-p_0(X)\right\}\right]=(\mathbb{P}_n-\mathbb{P})\varpi(X)\left\{\psi_{p_1}-\psi_{p_0}\right\}+o_{\mathbb{P}}(n^{-1/2}).
    \end{eqnarray*}
\end{lemma}
\noindent {\it Proof of Lemma~\ref{lem:effkappa}.}
We can write
\begin{eqnarray}
    \hat{C}_{\textup{eif}}-\mathbb{E}\left[\varpi(X)\left\{p_1(X)-p_0(X)\right\}\right]&=&
\nonumber\mathbb{P}_n\left[\varpi(X)\left\{\hat{\psi}_{p_1}-\hat{\psi}_{p_0}\right\}\right] - \mathbb{P}\left[\varpi(X)\left\{\psi_{p_1}-\psi_{p_0}\right\}\right] \\
&=& R_1 +R_2 + (\Pn - \Pb)\left[\varpi(X)\left\{\psi_{p_1}-\psi_{p_0}\right\}\right],\label{eqn:contermr}
\end{eqnarray}
where
\begin{eqnarray*}
    R_1 &=&(\mathbb{P}_n - \mathbb{P})\left[\varpi(X)\left\{\hat{\psi}_{p_1}-\hat{\psi}_{p_0}\right\}- \varpi(X)\left\{\psi_{p_1}-\psi_{p_0}\right\}\right],\\
    R_2 &=& \mathbb{P}\left[\varpi(X)\left\{\hat{\psi}_{p_1}-\hat{\psi}_{p_0}-\psi_{p_1}+\psi_{p_0}\right\}\right].
\end{eqnarray*}
From Conditions (a), (b), and (c) in Theorem~\ref{thm:convergence_margin} together with Lemma~2 of \citet{kennedy2020sharp}, we have $R_1 = o_{\Pb}(n^{-1/2})$. For $R_2$,
from Lemma~\ref{lem:bias_decomp}, we have 
\begin{eqnarray*}
\nonumber R_{2}&=& \Pb\left[\varpi(X)\left\{(\hat{\psi}_{p_1}-\hat{\psi}_{p_0})-(p_1(X)-p_0(X)\right\} \right]\\
\nonumber &=&\mathbb{P}\left[\varpi(X)\left\{\hat{p}_1(X)-p_1(X)\right\}\left\{\frac{Z}{e(X)}-\frac{Z}{\hat{e}(X)}\right\}\right]\\&&-\mathbb{P}\left[\varpi(X)\left\{\hat{p}_0(X)-p_0(X)\right\}\left\{\frac{1-Z}{1-e(X)}-\frac{1-Z}{1-\hat{e}(X)}\right\}\right].
\end{eqnarray*} 
By the Cauchy-Schwartz inequality, we have 
\begin{eqnarray}
\label{eqn:contermr2} R_{2}&=&O_{\mathbb{P}}\left(\left\|\hat{\lambda}_1-\lambda_1\right\|_2\left\|\hat{p}_1-p_1\right\|_2+\left\|\hat{\lambda}_0-\lambda_0\right\|_2\left\|\hat{p}_0-p_0\right\|_2\right).
\end{eqnarray}
Plugging \eqref{eqn:contermr2} and $R_1=o_{\mathbb{P}}(n^{-1/2})$ into \eqref{eqn:contermr} completes the proof. \QEDB

\begin{lemma}\label{lem:op_criterion}
Let $X_n$ be a sequence and $a$ a nonnegative constant. Suppose that, for every $\epsilon>0$, there exists a nonnegative sequence $Z_n^{(\epsilon)}=o_{\mathbb P}(1)$ such that $|X_n|\leq a+\epsilon+Z_n^{(\epsilon)}$. Then, $X_n=O_\mathbb{P}(1)$ and $\max(|X_n|-a,0)=o_{\mathbb P}(1)$. In particular, if $a=0$, then $X_n=o_{\mathbb P}(1)$.
\end{lemma}

\noindent {\it Proof of Lemma~\ref{lem:op_criterion}. }
Fixing $\epsilon>0$, consider a non-negative sequence
$Z_n^{(\epsilon/2)}=o_{\mathbb P}(1)$ satisfying $
|X_n|\leq a+\epsilon/2+Z_n^{(\epsilon/2)}$.
Then,
\begin{eqnarray*}
\mathbb P\left\{|X_n|>a+\epsilon\right\}
&\leq&
\mathbb P\left\{a+\epsilon/2+Z_n^{(\epsilon/2)}>a+\epsilon\right\} \\
&=&
\mathbb P\left\{Z_n^{(\epsilon/2)}>\epsilon/2\right\}\to 0
\quad \text{as } n\to\infty,
\end{eqnarray*}
since $Z_n^{(\epsilon/2)}=o_{\mathbb P}(1)$. Taking $\epsilon=1$ yields $X_n=O_\mathbb{P}(1)$. Since $\left\{\max(|X_n|-a,0)>\epsilon\right\}=\left\{|X_n|>a+\epsilon\right\}$, we have
\begin{eqnarray*}
    \mathbb P \left\{\max(|X_n|-a,0)>\epsilon\right\}\to0
\end{eqnarray*}
for any $\epsilon>0$, which yields $\max(|X_n|-a,0)=o_{\mathbb P}(1)$. When $a=0$, we have $\max(|X_n|-a,0)=|X_n|$, and hence $X_n=o_{\mathbb P}(1)$.

\QEDB

\begin{lemma}\label{lem:auxconsisemp1}
    Suppose that Assumptions~\ref{assum:margincondition1}(a) and Conditions (a), (b), and (c) in Theorem~\ref{thm:convergence_margin} hold, and $p_1(x)-p_0(x)>0$ for all  $x\in\mathcal{X}$.  Then, for any fixed $\eta$, \begin{eqnarray}
        (\mathbb{P}_n-\mathbb{P})\left[\varpi(X)\left\{\mathbf{1}\left(\hat{\rho}_{U}(X)>\eta\right)-\mathbf{1}\left(\rho_{U}(X)>\eta\right)\right\}\left(\psi_{p_1}-\psi_{p_0}\right)\right]&=&O_{\mathbb{P}}(n^{-1/2}),\label{eqn:auxconsisemp1_U}\\
(\mathbb{P}_n-\mathbb{P})\left[\left\{1-\varpi(X)\right\}\left\{\mathbf{1}\left(\hat{\rho}_{L}(X)>\eta\right)-\mathbf{1}\left(\rho_{L}(X)>\eta\right)\right\}\left(\psi_{p_1}-\psi_{p_0}\right)\right]&=&O_{\mathbb{P}}(n^{-1/2})\label{eqn:auxconsisemp1_L}.
    \end{eqnarray}
Moreover, when $\eta=\eta^S$,
\begin{eqnarray}
        (\mathbb{P}_n-\mathbb{P})\left[\varpi(X)\left\{\mathbf{1}\left(\hat{\rho}_{U}(X)>\eta^S\right)-\mathbf{1}\left(\rho_{U}(X)>\eta^S\right)\right\}\left(\psi_{p_1}-\psi_{p_0}\right)\right]&=&o_{\mathbb{P}}(n^{-1/2}),\label{eqn:auxconsisemp1_U_etas}\\
(\mathbb{P}_n-\mathbb{P})\left[\left\{1-\varpi(X)\right\}\left\{\mathbf{1}\left(\hat{\rho}_{L}(X)>\eta^S\right)-\mathbf{1}\left(\rho_{L}(X)>\eta^S\right)\right\}\left(\psi_{p_1}-\psi_{p_0}\right)\right]&=&o_{\mathbb{P}}(n^{-1/2})\label{eqn:auxconsisemp1_L_etas}.
    \end{eqnarray}
\end{lemma}

\noindent {\it Proof of Lemma~\ref{lem:auxconsisemp1}. } We only prove \eqref{eqn:auxconsisemp1_U} and \eqref{eqn:auxconsisemp1_U_etas}, the other two follow similarly. 

By the boundedness of $e(X)$, $p_1(X)$, $p_0(X)$ and $\varpi(X)$, there exists a sufficiently large constant $C_1<\infty$ such that
\begin{eqnarray*}
&& \| \varpi(X)\left\{\mathbf{1}(\hat{\rho}_U(X)>\eta)-\mathbf{1}(\rho_U(X)>\eta)\right\}\left\{\psi_{p_1}-\psi_{p_0}\right\}\|_2^2  \\
&\leq& C_1\mathbb{E} \left( \left| \mathbf{1}(\hat{\rho}_U(X)>\eta)-\mathbf{1}(\rho_U(X)>\eta) \right| \right)\\
&\leq & C_1\mathbb{P} \left( |\rho_U(X) - \eta|\leq |\hat{\rho}_U(X) - \rho_U(X)| \right)\\
&\leq & C_1\mathbb{P} \left( |\rho_U(X) - \eta^S|\leq |\hat{\rho}_U(X) - \rho_U(X)|+|\eta^S-\eta| \right),
\end{eqnarray*}
where the second inequality follows from Lemma~\ref{lem:indicator_diff} and the last inequality follows from triangle inequality. Then, for any $t>0$,
\begin{eqnarray*}
&&\| \varpi(X)\left\{\mathbf{1}(\hat{\rho}_U(X)>\eta)-\mathbf{1}(\rho_U(X)>\eta)\right\}\left\{\psi_{p_1}-\psi_{p_0}\right\}\|_2^2\\
&\leq& C_1\mathbb{P} \left( |\rho_U(X) - \eta^S|-|\eta^S-\eta| \leq t \right) + C_1\mathbb{P} \left( |\hat{\rho}_U(X) - \rho_U(X)| > t \right)  \\
&\leq & CC_1 (t+|\eta^S-\eta|)^\beta + C_1\frac{\|\hat{\rho}_U - \rho_U\|_2}{t},
\end{eqnarray*}
where the last line follows from Assumption~\ref{assum:margincondition1}(a) and Markov's inequality. For any $\epsilon>0$, choose
\begin{eqnarray*}
t_\epsilon=\frac{1}{2}\left[\left\{|\eta-\eta^S|^\beta+\frac{\epsilon}{CC_1}
\right\}^{1/\beta}-|\eta-\eta^S|\right],
\end{eqnarray*}
and define $Z_n^{(\epsilon)}=C_1\|\hat\rho_U-\rho_U\|/t_{\epsilon}$ and $a = CC_1|\eta-\eta^S|^\beta$.
Then 
\begin{eqnarray*}
\|\varpi(X)\left\{\mathbf{1}(\hat{\rho}_U(X)>\eta)-\mathbf{1}(\rho_U(X)>\eta)\right\}\left\{\psi_{p_1}-\psi_{p_0}\right\}\|_2^2
&\leq&
CC_1\left(t_\epsilon+|\eta-\eta^S|\right)^\beta+C_1\frac{\|\hat\rho_U-\rho_U\|}{t_\epsilon}\\
&\leq&CC_1|\eta-\eta^S|^\beta+\epsilon+C_1\frac{\|\hat\rho_U-\rho_U\|}{t_\epsilon}\\
&=&a+\epsilon+Z_n^{(\epsilon)}.
\end{eqnarray*} 
Under Condition (b) in Theorem~\ref{thm:convergence_margin}, $Z_n^{(\epsilon)}=o_{\mathbb P}(1)$.
By Lemma~\ref{lem:op_criterion}, we obtain that
\begin{eqnarray*}
\|\varpi(X)\left\{\mathbf{1}(\hat{\rho}_U(X)>\eta)-\mathbf{1}(\rho_U(X)>\eta)\right\}\left\{\psi_{p_1}-\psi_{p_0}\right\}\|_2^2
&=&O_{\mathbb P}(1).
\end{eqnarray*}
Applying Lemma~2 of \citet{kennedy2020sharp} proves \eqref{eqn:auxconsisemp1_U}. 

When $\eta=\eta^S$, the preceding argument applies with $a=CC_1|\eta^S-\eta^S|^\beta=0$. Hence, by Lemma~\ref{lem:op_criterion}, we obtain that 
\begin{eqnarray*}
\|\varpi(X)\left\{\mathbf{1}(\hat{\rho}_U(X)>\eta^S)-\mathbf{1}(\rho_U(X)>\eta^S)\right\}\left\{\psi_{p_1}-\psi_{p_0}\right\}\|_2^2
&=&o_{\mathbb P}(1).
\end{eqnarray*}
Applying Lemma~2 of \citet{kennedy2020sharp} proves \eqref{eqn:auxconsisemp1_U_etas}. \QEDB

\subsubsection{Uniform convergence of $\hat{Q}(\cdot)$}
We now prove that $\hat{Q}(\eta)-Q(\eta)=o_{\mathbb{P}}(1)$ uniformly for $\eta$ in a neighborhood of $\eta^E$. For a fixed $\eta$ in a neighborhood of $\eta^E$, we have
\begin{eqnarray}
    \nonumber&&|\hat{Q}(\eta)-Q(\eta)|\\
       \nonumber &\leq&\bigg|\Pn\!\left[ \varpi(X)\mathbf{1}\{\hat{\rho}_U(X)>\eta\} \left\{\hat{\psi}_{p_1}-\hat{\psi}_{p_0}\right\}\right]-\Pb\left[\varpi(X)\mathbf{1}\{\rho_U(X)>\eta\} \{p_1(X)-p_0(X)\} \right]\\
        \nonumber&&+\Pn\left[ \{1-\varpi(X)\}\mathbf{1}\{\hat{\rho}_L(X)>\eta\} \left\{\hat{\psi}_{p_1}-\hat{\psi}_{p_0}\right\}\right]\\
        \nonumber&&-\Pb\left[\{1-\varpi(X)\}\mathbf{1}\{\rho_L(X)>\eta\} \{p_1(X)-p_0(X)\} \right]\bigg|\\
        &\leq& R_1+R_2+R_3+R_4\label{eqn:cantellidecomp},
%
\end{eqnarray}
where
\begin{eqnarray*}
R_1 &=&\left|\mathbb{P}\left[\varpi(X)\left\{\mathbf{1}\left(\hat{\rho}_{U}(X)>\eta\right)-\mathbf{1}\left(\rho_{U}(X)>\eta\right)\right\}\{p_1(X)-p_0(X)\}\right]\right.\\
&&+\left.\mathbb{P}\left[\{1-\varpi(X)\}\left\{\mathbf{1}\left(\hat{\rho}_{L}(X)>\eta\right)-\mathbf{1}\left(\rho_{L}(X)>\eta\right)\right\}\{p_1(X)-p_0(X)\}\right]\right|,\\
R_2&=&\left|\mathbb{P}\left[\varpi(X)\mathbf{1}\left\{\hat{\rho}_{U}(X)>\eta\right\}\left\{(\hat{\psi}_{p_1}(X)-\hat{\psi}_{p_0}(X))-(p_1(X)-p_0(X))\right\}\right]\right|\\
&&+ \left|\mathbb{P}\left[\{1-\varpi(X)\}\mathbf{1}\left\{\hat{\rho}_{L}(X)>\eta\right\}\left\{(\hat{\psi}_{p_1}(X)-\hat{\psi}_{p_0}(X))-(p_1(X)-p_0(X))\right\}\right]\right|,\\
R_3&=&\left|(\mathbb{P}_n-\mathbb{P})\left[\left\{\varpi(X)\mathbf{1}\left(\rho_{U}(X)>\eta\right)+\left\{1-\varpi(X)\right\}\mathbf{1}\left(\rho_{L}(X)>\eta\right)\right\}\left\{\psi_{p_1}-\psi_{p_0}\right\}\right]\right|,\\
R_4&=&\left|(\mathbb{P}_n-\mathbb{P})\left[\varpi(X)\left\{\mathbf{1}\left(\hat{\rho}_{U}(X)>\eta\right)\left(\hat{\psi}_{p_1}-\hat{\psi}_{p_0}\right)-\mathbf{1}\left(\rho_{U}(X)>\eta\right)\left(\psi_{p_1}-\psi_{p_0}\right)\right\}\right]\right.\\
&&+\left.(\mathbb{P}_n-\mathbb{P})\left[\left\{1-\varpi(X)\right\}\left\{\mathbf{1}\left(\hat{\rho}_{L}(X)>\eta\right)\left(\hat{\psi}_{p_1}-\hat{\psi}_{p_0}\right)-\mathbf{1}\left(\rho_{L}(X)>\eta\right)\left(\psi_{p_1}-\psi_{p_0}\}\right)\right\}\right]\right|.
\end{eqnarray*}

The first term $R_1$ is $o_{\Pb}(1)$ from Condition (b) in Theorem~\ref{thm:convergence_margin} and Assumption~\ref{assum:c311}, and \eqref{eqn:fixthreshold}. 

For the second term $R_2$, from \eqref{eq:bias_pz} and the Cauchy-Schwartz inequality, we have 
\begin{eqnarray*}
  R_2  & = & O_{\mathbb{P}}\left(\left\|\hat{\lambda}_0-\lambda_0\right\|_2 \left\|\hat{p}_1-p_1\right\|_2+\left\|\hat{\lambda}_1-\lambda_1\right\|_2\left\|\hat{p}_0-p_0\right\|_2\right) =\ o_{\mathbb{P}}(n^{-1/2}),
\end{eqnarray*}
where the last equality follows from Assumption~\ref{assum:rates}(a). 

The third term $R_3$ is $O_{\Pb}(n^{-1/2})$ from central limit theorem.

For the forth term $R_4$, with the boundedness of $\varpi(X)$, $\mathbf{1}\left(\hat{\rho}_{U}(X)>\eta\right)$ and $\mathbf{1}\left(\hat{\rho}_{L}(X)>\eta\right)$, we can write \begin{eqnarray*}
   R_4
    &=&\bigg|(\mathbb{P}_n-\mathbb{P})\left[\varpi(X)\mathbf{1}\left(\hat{\rho}_{U}(X)>\eta\right)\left\{\left(\hat{\psi}_{p_1}-\hat{\psi}_{p_0}\right)-\left(\psi_{p_1}-\psi_{p_0}\right)\right\}\right]\nonumber\\
    &&+(\mathbb{P}_n-\mathbb{P})\left[\varpi(X)\left\{\mathbf{1}\left(\hat{\rho}_{U}(X)>\eta\right)-\mathbf{1}\left(\rho_{U}(X)>\eta\right)\right\}\left(\psi_{p_1}-\psi_{p_0}\right)\right]\nonumber\\
    &&+(\mathbb{P}_n-\mathbb{P})\left[\left\{1-\varpi(X)\right\}\mathbf{1}\left(\hat{\rho}_{L}(X)>\eta\right)\left\{\left(\hat{\psi}_{p_1}-\hat{\psi}_{p_0}\right)-\left(\psi_{p_1}-\psi_{p_0}\right)\right\}\right]\nonumber\\
    &&+(\mathbb{P}_n-\mathbb{P})\left[\left\{1-\varpi(X)\right\}\left\{\mathbf{1}\left(\hat{\rho}_{L}(X)>\eta\right)-\mathbf{1}\left(\rho_{L}(X)>\eta\right)\right\}\left(\psi_{p_1}-\psi_{p_0}\right)\right]\bigg|\nonumber\\
    &\leq& T_1+T_2+T_3,
\end{eqnarray*}
where \begin{eqnarray*}
    T_1&=&2\left|(\mathbb{P}_n-\mathbb{P})\left\{\left(\hat{\psi}_{p_1}-\hat{\psi}_{p_0}\right)-\left(\psi_{p_1}-\psi_{p_0}\right)\right\}\right|,\\
    T_2&=&\left|(\mathbb{P}_n-\mathbb{P})\left[\varpi(X)\left\{\mathbf{1}\left(\hat{\rho}_{U}(X)>\eta\right)-\mathbf{1}\left(\rho_{U}(X)>\eta\right)\right\}\left(\psi_{p_1}-\psi_{p_0}\right)\right]\right|,\\
    T_3&=&\left|(\mathbb{P}_n-\mathbb{P})\left[\left\{1-\varpi(X)\right\}\left\{\mathbf{1}\left(\hat{\rho}_{L}(X)>\eta\right)-\mathbf{1}\left(\rho_{L}(X)>\eta\right)\right\}\left(\psi_{p_1}-\psi_{p_0}\right)\right]\right|.
\end{eqnarray*}
The first term $T_1$ is $o_{\mathbb{P}}(n^{-1/2})$ from Conditions (a), (b), and (c) in Theorem~\ref{thm:convergence_margin} together with Lemma~2 of \citet{kennedy2020sharp}. The other two terms $T_2$ and $T_3$ are $O_{\mathbb{P}}(n^{-1/2})$ from \eqref{eqn:auxconsisemp1_U} and \eqref{eqn:auxconsisemp1_L}. Thus, we have $R
_4=O_\mathbb{P}(n^{-1/2})$.

Plugging $R_1=o_\Pb(1)$, $R_2=o_{\mathbb{P}}(n^{-1/2})$, $R_3=O_{\mathbb{P}}(n^{-1/2})$ and $R_4=O_{\mathbb{P}}(n^{-1/2})$ into \eqref{eqn:cantellidecomp}, we have $\hat{Q}(\eta)-Q(\eta)=o_{\mathbb{P}}(1)$.

\subsubsection{Consistency of $\hat{\eta}^S$}

The uniform consistency of $\hat{Q}(\cdot)$ implies
\begin{eqnarray*}
    \hat{Q}(\eta^E - \epsilon) &=& Q(\eta^E - \epsilon) + o_{\mathbb{P}}(1), \\
    \hat{Q}(\eta^E + \epsilon) &=& Q(\eta^E + \epsilon) + o_{\mathbb{P}}(1),
\end{eqnarray*}
for a sufficiently small $\epsilon>0$. We have $\hat{C}_{\textnormal{eif}}=\mathbb{E}\left[\varpi(X)\left\{p_1(X)-p_0(X)\right\}\right]+o_{\mathbb{P}}(1)$ by Lemma~\ref{lem:effkappa}. Assumption~\ref{assum:c311} implies that $Q(\eta^E - \epsilon) > \mathbb{E}\left[\varpi(X)\left\{p_1(X)-p_0(X)\right\}\right] > Q(\eta^E + \epsilon)$.
Therefore, we have 
\begin{eqnarray*}
    \hat{Q}(\eta^E - \epsilon) >\hat{C}_{\textnormal{eif}}> \hat{Q}(\eta^E + \epsilon).
\end{eqnarray*}
with probability tending to one.
By the definition of $\hat{C}_{\textnormal{eif}}$, we obtain 
\begin{eqnarray*}
    \eta^E - \epsilon \leq \hat{\eta}^E \leq \eta^E + \epsilon
\end{eqnarray*}
with probability tending to one. Letting $\epsilon\rightarrow0$, we have $\hat{\eta}^E \xrightarrow{p} \eta^E$. By continuous mapping theorem, we obtain $\hat{\eta}^S \xrightarrow{p} \eta^S$. \QEDB

\subsection{Proof of Theorem~\ref{thm:constraintviolation}}
\subsubsection{Lemmas}
We introduce some lemmas to simplify the proof.   
Denote 
\begin{eqnarray*}
\pi^*_{f}(X) &=& (1 - \varpi(X)) \bone(\rho_{L }(X) > \eta ^S) + \varpi(X) \bone(\rho_{U }(X) > \eta ^S),
\end{eqnarray*}
which represents a simplified version of the optimal policy $\pi^*(X)$ that excludes the boundary cases  $\rho_L(X)=\eta^S$ and $\rho_U(X)=\eta^S$. Similarly, 
denote 
\begin{eqnarray*}
\hat{\pi}^*_{f}(X)&=& (1 - \varpi(X)) \bone(\hat{\rho}_{L }(X) > \hat{\eta} ^S) + \varpi(X) \bone(\hat{\rho}_{U }(X) > \hat{\eta} ^S),
\end{eqnarray*}
which is the corresponding simplified version of the estimated optimal policy
$\hat{\pi}^*(X)$, excluding the boundary cases $\hat{\rho}_L(X)=\hat{\eta}^S$ and $\hat{\rho}_U(X)=\hat{\eta}^S$.
\begin{lemma}
\label{lem:pi-diff}
Under Assumption~\ref{assum:margincondition1}(a), $\mathbb{P}\big(\pi^*(X)\neq \pi_f^*(X)\big) = 0$.
\end{lemma}
\noindent {\it Proof of Lemma~\ref{lem:pi-diff}.} By the definitions of $\pi_f^*(\cdot)$ and $\pi^*(\cdot)$, the only cases where $\pi^*(X)$ may differ from $\pi_f^*(X)$ are the boundary cases
$\rho_L(X)=\eta^S$ or $\rho_U(X)=\eta^S$. Therefore, we have 
\begin{eqnarray*}
\mathbb{P}\big(\pi^*(X)\neq \pi_f^*(X)\big)
&\leq & \mathbb{P}\big(\rho_L(X)=\eta^S\big)+\mathbb{P}\big(\rho_U(X)=\eta^S\big)\ =\ 0,
\end{eqnarray*}
where the last equality follows from Assumption~\ref{assum:margincondition1}(a). \QEDB

 \begin{lemma}\label{lem:pihat-diff}
 Under Assumption~\ref{assum:margincondition1}(b), $\mathbb{P}(\hat\pi^*(X)\neq \hat\pi_f^*(X)) = O_{\mathbb P}(n^{-1/2})$.
\end{lemma}
\noindent {\it Proof of Lemma~\ref{lem:pihat-diff}. }
By the definitions of $\hat\pi_f^*(\cdot)$ and $\hat\pi^*(\cdot)$, the only cases where $\hat\pi^*(X)$ may differ from $\hat\pi_f^*(X)$ are the boundary cases
$\hat\rho_L(X)=\hat\eta^S$ or $\hat\rho_U(X)=\hat\eta^S$. Therefore, we have
\begin{eqnarray*}
\mathbb{P}\big(\hat\pi^*(X)\neq \hat\pi_f^*(X)\big)
&\le& \mathbb{P}\big(\hat\rho_L(X)=\hat\eta^S\big)+\mathbb{P}\big(\hat\rho_U(X)=\hat\eta^S\big)
\ =\ O_{\mathbb P}(n^{-1/2}).
\end{eqnarray*}
where the last equality follows from Assumption~\ref{assum:margincondition1}(b).
\QEDB

\subsubsection{Constraint slackness}
From Lemmas~\ref{lem:pi-diff}~and~\ref{lem:pihat-diff}, we have
\begin{eqnarray}
    \nonumber &&(p_1-p_0)\left\{B(\hat{\pi}^*)-B(\varpi)\right\}\\
   \nonumber &=&(p_1-p_0)\left\{B(\hat{\pi}^*)-B(\pi^*)\right\}+(p_1-p_0)\left\{B(\pi^*)-B(\varpi)\right\}\\
    \nonumber&=&\mathbb{E}\left[\left\{\hat{\pi}^*(X)-\pi^*(X)\right\}\{p_1(X)-p_0(X)\}\right]+\mathbb{E}\left[\left\{\pi^*(X)-\varpi(X)\right\}\{p_1(X)-p_0(X)\}\right]\\
    \nonumber&=&\mathbb{E}\left[\left\{\hat{\pi}_f^*(X)-\pi_f^*(X)\right\}\{p_1(X)-p_0(X)\}\right]+\mathbb{E}\left[\left\{\hat{\pi}^*(X)-\hat{\pi}_f^*(X)\right\}\{p_1(X)-p_0(X)\}\right]\\
    \nonumber&&-\mathbb{E}\left[\left\{\pi^*(X)-\pi_f^*(X)\right\}\{p_1(X)-p_0(X)\}\right]+\mathbb{E}\left[\left\{\pi^*(X)-\varpi(X)\right\}\{p_1(X)-p_0(X)\}\right]\\
    &=&R_1+R_2+o_{\mathbb P}(1),\label{eqn:con_vio_decomp}
    \end{eqnarray}
where \begin{eqnarray*}
R_1&=&\mathbb{E}\left[\left\{\hat{\pi}_f^*(X)-\pi_f^*(X)\right\}\{p_1(X)-p_0(X)\}\right],\\
R_2&=&\mathbb{E}\left[\left\{\pi^*(X)-\varpi(X)\right\}\{p_1(X)-p_0(X)\}\right].
\end{eqnarray*}

For the first term $R_1$, we have
\begin{eqnarray*}
    \nonumber R_1&=&\mathbb{P}\left[\left\{\hat{\pi}_f^*(X)-\pi_f^*(X)\right\}\left\{p_1(X)-p_0(X)\right\}\right]\\
  \nonumber &=&\left|\mathbb{P}\left[\varpi(X)\left\{\mathbf{1}\left(\hat{\rho}_{U}(X)>\hat{\eta}^S\right)-\mathbf{1}\left(\rho_{U}(X)>\eta^S\right)\right\}\{p_1(X)-p_0(X)\}\right]\right.\\
\nonumber &&+\left.\mathbb{P}\left[\{1-\varpi(X)\}\left\{\mathbf{1}\left(\hat{\rho}_{L}(X)>\hat{\eta}^S\right)-\mathbf{1}\left(\rho_{L}(X)>\eta^S\right)\right\}\{p_1(X)-p_0(X)\}\right]\right|\\
&\leq &T_1+T_2,
\end{eqnarray*}
where \begin{eqnarray*}
T_1&=&\left|\mathbb{P}\left[\varpi(X)\left\{\mathbf{1}\left(\rho_{U}(X)>\hat{\eta}^S\right)-\mathbf{1}\left(\rho_{U}(X)>\eta^S\right)\right\}\{p_1(X)-p_0(X)\}\right]\right.\\
\nonumber &&+\left.\mathbb{P}\left[\{1-\varpi(X)\}\left\{\mathbf{1}\left(\rho_{L}(X)>\hat{\eta}^S\right)-\mathbf{1}\left(\rho_{L}(X)>\eta^S\right)\right\}\{p_1(X)-p_0(X)\}\right]\right|,\\
T_2&=&\left|\mathbb{P}\left[\varpi(X)\left\{\mathbf{1}\left(\hat{\rho}_{U}(X)>\hat{\eta}^S\right)-\mathbf{1}\left(\rho_{U}(X)>\hat{\eta}^S\right)\right\}\{p_1(X)-p_0(X)\}\right]\right.\\
\nonumber &&+\left.\mathbb{P}\left[\{1-\varpi(X)\}\left\{\mathbf{1}\left(\hat{\rho}_{L}(X)>\hat{\eta}^S\right)-\mathbf{1}\left(\rho_{L}(X)>\hat{\eta}^S\right)\right\}\{p_1(X)-p_0(X)\}\right]\right|.
\end{eqnarray*}
Consider $T_1$, when $\eta^E\leq 0$, we have $\eta^S=\eta^E$. Thus, we have $T_1=o_{\Pb}(1)$ from Assumption~\ref{assum:c311} and Proposition~\ref{lem:consistent} together with continuous mapping theorem. When $\eta^E>0$, we have $\eta^S=0$. Thus, for every $\epsilon>0$, we have \begin{eqnarray*}
    \mathbb P(T_1>\epsilon)\leq\mathbb P(\hat\eta^S\neq0)=\mathbb P(\hat\eta^E<0)\to0
\end{eqnarray*} from the consistency of $\hat\eta^E$. Hence $T_1=o_\Pb(1)$. Consider $T_2$, we have $T_2=o_{\Pb}(1)$ from Proposition~\ref{lem:consistent} and \eqref{eqn:randomthreshold}.  As a result, we obtain $R_1=o_{\mathbb{P}}(1)$. 

For the second term $R_2$, when $\eta^S<0$, we have $R_2=0$ from \eqref{eqn:proof_ext_x+2_new}. When $\eta^S=0$, we have \begin{eqnarray*}
    R_2&=&Q(0)-\mathbb{E}\left[\varpi(X)\{p_1(X)-p_0(X)\}\right]
\end{eqnarray*}
from \eqref{eqn:proof-constraint_excess}.

Combining the results for $R_1$ and $R_2$ together with \eqref{eqn:con_vio_decomp} yields the result.
\QEDB
\subsection{Proof of Theorem~\ref{thm:regretbound}}
\subsubsection{Lemmas}
Denote the positive and negative parts of $x$ by 
 $\phi_{+}(x) = x \bone(x>0)$ and  $\phi_{-}(x) = x \bone(x<0)$, respectively.

\begin{lemma}\label{lem:replace_pi_by_pif}
Suppose that Assumption~\ref{assum:margincondition1}
holds.
Let $G(x)$ be any function satisfying $|G(x)|\le M$ for some constant $M<\infty$ and all $x$. Then
\begin{eqnarray}
\nonumber &&\E  \left[G(X)  \left\{\phi_{+}(\varpi(X)-\hat{\pi} ^*(X)) - \phi_{+}(\varpi(X)-\pi ^*(X))\right\} \right]\\
\label{eq:main_claim_replace_u}  &=&\E  \left[G(X)  \left\{\phi_{+}(\varpi(X)-\hat{\pi}_f ^*(X)) - \phi_{+}(\varpi(X)-\pi^*_f(X))\right\} \right] +\ O_{\mathbb P}\big(n^{-1/2}\big),\\
\nonumber &&\E  \left[G(X)  \left\{\phi_{-}(\varpi(X)-\hat{\pi} ^*(X)) - \phi_{-}(\varpi(X)-\pi ^*(X))\right\} \right]\\
\label{eq:main_claim_replace_l}  &=&\E  \left[G(X)  \left\{\phi_{-}(\varpi(X)-\hat{\pi}_f ^*(X)) - \phi_{-}(\varpi(X)-\pi^*_f(X))\right\} \right] +\ O_{\mathbb P}\big(n^{-1/2}\big).
\end{eqnarray}
\end{lemma}
\noindent {\it Proof of Lemma~\ref{lem:replace_pi_by_pif}. } We only prove the first equality; the second equality follows similarly.

By the definition of $\phi_{+}(x)$, we have 
\begin{eqnarray*}
|\phi_{+}(x)-\phi_{+}(y)|
&=&  | \phi_{+}(x)-\phi_{+}(y)| \cdot \bone(x\neq y)
\ \le \  (|\phi_{+}(x)|+|\phi_{+}(y)|)\cdot \bone(x\neq y).
\label{eq:phi_diff_bound}
\end{eqnarray*}
Therefore, we have 
\begin{eqnarray}
\nonumber &&\left|\mathbb{E}\left[G(X)\{\phi_{+}(\varpi(X)-\pi^*(X))-\phi_{+}(\varpi(X)-\pi_f^*(X))\}\right]\right|\\
&\le&
\mathbb{E}\big[|G(X)| \{|\phi_{+}(\varpi(X)-\pi^*(X))| + |\phi_{+}(\varpi(X)-\pi_f^*(X))|\} \cdot  \bone(\pi^*(X)\neq \pi_f^*(X)) \big]
\nonumber\\
&\le&
2\,\mathbb{E}\big[|G(X)|\,\bone(\pi^*(X)\neq \pi_f^*(X))\big]
\nonumber\\
&\le&
\nonumber 2M\,\mathbb{P}\big(\pi^*(X)\neq \pi_f^*(X)\big)\\
&=& 
0,
\label{eq:second_term_rate}
\end{eqnarray}
where the second inequality follows from $|\phi_{+}(\varpi(X)-\pi^*(X))|\leq 1$ and $|\phi_{+}(\varpi(X)-\pi_f^*(X))|\leq 1$, and the last equality follows from Lemma~\ref{lem:pi-diff}.
Similarly, we have 
\begin{eqnarray}
\nonumber &&\left |\mathbb{E}\big[G(X)\{\phi_{+}(\varpi(X)-\hat\pi^*(X))-\phi_{+}(\varpi(X)-\hat\pi_f^*(X))\}\big]\right|\\
\nonumber &\le&
2M\,\mathbb{P}\left(\hat\pi^*(X)\neq \hat\pi_f^*(X)\right)\\
&=& O_{\mathbb P}\left(n^{-1/2}\right),
\label{eq:first_term_rate}
\end{eqnarray}
where the last equality follows from Lemma~\ref{lem:pihat-diff}.
Combining~\eqref{eq:second_term_rate}~and~\eqref{eq:first_term_rate} yields
\begin{eqnarray*}
&&\left|\E  \left[G(X)  \left\{\phi_{+}(\varpi(X)-\hat{\pi} ^*(X)) - \phi_{+}(\varpi(X)-\pi ^*(X))\right\} \right] \right.\\
&&\left.-\E  \left[G(X)  \left\{\phi_{+}(\varpi(X)-\hat{\pi}_f ^*(X)) - \phi_{+}(\varpi(X)-\pi^*_f(X))\right\} \right]  \right|\\
&\leq&\left |\mathbb{E}\big[G(X)\{\phi_{+}(\varpi(X)-\hat\pi^*(X))-\phi_{+}(\varpi(X)-\hat\pi_f^*(X))\}\big]\right|\\
&&+\left|\mathbb{E}\left[G(X)\{\phi_{+}(\varpi(X)-\pi^*(X))-\phi_{+}(\varpi(X)-\pi_f^*(X))\}\right]\right|\\
&\leq& O_{\mathbb P}\big(n^{-1/2}\big).
\end{eqnarray*} \QEDB

\begin{lemma}\label{lem:cfempr2term}
    Suppose that Assumptions~\ref{assum:c311},~\ref{assum:margincondition1} and Conditions (a), (b), and (c) in Theorem~\ref{thm:convergence_margin} hold, and $p_1(x)-p_0(x)>0$ for all  $x\in\mathcal{X}$. Then
    \begin{eqnarray}
(\mathbb{P}_n-\mathbb{P})\left[\varpi(X)\left\{\mathbf{1}(\hat{\rho}_U(X)>\hat{\eta}^S)-\mathbf{1}(\hat{\rho}_U(X)>\eta^S)\right\}\left\{\psi_{p_1}-\psi_{p_0}\right\}\right]&=&o_\Pb(n^{-1/2}),\label{eqn:cfempr2term1}\\
(\mathbb{P}_n-\mathbb{P})\left[\left\{1-\varpi(X)\right\}\left\{\mathbf{1}(\hat{\rho}_L(X)>\hat{\eta}^S)-\mathbf{1}(\hat{\rho}_L(X)>\eta^S)\right\}\left\{\psi_{p_1}-\psi_{p_0}\right\}\right]&=&o_\Pb(n^{-1/2}).\label{eqn:cfempr2term2}
    \end{eqnarray}
\end{lemma}

\noindent {\it Proof of Lemma~\ref{lem:cfempr2term}. } We only prove \eqref{eqn:cfempr2term1}, the other one follows similarly. 

For ease of exposition, denote the independent sample used for estimating the nuisance functions by $\mathcal{I}_n$, and let \begin{eqnarray*}
    R_n&=&(\mathbb{P}_n-\mathbb{P})\left[\varpi(X)\left\{\mathbf{1}(\hat{\rho}_U(X)>\hat{\eta}^S)-\mathbf{1}(\hat{\rho}_U(X)>\eta^S)\right\}\left\{\psi_{p_1}-\psi_{p_0}\right\}\right].
\end{eqnarray*}By the boundedness of $e(X)$, $p_1(X)$, $p_0(X)$ and $\varpi(X)$, there exists a sufficiently large constant $C_1<\infty$ such that
\begin{eqnarray*}
&& \| \varpi(X)\left\{\mathbf{1}(\hat{\rho}_U(X)>\hat{\eta}^S)-\mathbf{1}(\hat{\rho}_U(X)>\eta^S)\right\}\left\{\psi_{p_1}-\psi_{p_0}\right\}\|_2^2  \\
&\leq& C_1\mathbb{E} \left( \left| \mathbf{1}(\hat{\rho}_U(X)>\hat{\eta}^S)-\mathbf{1}(\hat{\rho}_U(X)>\eta^S) \right| \right)\\
&\leq & C_1\mathbb{P} \left( |\hat{\rho}_U(X) - \eta^S|\leq |\hat{\eta}^S-\eta^S| \right)\\
&\leq& C_1\mathbb{P}\left(|\rho_U(X)-\eta^S|\leq |\hat{\eta}^S-\eta^S|+|\hat{\rho}_U(X)-\rho_U(X)|\right),
\end{eqnarray*}
where the second inequality follows from Lemma~\ref{lem:indicator_diff} and the last inequality follows
by the triangle inequality. Then, for any $t>0$,
\begin{eqnarray}
&&\| \varpi(X)\left\{\mathbf{1}(\hat{\rho}_U(X)>\hat{\eta}^S)-\mathbf{1}(\hat{\rho}_U(X)>\eta^S)\right\}\left\{\psi_{p_1}-\psi_{p_0}\right\}\|_2^2\nonumber  \\
&\leq&  C_1\mathbb{P} \left( |\rho_U(X) - \eta^S| \leq t \right) + C_1\mathbb{P} \left( |\hat{\eta}^S - \eta^S| + |\hat{\rho}_U(X)-\rho_U(X)| > t \right)\nonumber\\
&\leq & CC_1t^\beta+C_1\frac{\|\hat\rho_U-\rho_U\|_2+|\hat{\eta}^S-\eta^S|}{t}\nonumber,
\end{eqnarray} 
where the second inequality follows from Assumption~\ref{assum:margincondition1} and Markov's inequality. For any $\epsilon>0$, choose $t_\epsilon=\left(\frac{\epsilon}{CC_1}
\right)^{1/\beta}$, and define $Z_n^{(\epsilon)}=C_1\left\{\|\hat\rho_U-\rho_U\|+|\hat{\eta}^S-\eta^S|\right\}/t_{\epsilon}$ and $a=0$.
Then 
\begin{eqnarray*}
\| \varpi(X)\left\{\mathbf{1}(\hat{\rho}_U(X)>\hat{\eta}^S)-\mathbf{1}(\hat{\rho}_U(X)>\eta^S)\right\}\left\{\psi_{p_1}-\psi_{p_0}\right\}\|_2^2
&\leq&
CC_1t_\epsilon^\beta+C_1\frac{\|\hat\rho_U-\rho_U\|+|\hat{\eta}^S-\eta^S|}{t_\epsilon}\\
&=&a+\epsilon+Z_n^{(\epsilon)}.
\end{eqnarray*} 
Under Condition (b) in Theorem~\ref{thm:convergence_margin} and Proposition~\ref{lem:consistent}, $Z_n^{(\epsilon)}=o_{\mathbb P}(1)$.
From Lemma~\ref{lem:op_criterion}, we obtain that
\begin{eqnarray}
\| \varpi(X)\left\{\mathbf{1}(\hat{\rho}_U(X)>\hat{\eta}^S)-\mathbf{1}(\hat{\rho}_U(X)>\eta^S)\right\}\left\{\psi_{p_1}-\psi_{p_0}\right\}\|_2^2
&=&o_{\mathbb P}(1).\label{eqn:cfnorm}
\end{eqnarray} 

Note that the function class $\left\{ \mathbf{1}(\beta(\cdot)>\eta) : \eta\in \mathbb{R} \right\}$ is Donsker in $\eta$ for any fixed function $\beta:\mathcal{X}\to\mathbb{R}$ (see Example 2.5.4 in \citet{van1996weak}).
Since $\varpi(X)$ and $\psi_{p_1}-\psi_{p_0}$ are uniformly bounded, the class $\left\{ \varpi(\cdot)\mathbf{1}(\beta(\cdot)>\eta)\left\{\psi_{p_1}-\psi_{p_0}\right\} : \eta\in(-\infty,0]) \right\}$ is also Donsker
(see Example 2.10.10 in \citet{van1996weak}). Conditionally on the independent sample $\mathcal{I}_n$, we can view $\hat{\rho}_U$ as fixed and $\left\{ \varpi(\cdot)\mathbf{1}(\hat{\rho}_U(\cdot)>\eta)\left\{\psi_{p_1}-\psi_{p_0}\right\} : \eta\in (-\infty,0] \right\}$ is Donsker. Combining \eqref{eqn:cfnorm} and applying Lemma 19.24 of \citet{van2000asymptotic} conditionally on $\mathcal{I}_n$ yields \begin{eqnarray*}
\mathbb{P}\left(\sqrt{n}R_n>\epsilon\mid \mathcal{I}_n\right)\to 0
\end{eqnarray*}
for every $\epsilon>0$. By the law of iterated expectations,
\begin{eqnarray*}
\mathbb{P}\left\{\sqrt{n}R_n>\epsilon \right\}=\mathbb{E}\left[\mathbb{P}\left(\sqrt{n}R_n>\epsilon \mid \mathcal{I}_n\right)\right]\to 0.
\end{eqnarray*}
Therefore,
$R_n=o_\mathbb{P}(n^{-1/2})$. \QEDB

\begin{lemma}\label{lem:cfempr2}Suppose that Assumptions~\ref{assum:c311},~\ref{assum:margincondition1} and Conditions (a), (b), and (c) in Theorem~\ref{thm:convergence_margin} hold, and $p_1(x)-p_0(x)>0$ for all  $x\in\mathcal{X}$. Then
    \begin{eqnarray*}
        (\mathbb{P}_n-\mathbb{P})\left[\left\{\hat{\pi}_f^*(X)-\pi_f^*(X)\right\}\left\{\psi_{p_1}-\psi_{p_0}\right\}\right]&=&o_{\mathbb{P}}(n^{-1/2}).
    \end{eqnarray*}
\end{lemma}
{\it Proof of Lemma~\ref{lem:cfempr2}. }
With the definition of $\pi^*_f$ and $\hat{\pi}^*_f$, we can write \begin{eqnarray}
    &&(\mathbb{P}_n-\mathbb{P})\left[\left\{\hat{\pi}_f^*(X)-\pi_f^*(X)\right\}\left\{\psi_{p_1}-\psi_{p_0}\right\}\right]\nonumber\\
    &=&(\mathbb{P}_n-\mathbb{P})\left[\left\{\varpi(X)\mathbf{1}(\hat{\rho}_U(X)>\hat{\eta}^S)+(1-\varpi(X))\mathbf{1}(\hat{\rho}_L(X)>\hat{\eta}^S)\right\}\left\{\psi_{p_1}-\psi_{p_0}\right\}\right]\nonumber\\
    &&-(\mathbb{P}_n-\mathbb{P})\left[\left\{\varpi(X)\mathbf{1}(\rho_U(X)>\eta^S)+(1-\varpi(X))\mathbf{1}(\rho_L(X)>\eta^S)\right\}\left\{\psi_{p_1}-\psi_{p_0}\right\}\right]\nonumber\\
    &=&R_1+R_2+R_3+R_4,\label{eqn:cfempr2}
\end{eqnarray}
where \begin{eqnarray*}
    R_1&=&(\mathbb{P}_n-\mathbb{P})\left[\varpi(X)\left\{\mathbf{1}(\hat{\rho}_U(X)>\eta^S)-\mathbf{1}(\rho_U(X)>\eta^S)\right\}\left\{\psi_{p_1}-\psi_{p_0}\right\}\right],\\
    R_2&=&(\mathbb{P}_n-\mathbb{P})\left[\left\{1-\varpi(X)\right\}\left\{\mathbf{1}(\hat{\rho}_L(X)>\eta^S)-\mathbf{1}(\rho_L(X)>\eta^S)\right\}\left\{\psi_{p_1}-\psi_{p_0}\right\}\right],\\
    R_3&=&(\mathbb{P}_n-\mathbb{P})\left[\varpi(X)\left\{\mathbf{1}(\hat{\rho}_U(X)>\hat{\eta}^S)-\mathbf{1}(\hat{\rho}_U(X)>\eta^S)\right\}\left\{\psi_{p_1}-\psi_{p_0}\right\}\right],\\
    R_4&=&(\mathbb{P}_n-\mathbb{P})\left[\left\{1-\varpi(X)\right\}\left\{\mathbf{1}(\hat{\rho}_L(X)>\hat{\eta}^S)-\mathbf{1}(\hat{\rho}_L(X)>\eta^S)\right\}\left\{\psi_{p_1}-\psi_{p_0}\right\}\right].
\end{eqnarray*}

The first two terms $R_1$ and $R_2$ are $o_\mathbb{P}(n^{-1/2})$ from \eqref{eqn:auxconsisemp1_U_etas} and \eqref{eqn:auxconsisemp1_L_etas}. The last two terms $R_3$ and $R_4$ are $o_\mathbb{P}(n^{-1/2})$ from \eqref{eqn:cfempr2term1} and \eqref{eqn:cfempr2term2}. Plugging $R_1=o_\mathbb{P}(n^{-1/2})$, $R_2=o_\mathbb{P}(n^{-1/2})$, $R_3=o_\mathbb{P}(n^{-1/2})$ and $R_4=o_\mathbb{P}(n^{-1/2})$ into \eqref{eqn:cfempr2} completes the proof. \QEDB

\begin{lemma}\label{lem:posterm}
Suppose that Assumptions~\ref{assum:rates}(a),~\ref{assum:c311}, and~\ref{assum:margincondition1}, and Conditions (a), (b), and (c) in Theorem~\ref{thm:convergence_margin} hold, and $p_1(x)-p_0(x)>0$ for all  $x\in\mathcal{X}$. If $\eta^S <0$, then
    \begin{eqnarray}
  \nonumber  && \mathbb{P}\left[\left\{\hat{\pi}_f^*(X)-\pi_f^*(X)\right\}\left\{p_1(X)-p_0(X)\right\}\right]\\
  \label{eqn:posterm}  
  &=&(\mathbb{P}_n-\mathbb{P})\left[ \{\varpi(X)-\pi_f^*(X)\}\left(\psi_{p_1}-\psi_{p_0}\right)\right]+o_\mathbb{P}(n^{-1/2}).
    \end{eqnarray}

\end{lemma}
\noindent {\it Proof of Lemma~\ref{lem:posterm}. } 
We have
 \begin{eqnarray*}
&&\nonumber\mathbb{P}_n\left\{\hat{\pi}_f^*(X)\left(\hat{\psi}_{p_1}-\hat{\psi}_{p_0}\right)\right\}-\mathbb{P}\left[\pi_f^*(X)\left\{p_1(X)-p_0(X)\right\}\right]\\
&=&(\mathbb{P}_n-\mathbb{P})\left[\hat{\pi}_f^*(X)\left\{\hat{\psi}_{p_1}-\hat{\psi}_{p_0}\right\}-\pi_f^*(X)\left\{\psi_{p_1}-\psi_{p_0}\right\}\right]+(\mathbb{P}_n-\mathbb{P})\left[\pi_f^*(X)\left\{\psi_{p_1}-\psi_{p_0}\right\}\right]\\
&&+\Pb\left\{\hat{\pi}_f^*(X)\left(\hat{\psi}_{p_1}-\hat{\psi}_{p_0}\right)\right\}-\mathbb{P}\left[\pi_f^*(X)\left\{p_1(X)-p_0(X)\right\}\right]\\
&=&(\mathbb{P}_n-\mathbb{P})\left[\hat{\pi}_f^*(X)\left\{\hat{\psi}_{p_1}-\hat{\psi}_{p_0}\right\}-\pi_f^*(X)\left\{\psi_{p_1}-\psi_{p_0}\right\}\right]+(\mathbb{P}_n-\mathbb{P})\left[\pi_f^*(X)\left\{\psi_{p_1}-\psi_{p_0}\right\}\right]\\
&&+\mathbb{P}\left[\hat{\pi}_f^*(X)\left\{\hat{\psi}_{p_1}-\hat{\psi}_{p_0}-p_1(X)+p_0(X)\right\}\right]+\mathbb{P}\left[\left\{\hat{\pi}_f^*(X)-\pi_f^*(X)\right\}\left\{p_1(X)-p_0(X)\right\}\right].
\end{eqnarray*}
Therefore, we obtain 
\begin{eqnarray}
\nonumber&&\mathbb{P}\left[\left\{\hat{\pi}_f^*(X)-\pi_f^*(X)\right\}\left\{p_1(X)-p_0(X)\right\}\right]\\
\label{eqn:proofthm5-lem1}&=&R_1-R_2-R_3-(\mathbb{P}_n-\mathbb{P})\left[\pi_f^*(X)\left\{\psi_{p_1}-\psi_{p_0}\right\}\right],
\end{eqnarray}
where
\begin{eqnarray*}
R_1&=&\mathbb{P}_n\left[\hat{\pi}_f^*(X)\left\{\hat{\psi}_{p_1}-\hat{\psi}_{p_0}\right\}\right]-\mathbb{P}\left[\pi_f^*(X)\left\{p_1(X)-p_0(X)\right\}\right],\\
R_2&=&(\mathbb{P}_n-\mathbb{P})\left[\hat{\pi}_f^*(X)\left\{\hat{\psi}_{p_1}-\hat{\psi}_{p_0}\right\}-\pi_f^*(X)\left\{\psi_{p_1}-\psi_{p_0}\right\}\right],\\
R_3&=&\mathbb{P}\left[\hat{\pi}_f^*(X)\left\{\hat{\psi}_{p_1}-\hat{\psi}_{p_0}-p_1(X)+p_0(X)\right\}\right].
\end{eqnarray*}
For the first term $R_1$, we have  $\hat{\eta}^E=\hat{\eta}^S<0$ with probability tending to one from Proposition~\ref{lem:consistent}. Therefore, we have 
\begin{eqnarray*}
\mathbb{P}_n\left[\hat{\pi}_f^*(X)\left\{\hat{\psi}_{p_1}-\hat{\psi}_{p_0}\right\}\right]&=&\hat{Q}(\hat{\eta}^S)\ =\ \hat{Q}(\hat{\eta}^E)+o_\mathbb{P}(n^{-1/2})\ =\ \hat{C}_{\textup{eif}}+o_\mathbb{P}(n^{-1/2}),
\end{eqnarray*}
where the last equality follows from Assumption~\ref{assum:c311}(b).
In addition, since $\eta^S<0$, we have
\begin{eqnarray*}
    \mathbb{P}\left[\pi_f^*(X)\left\{p_1(X)-p_0(X)\right\}\right]&=&\mathbb{P}\left[\pi^*(X)\left\{p_1(X)-p_0(X)\right\}\right]\ =\ \mathbb{E}\left[\varpi(X)\left\{p_1(X)-p_0(X)\right\}\right],
\end{eqnarray*} 
where the first equality follows from Lemma~\ref{lem:pi-diff} and the last equality follows from \eqref{eqn:proof_ext_x+2_new}.
As a result, from Lemma~\ref{lem:effkappa}, we obtain 
\begin{eqnarray}
\label{eqn:proofthm5-lem1-R1}
R_1&=&(\mathbb{P}_n-\mathbb{P})\{\varpi(X)\left(\psi_{p_1}-\psi_{p_0}\right)\}+o_{\mathbb{P}}(n^{-1/2}).
\end{eqnarray}

For the second term $R_2$, we can write \begin{eqnarray}
    |R_2|
    &=&\left|(\mathbb{P}_n-\mathbb{P})\left[\hat{\pi}^*_f(X)\left\{\hat{\psi}_{p_1}-\hat{\psi}_{p_0}-\psi_{p_1}+\psi_{p_0}\right\}\right]+(\mathbb{P}_n-\mathbb{P})\left[\left\{\hat{\pi}_f^*(X)-\pi_f^*(X)\right\}\left\{\psi_{p_1}-\psi_{p_0}\right\}\right]\right|\nonumber\\
    &\leq &T_1+T_2,\label{eqn:cfemp}
\end{eqnarray}
where \begin{eqnarray*}
    T_1&=&\left|(\mathbb{P}_n-\mathbb{P})\left\{\hat{\psi}_{p_1}-\hat{\psi}_{p_0}-\psi_{p_1}+\psi_{p_0}\right\}\right|,\\
    T_2&=&\left|(\mathbb{P}_n-\mathbb{P})\left[\left\{\hat{\pi}_f^*(X)-\pi_f^*(X)\right\}\left\{\psi_{p_1}-\psi_{p_0}\right\}\right]\right|.
\end{eqnarray*}
From Lemma~2 of \citet{kennedy2020sharp}, we have $T_1=o_\mathbb{P}(n^{-1/2})$. From Lemma~\ref{lem:cfempr2}, we obtain $T_2=o_\mathbb{P}(n^{-1/2})$. Plugging $T_1=o_\mathbb{P}(n^{-1/2})$ and $T_2=o_\mathbb{P}(n^{-1/2})$ into \eqref{eqn:cfemp} yields 
\begin{eqnarray}
\label{eqn:proofthm5-lem1-R2} R_2&=& o_\mathbb{P}(n^{-1/2}).
\end{eqnarray}

For the third term $R_3$, from~\eqref{eq:bias_pz} and the Cauchy-Schwartz inequality, we have 
\begin{eqnarray}
\label{eqn:proofthm5-lem1-R3}   R_3  & = & O_{\mathbb{P}}\left(\left\|\hat{\lambda}_0-\lambda_0\right\|_2\left\{\left\|\hat{p}_1-p_1\right\|_2+\left\|\hat{p}_0-p_0\right\|_2\right\}\right) =\ o_{\mathbb{P}}(n^{-1/2}).
\end{eqnarray} 

Substituting~\eqref{eqn:proofthm5-lem1-R1},~\eqref{eqn:proofthm5-lem1-R2} and \eqref{eqn:proofthm5-lem1-R3} into~\eqref{eqn:proofthm5-lem1} yields \eqref{eqn:posterm}. \QEDB

\begin{lemma}\label{lem::relaxkey}
For any $x$ such that $p_1(x)-p_0(x)>0$, the following statements hold.
\begin{enumerate}[(a)]
    \item  If  $\phi_{+}(\varpi(x) - \pi^*_{f}(x))\neq \phi_{+}(\varpi(x) - \hat{\pi}^*_{f}(x))$,
    then 
    \begin{eqnarray}\label{eqn:relaxresult1}
       |\rho_{U }(x) - \eta ^S | &\leq& |\rho_{U }(x) - \eta ^S  - \hat{\rho}_{U }(x) + \hat{\eta} ^S|.
    \end{eqnarray}
    
    \item If  $\phi_{-}(\varpi(x) - \pi^*_{f}(x))\neq \phi_{-}(\varpi(x) - \hat{\pi}^*_{f}(x))$
    then 
    \begin{eqnarray}\label{eqn:relaxresult2}
           |\rho_{L }(x) - \eta ^S | &\leq& |\rho_{L }(x) - \eta ^S   - \hat{\rho}_{L }(x) + \hat{\eta} ^S |.
    \end{eqnarray}
\end{enumerate}
\end{lemma}
\noindent {\it Proof of Lemma~\ref{lem::relaxkey}. }
We prove \eqref{eqn:relaxresult1} and \eqref{eqn:relaxresult2} for a fixed $x\in\mathcal X$, and thus omit $x$ in the argument of functions. 

We first prove Lemma~\ref{lem::relaxkey}(a).
Observe that  both $\phi_{+}(\varpi - \pi^*_{f})$ and $\phi_{+}(\varpi - \hat{\pi}^*_{f})$ take values in the set $\{0, \varpi\}$.
We consider two cases depending on the value of $\phi_{+}(\varpi - \pi^*_{f})$.
\begin{itemize}
    \item $\phi_{+}(\varpi - \pi^*_{f})= 0$. We have $\phi_{+}(\varpi - \hat{\pi}^*_{f}) = \varpi$. Therefore, $\pi^*_{f} \geq \varpi$ and $\hat{\pi}^*_{f} = 0$. By the definitions of $\pi^*_{f}$ and $\hat{\pi}^*_{f} $, we obtain $\rho_{U } - \eta ^S > 0$ and $\hat{\rho}_{U } - \hat{\eta} ^S < 0$. As a result,~\eqref{eqn:relaxresult1} holds.

    \item  $\phi_{+}(\varpi - \pi^*_{f}) =  \varpi$. We have $\phi_{+}(\varpi - \hat{\pi}^*_{f}) = 0$. Therefore,  $\pi^*_{f} = 0$ and $\hat{\pi}^*_{f} \geq \varpi$. By definition, we obtain $\rho_{U } - \eta ^S < 0$ and $\hat{\rho}_{U } - \hat{\eta} ^S > 0$. Again,~\eqref{eqn:relaxresult1} holds.
\end{itemize}

We then prove  Lemma~\ref{lem::relaxkey}(b).
Observe that  both $\phi_{-}(\varpi - \pi^*_{f})$ and $\phi_{-}(\varpi - \hat{\pi}^*_{f})$ take values in the set $\{0, \varpi-1\}$.
We consider two cases depending on the value of $\phi_{-}(\varpi - \pi^*_{f})$.
\begin{itemize}
    \item $\phi_{-}(\varpi - \pi^*_{f}) = 0$. We have $\phi_{-}(\varpi - \hat{\pi}^*_{f}) = \varpi-1$. Therefore, $\pi^*_{f} \leq \varpi$ and $\hat{\pi}^*_{f} = 1$. By the definitions of $\pi^*_{f}$ and $\hat{\pi}^*_{f} $, we obtain $\rho_{L} - \eta ^S  < 0$ and $\hat{\rho}_{L } - \hat{\eta} ^S  > 0$. As a result,~\eqref{eqn:relaxresult2} holds.

    \item  $\phi_{-}(\varpi - \pi^*_{f}) =  \varpi(x)-1$. We have $\phi_{-}(\varpi - \hat{\pi}^*_{f}) = 0$. Therefore,  $\pi^*_{f} = 1$ and $\hat{\pi}^*_{f} \leq \varpi$. By definition, we obtain $\rho_{L } - \eta ^S   > 0$ and $\hat{\rho}_{L } - \hat{\eta} ^S  < 0$. Again,~\eqref{eqn:relaxresult2} holds.
\end{itemize} \QEDB

\subsubsection{Rate bound on the excess worst-case regret}
\begin{eqnarray}
\nonumber&&p_0\{R_{\sup}(\hat{\pi} ^*,\varpi)-R_{\sup}(\pi ^*,\varpi)\}\\
\nonumber&=&\mathbb{E} \left[ \left\{\varpi(X)-\hat{\pi} ^*(X)\right\} \left\{A_{U }(X)\bone(\hat{\pi} ^*(X)<\varpi(X))+A_{L }(X)\bone(\hat{\pi} ^*(X)> \varpi(X))\right\} \right]\\
\nonumber&&- \mathbb{E} \left[ \left\{\varpi(X)-\pi ^*(X)\right\} \left\{A_{U }(X)\bone(\pi ^*(X)<\varpi(X))+A_{L }(X)\bone(\pi ^*(X)> \varpi(X))\right\} \right]\\
\nonumber&=& \mathbb{E}\left[ A_{U }(X)\left\{\phi_{+}(\varpi(X)-\hat{\pi} ^*(X))-\phi_{+}(\varpi(X)-\pi ^*(X))\right\}\right]\\
\nonumber&&+\mathbb{E}\left[ A_{L }(X)\left\{ \phi_{-}(\varpi(X)-\hat{\pi} ^*(X))-\phi_{-}(\varpi(X)-\pi ^*(X))\right\}\right]\\
\nonumber&=& \E\left[ \left\{A_{U }(X)-\eta ^S\left(p_1(X)-p_0(X)\right)\right\}\left\{ \phi_{+}(\varpi(X)-\hat{\pi} ^*(X))-\phi_{+}(\varpi(X)-\pi ^*(X))\right\}\right]\\
\nonumber&&+\E\left[ \left\{A_{L }(X)-\eta ^S\left(p_1(X)-p_0(X)\right)\right\}\left\{\phi_{-}(\varpi(X)-\hat{\pi} ^*(X))-\phi_{-}(\varpi(X)-\pi ^*(X))\right\}\right]\\
\nonumber&&+\mathbb{E}\Big[ \eta ^S\left(p_1(X)-p_0(X)\right)\big\{(\varpi(X)-\hat{\pi} ^*(X))\bone(\hat{\pi}^*(X)\ne \varpi(X))-(\varpi(X)-\pi ^*(X))\bone(\pi^*(X)\ne \varpi(X))\big\}\Big]\\
&=&R_1+R_2+R_3,\label{eqn:res}
\end{eqnarray}
where 
\begin{eqnarray*}
R_1 &=&  \E\left[ \left\{A_{U }(X)-\eta ^S\left(p_1(X)-p_0(X)\right)\right\}\left\{ \phi_{+}(\varpi(X)-\hat{\pi} ^*(X))-\phi_{+}(\varpi(X)-\pi ^*(X))\right\}\right], \\
R_2&=&\E\left[ \left\{A_{L }(X)-\eta ^S\left(p_1(X)-p_0(X)\right)\right\}\left\{\phi_{-}(\varpi(X)-\hat{\pi} ^*(X))-\phi_{-}(\varpi(X)-\pi ^*(X))\right\}\right],\\
R_3&=&\eta ^S\,\mathbb{E}\Big[ \left\{\pi ^*(X)-\hat{\pi} ^*(X)\right\}\left(p_1(X)-p_0(X)\right)\Big].
\end{eqnarray*}

For the first term $R_1$,  from Lemma~\ref{lem:replace_pi_by_pif}, we have \begin{eqnarray*}
   R_1&=&T_1+O_\mathbb{P}(n^{-1/2}),
\end{eqnarray*}
where
\begin{eqnarray*}
    T_1&=&\E\left[ \left\{A_{U }(X)-\eta ^S\left(p_1(X)-p_0(X)\right)\right\}\left\{ \phi_{+}(\varpi(X)-\hat{\pi}_f ^*(X))-\phi_{+}(\varpi(X)-\pi_f^*(X))\right\}\right].
\end{eqnarray*}
For the term $T_1$, we have
\begin{eqnarray*}
|T_1|
\nonumber &\leq&\E\left[ \left| p_1(X)-p_0(X)\right| \left| \rho_{U }(X)-\eta ^S\right| \bone\left\{ \phi_{+}(\varpi(X)-\hat{\pi}_f ^*(X))\neq \phi_{+}(\varpi(X)-\pi_f^*(X))\right\}\right]\\
&\leq&\E\left[ \left| \rho_{U }(X)-\eta ^S\right| \bone\left\{  |\rho_{U }(x) - \eta ^S | \leq |\rho_{U }(x) - \eta ^S  - \hat{\rho}_{U }(x) + \hat{\eta} ^S|\right\}\right],
\end{eqnarray*}
where the last inequality follows from $p_1(X)-p_0(X)\leq 1$ and  Lemma~\ref{lem::relaxkey}.  Under Assumption~\ref{assum:margincondition1}(a), we further have
\begin{eqnarray*}
|T_1|
\nonumber&\leq& \E\left[ \left|\rho_{U }(x) - \eta ^S  - \hat{\rho}_{U }(x) + \hat{\eta} ^S\right| \bone\left\{  |\rho_{U }(x) - \eta ^S | \leq |\rho_{U }(x) - \eta ^S  - \hat{\rho}_{U }(x) + \hat{\eta} ^S|\right\}\right]\\
\nonumber&\leq&  ||\rho_{U }(x) - \eta ^S  - \hat{\rho}_{U }(x) + \hat{\eta} ^S||_{\infty}\Pb( |\rho_{U }(x) - \eta ^S | \leq |\rho_{U }(x) - \eta ^S  - \hat{\rho}_{U }(x) + \hat{\eta} ^S|)\\
\nonumber&\leq& C||\rho_{U }(x) - \eta ^S  - \hat{\rho}_{U }(x) + \hat{\eta} ^S||_{\infty}^{1+\beta}\\
&\leq& C(\|\hat{\rho}_{U }-\rho_{U }\|_{\infty}+|\hat{\eta} ^S-\eta ^S|)^{1+\beta}. 
\end{eqnarray*}
As a result, we obtain 
\begin{eqnarray}
\label{eqn:resr1} |R_1|&\leq&C(\|\hat{\rho}_{U }-\rho_{U }\|_{\infty}+|\hat{\eta} ^S-\eta ^S|)^{1+\beta}+O_\mathbb{P}(n^{-1/2}).
\end{eqnarray}

For the second term $R_2$,  from Lemma~\ref{lem:replace_pi_by_pif}, we have
\begin{eqnarray*}
\nonumber R_2
&=&T_2+O_\mathbb{P}(n^{-1/2}),
\end{eqnarray*}
where \begin{eqnarray*}
    T_2&=&\E\left[ \left\{A_{L }(X)-\eta ^S\left(p_1(X)-p_0(X)\right)\right\}\left\{ \phi_{-}(\varpi(X)-\hat{\pi}_f ^*(X))-\phi_{-}(\varpi(X)-\pi_f^*(X))\right\}\right].
\end{eqnarray*}
For the term $T_2$, we have
\begin{eqnarray*}
|T_2|\nonumber &\leq&\E\left[ \left| p_1(X)-p_0(X)\right| \left| \rho_{L }(X)-\eta ^S\right| \bone\left\{ \phi_{-}(\varpi(X)-\hat{\pi}_f ^*(X))\neq \phi_{-}(\varpi(X)-\pi_f^*(X))\right\}\right]\\
&\leq&\E\left[ \left| \rho_{L }(X)-\eta ^S\right| \bone\left\{  |\rho_{L}(x) - \eta ^S | \leq |\rho_{L}(x) - \eta ^S  - \hat{\rho}_{L}(x) + \hat{\eta} ^S|\right\}\right],
\end{eqnarray*}
where the last inequality follows from $p_1(X)-p_0(X)\leq 1$ and  Lemma~\ref{lem::relaxkey}.  Under Assumption~\ref{assum:margincondition1}(a), we further have
\begin{eqnarray*}
|T_2|\nonumber&\leq& \E\left[ \left|\rho_{L }(x) - \eta ^S  - \hat{\rho}_{L }(x) + \hat{\eta} ^S\right| \bone\left\{  |\rho_{L }(x) - \eta ^S | \leq |\rho_{L }(x) - \eta ^S  - \hat{\rho}_{L }(x) + \hat{\eta} ^S|\right\}\right]\\
\nonumber&\leq&  ||\rho_{L }(x) - \eta ^S  - \hat{\rho}_{L }(x) + \hat{\eta} ^S||_{\infty}\Pb( |\rho_{L}(x) - \eta ^S | \leq |\rho_{L }(x) - \eta ^S  - \hat{\rho}_{L }(x) + \hat{\eta} ^S|)\\
\nonumber&\leq& C||\rho_{L}(x) - \eta ^S  - \hat{\rho}_{L }(x) + \hat{\eta} ^S||_{\infty}^{1+\beta}\\
&\leq& C(\|\hat{\rho}_{L }-\rho_{L }\|_{\infty}+|\hat{\eta} ^S-\eta ^S|)^{1+\beta}.\
\end{eqnarray*}
As a result, we obtain 
\begin{eqnarray}
\label{eqn:resr2} |R_2|&\leq&C(\|\hat{\rho}_{L }-\rho_{L }\|_{\infty}+|\hat{\eta} ^S-\eta ^S|)^{1+\beta}+O_\mathbb{P}(n^{-1/2}).
\end{eqnarray}

For the third term $R_{3}$,
we have
\begin{eqnarray}\label{eqn:resr3}
 R_{3}&=&\nonumber\eta^S\mathbb{E}\left[ \left\{\pi_f ^*(X)-\hat{\pi}_f ^*(X)\right\}\left\{p_1(X)-p_0(X)\right\}\right]\nonumber 
\\&=&-\eta ^S(\mathbb{P}_n-\mathbb{P})\left[\left\{\varpi(X)-\pi^*_f(X)\right\}(\psi_{p_1}-\psi_{p_0})\right]+o_\mathbb{P}(n^{-1/2})\nonumber
\\&=&O_\mathbb{P}(n^{-1/2}),
\end{eqnarray}
where the second equality follows from Lemma~\ref{lem:posterm}.

Substituting \eqref{eqn:resr1}-\eqref{eqn:resr3} into \eqref{eqn:res} completes the proof. \QEDB

\subsection{Proof of Theorem~\ref{thm:pi_opt_knownm} }
\begin{lemma}
    \label{lem:pointwise_oracle}
For any fix $x$ such that $p_1(x)-p_0(x)>0$, define
\begin{eqnarray*}
g_o(\pi) &=& \{\pi - \varpi(x)\} A(x)-\eta_o^S \pi\left\{p_1(x)-p_0(x)\right\},\quad \pi\in[0,1].
\end{eqnarray*}
Then one maximizer of $g_o(\pi)$ over $\pi\in[0,1]$ is given by
\begin{eqnarray*}
\pi^*_o(x)
=
\begin{cases}
1, & \rho(x)>\eta^S_o,\\
\alpha_o, & \rho(x)= \eta^S_o,\\
0, & \rho(x)<\eta^S_o.
\end{cases}
\end{eqnarray*}
\end{lemma}

\noindent {\it Proof of Lemma~\ref{lem:pointwise_oracle}.} 
We derive the maximizer for a fixed $x\in\mathcal X$ and thus omit $x$ in the argument of functions.

Because $g_o(\pi)$ is a linear function in $\pi$, its derivatives is
\begin{eqnarray*}
g_o'(\pi) &=&
(p_1-p_0)(\rho - \eta_o^S).
\end{eqnarray*}
We then analyze $g_o(\pi)$ case by case. 
\begin{itemize}
\item $\rho > \eta_o^S$. We have derivatives to be positive on $[0,1]$. Therefore, $g_o(\pi)$ achieves its maximum at $\pi=1$.
\item $\rho < \eta_o^S$. We have derivatives to be negative on $[0,1]$. Therefore, $g_o(\pi)$ achieves its maximum at $\pi=0$.
\item $ \rho = \eta_o^S$. We have derivatives to be zero on $[0,1]$. Therefore, $g_o(\pi)$ is constant on $[0, 1]$. As a result, $g_o(\pi)$ achieves its maximum at any value in $[0, 1]$.
\end{itemize}

From the definition of $\alpha_o$, we have $\alpha_o\in [0,1]$, therefore,  $\pi_o^*$ maximizes $g_o(\pi)$ in all cases. \QEDB

We then prove Theorem~\ref{thm:pi_opt_knownm}. We first show that $\pi_o^*(X)$ satisfies the constraint, i.e, \begin{eqnarray*}\mathbb{E}[\pi_o^*(X)\left\{p_1(X)-p_0(X)\right\}]\geq \mathbb{E}\left[\varpi(X)\left\{p_1(X)-p_0(X)\right\}\right].\end{eqnarray*}

\noindent
\textbf{Case 1: $\eta_o^S = 0$.} By definition, we have $\eta_o^E \geq 0$, which implies $Q_o(0) \geq  \mathbb{E}\left[\varpi(X)\left\{p_1(X)-p_0(X)\right\}\right]$. Therefore, we have
\begin{eqnarray*}
\mathbb{E}[\pi_o^*(X)\left\{p_1(X)-p_0(X)\right\}] 
&=& \mathbb{E}[\left\{p_1(X)-p_0(X)\right\} \bone(\rho(X) > 0)] \\
&=& Q_o(0)\\
&\geq& \mathbb{E}\left[\varpi(X)\left\{p_1(X)-p_0(X)\right\}\right].
\end{eqnarray*}

\noindent
\textbf{Case 2: $\eta_o^S < 0$.}  We have 
\begin{eqnarray}
\nonumber &&\mathbb{E}[\pi_o^*(X)\left\{p_1(X)-p_0(X)\right\}] \\
&=& \mathbb{E}[\left\{p_1(X)-p_0(X)\right\} \bone(\rho(X) > \eta_o^S)]
\nonumber + \mathbb{E}\left[\alpha_o \left\{p_1(X)-p_0(X)\right\} \bone(\rho(X) = \eta_o^S)\right] \\
\label{eqn:proof-constraint_knownm}&=& \mathbb{E}\left[\varpi(X)\left\{p_1(X)-p_0(X)\right\}\right].\end{eqnarray}
where the last equality follows from the definition of $\alpha_o$. 

Thus $\pi^*_o$ satisfies the constraint.

Next, we prove that  $R(\pi^*,\varpi) \leq R(\pi',\varpi)$ for any  $\pi'(X)$ satisfying the constraint $B(\pi')\geq B(\varpi)$.

 From Lemma~\ref{lem:pointwise_oracle}, we have $g_o(\pi_o^*(x))\geq g_o(\pi'(x))$ for all $x$, which implies $\E\{g_o(\pi_o^*(X))\}\geq \E\{g_o(\pi'(X))\}$.
As a result,
\begin{eqnarray}
   \nonumber  && R(\pi',\varpi)-R(\pi_o^*,\varpi)\\
   \nonumber  &=&\mathbb{E}\left[\left\{\pi_o^*(X) - \varpi(X)\right\} A(X) \right]-\mathbb{E}\left[\left\{\pi'(X) - \varpi(X)\right\} A(X) \right]\\
    \nonumber &\geq& \eta_o^S\mathbb{E}\left[\left\{\pi_o^*(X)-\varpi(X)\right\}\left\{p_1(X)-p_0(X)\right\}\right]-\eta_o^S\mathbb{E}\left[\left\{\pi'(X)-\varpi(X)\right\}\left\{p_1(X)-p_0(X)\right\}\right]\\\nonumber&=& \eta_o^S\mathbb{E}\left[\left\{\pi_o^*(X)-\pi'(X)\right\}\left\{p_1(X)-p_0(X)\right\}\right]\\
  \label{eqn::proofthm4_knownm}  &\geq & \eta_o^S\mathbb{E}\left[\left\{\pi_o^*(X)-\varpi(X)\right\}\left\{p_1(X)-p_0(X)\right\}\right],
\end{eqnarray}
where the last inequality follows from $\eta_o^S\leq 0$ and $B(\pi')\geq B(\varpi)$. If $\eta_o^S=0$, then the right hand side of~\eqref{eqn::proofthm4_knownm} equals zero; if $\eta_o^S<0$, the the right hand side of~\eqref{eqn::proofthm4_knownm} equals zero by~\eqref{eqn:proof-constraint_knownm}. In both cases, we obtain $R(\pi_o^*,\varpi)\leq R(\pi',\varpi)$. \QEDB

\subsection{Proof of Theorem~\ref{thm:dif_to_oracle}}
Analogously, denote 
\begin{eqnarray*}
\pi^*_{o,f}(X) &=& \bone(\rho_{}(X) > \eta_o ^S),
\end{eqnarray*}
which represents a simplified version of the oracle policy $\pi_o^*(X)$ that excludes the boundary cases  $\rho(X)=\eta_o^S$.

\begin{lemma}
\label{lem:pi-difforacle}
Under Assumption~\ref{assum:margincondition2}, $\mathbb{P}\big(\pi_o^*(X)\neq \pi_{o,f}^*(X)\big) = 0$.
\end{lemma}

\noindent {\it Proof of Lemma~\ref{lem:pi-difforacle}.} By the definitions of $\pi_{o,f}^*(\cdot)$ and $\pi_o^*(\cdot)$, the only case where $\pi_o^*(X)$ may differ from $\pi_{o,f}^*(X)$ are the boundary cases
$\rho(X)=\eta_o^S$. Therefore, we have 
\begin{eqnarray*}
\mathbb{P}\big(\pi_o^*(X)\neq \pi_{o,f}^*(X)\big)
&= & \mathbb{P}\big(\rho(X)=\eta_o^S\big)\ =\ 0,
\end{eqnarray*}
where the last equality follows from Assumption~\ref{assum:margincondition2}. \QEDB

\begin{lemma}\label{lem:replace_pio_by_pifo} Suppose that Assumption~\ref{assum:margincondition1} and \ref{assum:margincondition2} holds. Let $G(x)$ be any function satisfying $|G(x)|\leq M$ for some constant $M<\infty$ and all $x$.
    Then \begin{eqnarray*}
        \mathbb{E}\left[G(X)\phi_+(\pi_o^*(X)-\pi^*(X))\right]&=&\mathbb{E}\left[G(X)\phi_+(\pi_{o,f}^*(X)-\pi^*_{f}(X))\right],\\
        \mathbb{E}\left[G(X)\phi_-(\pi_o^*(X)-\pi^*(X))\right]&=&\mathbb{E}\left[G(X)\phi_-(\pi_{o,f}^*(X)-\pi^*_{f}(X))\right].
    \end{eqnarray*}
\end{lemma}
{\it Proof of Lemma~\ref{lem:replace_pio_by_pifo}. }
We only prove the first equality; the second equality follows similarly.

By the definition of $\phi_{+}(x)$, we have 
\begin{eqnarray*}
|\phi_{+}(x)-\phi_{+}(y)|
&=&  | \phi_{+}(x)-\phi_{+}(y)| \cdot \bone(x\neq y)
\ \le \  (|\phi_{+}(x)|+|\phi_{+}(y)|)\cdot \bone(x\neq y).
\label{eq:phi_diff_bound}
\end{eqnarray*}
Therefore, we have 
\begin{eqnarray*}
\nonumber &&\left |\mathbb{E}\big[G(X)\{\phi_+(\pi^*(X)-\pi^*_o(X))-\phi_+(\pi_f^*(X)-\pi^*_{o,f}(X))\big]\right|\\
&\le&
\mathbb{E}\big[|G(X)| \{|\phi_{+}(\pi^*(X)-\pi^*_o(X))| + |\phi_{+}(\pi_f^*(X)-\pi^*_{o,f}(X))|\} \cdot  \bone(\pi^*(X)-\pi^*_o(X)\neq \pi_f^*(X)-\pi^*_{o,f}(X)) \big]
\nonumber\\
&\le&
2\,\mathbb{E}\big[|G(X)|\,\{\bone(\pi^*(X)\neq \pi_f^*(X))+\bone(\pi^*_o(X)\neq \pi^*_{o,f}(X))\}\big]
\nonumber\\
\nonumber &\le&
2M\,\{\mathbb{P}\left(\pi^*(X)\neq \pi_f^*(X)\right)+\mathbb{P}\left(\pi_o^*(X)\neq \pi_{o,f}^*(X)\right)\}\\
&=& 0,
\label{eq:first_term_rate1}
\end{eqnarray*}
where the second inequality follows from $|\phi_{+}(\varpi(X)-\hat\pi^*(X))|\leq 1$, $|\phi_{+}(\varpi(X)-\hat\pi_f^*(X))|\leq 1$ and De Morgan's laws, and the last equality follows from Lemma~\ref{lem:pi-diff} and Lemma~\ref{lem:pi-difforacle}. \QEDB

\begin{lemma}\label{lem:thm5relaxcondition}
For any $x$ such that $p_1(x)-p_0(x)>0$, the following statements hold.
    \begin{enumerate}[(a)]
        \item If  $\pi^*_f(x)>\pi^*_{o,f}(x)$, then \begin{eqnarray}\label{eqn:thm5relaxcondition1}
            |\rho(x)-\eta_o^S|\leq |\rho(x)-\eta_o^S-\rho_U(x)+\eta^S|.
        \end{eqnarray}
        \item If  $\pi^*_f(x)<\pi^*_{o,f}(x)$, then \begin{eqnarray}\label{eqn:thm5relaxcondition2}
            |\rho(x)-\eta_o^S|\leq |\rho(x)-\eta_o^S-\rho_L(x)+\eta^S|.
        \end{eqnarray}
    \end{enumerate}
\end{lemma}
{\it Proof of Lemma~\ref{lem:thm5relaxcondition}. }
We prove \eqref{eqn:thm5relaxcondition1} and \eqref{eqn:thm5relaxcondition2} for a fixed $x\in\mathcal X$, and thus omit $x$ in the argument of functions. Observe that $\pi_f^*$ takes values in the set $\{0,\varpi,1\}$ and $\pi_{o,f}^*$ takes values in the set $\{0,1\}$. 

We first prove Lemma~\ref{lem:thm5relaxcondition}(a). If $\pi^*_f>\pi^*_{o,f}$, then $\pi^*_f\geq \varpi$ and $\pi_{o,f}^*=0$. By the definitions of $\pi^*_f$ and $\pi^*_{o,f}$, we obtain $\rho-\eta_o^S\leq 0$ and $\rho_U-\eta^S>0$. As a result, \eqref{eqn:thm5relaxcondition1} holds. 

We then prove Lemma~\ref{lem:thm5relaxcondition}(b). If $\pi^*_f<\pi^*_{o,f}$, then $\pi^*_f\leq \varpi$ and $\pi_{o,f}^*=1$. By the definitions of $\pi^*_f$ and $\pi^*_{o,f}$, we obtain $\rho-\eta_o^S> 0$ and $\rho_L-\eta^S\leq 0$. As a result, \eqref{eqn:thm5relaxcondition2} holds. 

\QEDB 

\noindent {\bf Regret bound relative to the oracle policy }
\begin{eqnarray}
    \nonumber p_0R(\pi^*,\pi_o^*)&=&\mathbb{E}\left[\{\pi_o^*(X)-\pi^*(X)\}A(X)\right]\\
    \nonumber &=&\mathbb{E}\left[\{\pi_o^*(X)-\pi^*(X)\}A(X)\bone\{\pi_o^*(X)>\pi^*(X)\}\right]\\
    \nonumber &&+\mathbb{E}\left[\{\pi_o^*(X)-\pi^*(X)\}A(X)\bone\{\pi_o^*(X)<\pi^*(X)\}\right]\\
    \label{eqn:thm5_decomp}&=& R_1+R_2+R_3,
\end{eqnarray}
where
\begin{eqnarray*}
    R_1&=&\mathbb{E}\left[\phi_+(\pi_o^*(X)-\pi^*(X))\{A(X)-\eta_o^S(p_1(X)-p_0(X))\}\right],\\
    R_2&=&\mathbb{E}\left[\phi_-(\pi_o^*(X)-\pi^*(X))\{A(X)-\eta_o^S(p_1(X)-p_0(X))\}\right],\\
    R_3&=&\eta_o^S\mathbb{E}\left[\{\pi_o^*(X)-\pi^*(X)\}(p_1(X)-p_0(X))\}\right].
\end{eqnarray*}
For $R_1$, from Lemma~\ref{lem:replace_pio_by_pifo}, we have \begin{eqnarray*}
    |R_1|&=& |\mathbb{E}\left[\phi_+(\pi_{o,f}^*(X)-\pi_{f}^*(X))\{A(X)-\eta_o^S(p_1(X)-p_0(X))\}\right]|\\
    &\leq &\mathbb{E}\left[\bone\{\pi_{o,f}^*(X)>\pi_{f}^*(X)\}|p_1(X)-p_0(X)||\rho(X)-\eta_o^S|\right]\\
    &\leq &\mathbb{E}\left[\bone\{|\rho(X)-\eta_o^S|\leq|\rho(X)-\eta_o^S-\rho_U(X)+\eta^S|\}|\rho(X)-\eta_o^S|\right],
\end{eqnarray*}
where the last inequality follows from $p_1(X)-p_0(X)\leq 1$ and Lemma~\ref{lem:thm5relaxcondition}. Then, we have \begin{eqnarray}
    \nonumber |R_1|
    &\leq& \mathbb{E}\left[\bone \{|\rho(X)-\eta_o^S|\leq|\rho(X)-\eta_o^S-\rho_U(X)+\eta^S|\}|\rho(X)-\eta_o^S-\rho_U(X)+\eta^S|\right]\\
    \nonumber&\leq & \|\rho(X)-\eta_o^S-\rho_U(X)+\eta^S\|_{\infty}\mathbb{P}\{|\rho(X)-\eta_o^S|\leq\|\rho(X)-\eta_o^S-\rho_U(X)+\eta^S\|_{\infty}\}\\
    \nonumber&\leq & C\|\rho(X)-\eta_o^S-\rho_U(X)+\eta^S\|_{\infty}^{1+b}\\
    \nonumber&\leq &C(\|\rho(X)-\rho_U(X)\|_{\infty}+|\eta_o^S-\eta^S|)^{1+b}\\
    & = & C\left(\left\|\frac{p_0(X)\{m_{11}(X)-U_{11}(X)\}}{p_1(X)-p_0(X)}\right\|_{\infty}+|\eta_o^S-\eta^S|\right)^{1+b},\label{eqn:thm5r1}
\end{eqnarray}
where the last equality follows from the definitions of $\rho(X)$ and $\rho_U(X)$.

For $R_2$, from Lemma~\ref{lem:replace_pio_by_pifo}, we have \begin{eqnarray*}
    |R_2|&=& |\mathbb{E}\left[\phi_-(\pi_{o,f}^*(X)-\pi_{f}^*(X))\{A(X)-\eta_o^S(p_1(X)-p_0(X))\}\right]|\\
    &\leq &\mathbb{E}\left[\bone\{\pi_{o,f}^*(X)<\pi_{f}^*(X)\}|p_1(X)-p_0(X)||\rho(X)-\eta_o^S|\right]\\
    &\leq &\mathbb{E}\left[\{|\rho(X)-\eta_o^S|\leq|\rho(X)-\eta_o^S-\rho_L(X)+\eta^S|\}|\rho(X)-\eta_o^S|\right],
\end{eqnarray*}
where the last inequality follows from $p_1(X)-p_0(X)\leq 1$ and Lemma~\ref{lem:thm5relaxcondition}.
Then, we have \begin{eqnarray}
    \nonumber |R_2|
    &\leq& \mathbb{E}\left[\bone \{|\rho(X)-\eta_o^S|\leq|\rho(X)-\eta_o^S-\rho_L(X)+\eta^S|\}|\rho(X)-\eta_o^S-\rho_L(X)+\eta^S|\right]\\
    \nonumber&\leq & \|\rho(X)-\eta_o^S-\rho_L(X)+\eta^S\|_{\infty}\mathbb{P}\{|\rho(X)-\eta_o^S|\leq\|\rho(X)-\eta_o^S-\rho_L(X)+\eta^S\|_{\infty}\}\\
    \nonumber&\leq & C\|\rho(X)-\eta_o^S-\rho_L(X)+\eta^S\|_{\infty}^{1+b}\\
    \nonumber&\leq & C(\|\rho(X)-\rho_L(X)\|_{\infty}+|\eta_o^S-\eta^S|)^{1+b}\\
    & = & C\left(\left\|\frac{p_0(X)\{m_{11}(X)-L_{11}(X)\}}{p_1(X)-p_0(X)}\right\|_{\infty}+|\eta_o^S-\eta^S|\right)^{1+b},\label{eqn:thm5r2}
\end{eqnarray}
where the last equality follows from the definitions of $\rho(X)$ and $\rho_L(X)$.

For term $R_3$, we consider three cases depending on the value of $\eta_o^S$ and $\eta^S$: \begin{itemize}
    \item If $\eta_o^S=0$, then $R_3=0$ obviously.
    \item If $\eta_o^S<0$ and $\eta^S<0$, then \begin{eqnarray*}
        R_3&=&\eta_o^S\{\mathbb{E}\left[\pi_o^*(X)(p_1(X)-p_0(X))\}\right]-\mathbb{E}\left[\pi^*(X)(p_1(X)-p_0(X))\}\right]\}\ =\ 0,
    \end{eqnarray*}
    where the last equality follows from\eqref{eqn:proof-constraint} and \eqref{eqn:proof-constraint_knownm}.
    \item If $\eta_o^S<0$ and $\eta^S=0$, then \begin{eqnarray*}
        R_3&=&\eta_o^S\{\mathbb{E}\left[\pi_o^*(X)(p_1(X)-p_0(X))\}\right]-\mathbb{E}\left[\pi^*(X)(p_1(X)-p_0(X))\}\right]\}\\& =& -\eta_o^S\{Q(0)-\mathbb{E}\left[\varpi(X)\{p_1(X)-p_0(X)\}\right]\},
    \end{eqnarray*}
    where the last equality follows from\eqref{eqn:proof-constraint_excess} and \eqref{eqn:proof-constraint_knownm}.
\end{itemize}
Thus \begin{eqnarray}\label{eqn:thm5r3}
    R_3= -\bone\{\sgn(\eta^S)\ne\sgn(\eta_o^S)\}\eta_o^S\{Q(0)-\mathbb{E}\left[\varpi(X)\{p_1(X)-p_0(X)\}\right]\}.
\end{eqnarray}
Substituting \eqref{eqn:thm5r1}-\eqref{eqn:thm5r3} into \eqref{eqn:thm5_decomp} yields the result. \QEDB

\end{document}